\documentclass[trsc,nonblindrev]{informs3} %

\renewcommand{\theARTICLETOP}{}
\renewcommand{\theRRHSecondLine}{}
\renewcommand{\theLRHSecondLine}{}

\usepackage{color-edits}
\addauthor[Chiwei]{cy}{cyan}

\OneAndAHalfSpacedXI %

\usepackage{natbib}
 \bibpunct[, ]{(}{)}{,}{a}{}{,}%
 \def\bibfont{\small}%

\TheoremsNumberedThrough     %

\EquationsNumberedThrough    %

\usepackage[utf8]{inputenc}
\usepackage{bm}
\usepackage{bbm}
\usepackage{bbold}
\usepackage{hyperref}
\usepackage{soul}
\usepackage{enumitem}
\usepackage{mathtools}

\usepackage[ruled,linesnumbered]{algorithm2e}
\usepackage{algpseudocode}

\usepackage{nicefrac}

\usepackage{cleveref}

\usepackage{array}
\usepackage{booktabs}
\usepackage{multirow}
\usepackage{makecell}

\usepackage{tikz}
\usepackage{tikz-network}
\usetikzlibrary{shapes, arrows, positioning }
\usepackage{subcaption}

\usepackage{color}

\allowdisplaybreaks

\begin{document}

\RUNAUTHOR{Gao et al.} %

\RUNTITLE{Markovian Traffic Equilibrium for Ride-Hailing}

\TITLE{A Markovian Traffic Equilibrium Model for Ride-Hailing}

\ARTICLEAUTHORS{%
\AUTHOR{Song Gao}
\AFF{Department of Civil and Environmental Engineering, University of Massachusetts Amherst, \EMAIL{sgao@umass.edu}}
\AUTHOR{Hanyu Cheng}
\AFF{Department of Industrial Engineering, Tsinghua University, Beijing, \EMAIL{hanyucheng2026@gmail.com}}
\AUTHOR{Chiwei Yan}
\AFF{Department of Industrial Engineering and Operations Research, University of California, Berkeley, \EMAIL{chiwei@berkeley.edu}} %
\AUTHOR{Guocheng Jiang}
\AFF{Department of Civil and Environmental Engineering, University of Massachusetts Amherst, \EMAIL{guochengjian@umass.edu}}
} %

\ABSTRACT{
We develop a Markovian traffic equilibrium model for ride-hailing in which vehicles, whether empty or hired, make sequential order-acceptance and link-choice decisions over a traffic network to maximize total discounted return in an infinite-horizon  semi-Markov decision process. The model endogenizes both competition among empty vehicles for passenger demand and traffic congestion arising from road usage at the link level. We characterize equilibrium as the solution to a fixed-point system, establish its existence, and develop relaxed fixed-point iteration algorithms for equilibrium computation, with convergence results for specialized network structures. Computational experiments on realistic networks demonstrate the model's practical value for transportation planning. Ablation analyses reveal that ignoring either traffic congestion or drivers' forward-looking behavior can lead to potentially substantial biases in policy evaluation.
}

\HISTORY{First version, 04/2026.}

\maketitle

\section{Introduction}

The expansion of ride-hailing services such as Uber and Lyft has continued steadily since their introduction more than a decade ago \citep{yan2020dynamic}. Although the pandemic caused a temporary decline, the growth momentum has since recovered. Uber's gross bookings, which measure the total dollar value of transactions on the platform and provide a more accurate indicator of service scale than revenue, reached \$162 billion in 2024, representing an 18\% increase from the previous year and more than an eightfold increase relative to 2016~\citep{UberRevenue2025}. Meanwhile, although autonomous vehicle (AV) technology has moved beyond the peak of its early hype, when major automakers claimed that fully autonomous fleets would be deployed by 2030, steady progress has nevertheless been made. Robotaxi services are now operating in dense cities with favorable weather conditions, such as San Francisco, CA and Phoenix, AZ through Waymo, and most recently in Austin, Texas through Tesla. Globally, China is leading the expansion of robotaxi services, with around 50 cities participating in pilot programs~\citep{ChinaRobotaxi}.

Although ride-hailing trips currently account for only a small share of overall traffic volume, typically in the single-digit percentages overall but higher in major urban cores~\citep{RHTrafficEffect2019}, it is conceivable that the continued expansion of both human-driven and AI-driven vehicles will eventually have a non-negligible impact on traffic conditions. A network equilibrium model for ride-hailing is therefore indispensable for addressing a wide range of planning, policy, and regulatory questions, and for enabling relevant agencies to proactively adopt effective strategies for the future transportation system.

Our equilibrium model is built on a semi-Markov decision process (SMDP) framework in which ride-hailing vehicles make sequential order-acceptance and link-choice decisions in a traffic network to maximize discounted total return over an infinite horizon, subject to perception errors. These sequential decisions are coupled through flow-dependent congestion and competition for passenger demand, and they equilibrate endogenously within the system. %
This unified framework treats hired and empty vehicles in a consistent manner: they correspond to different states within a single underlying stochastic decision process. This feature distinguishes our work from existing studies of traffic equilibrium models for ride-hailing, or its predecessor, taxi services, which extend conventional deterministic path-based traffic equilibrium models by incorporating endogenous variables such as vehicle and customer waiting times to balance hired and vacant flows, while still treating the decision processes of hired and vacant vehicles separately~\citep{WongWongYang01,WongWongYangWu08,BanDessoukyPangFan2019,RideSourcingEquilibrium2021}.

The remainder of the paper is organized as follows. Section~\ref{sec:lit} reviews the related literature and summarizes our contributions. Section~\ref{sec:formulation} presents the model, establishes the existence of the Markovian traffic equilibrium for ride-hailing (MTER), and develops a solution algorithm. Section~\ref{sec:special} provides stronger theoretical results for directed-cycle networks. Section~\ref{sec:computational} reports computational experiments on stylized networks as well as the Sioux Falls and Chicago sketch networks. Section~\ref{sec:extension} briefly discusses an extension with endogenous driver participation. Finally, Section~\ref{sec:conclusion} concludes.

\section{Literature Review}
\label{sec:lit}
The literature review is organized into three parts. The first reviews traffic equilibrium studies of taxi and ride-hailing systems with flow-dependent travel times. The second and third examine Markovian equilibrium models for traffic congestion and for ride-hailing, respectively, both of which provide a strong foundation for the integration of traffic congestion and vehicle competition developed in this paper.

\paragraph{Traffic equilibrium for taxi/ride-hailing.}

One of the earliest network equilibrium models for for-hire vehicles is developed by \cite{YangWong98}. Although traffic congestion is not considered, its core framework has been adopted and extended in subsequent work. In their model, a fixed number of taxis circulate throughout the network, and after dropping off passengers, vacant taxis choose repositioning destinations to minimize the sum of travel time to the target zone and expected waiting time upon arrival. A Logit model is used to capture perception errors and other unobserved factors. At equilibrium, the origin-destination flows of vacant taxis must match those of passengers, with zonal waiting times serving as the endogenous variables that clear the system. This model is extended in \cite{YangLeunWongBell10} to incorporate bilateral customer-taxi search, in which both taxi and customer search times are endogenous, affect the choices of taxis and customers respectively, and are coupled through a Cobb-Douglas meeting function. In a different direction, \cite{WongWongYang01} incorporate traffic congestion through a bi-level model, where the lower level solves a combined trip distribution and traffic assignment problem given zonal trip totals of hired and vacant vehicles, which need not match, while treating customer and vehicle waiting times as free variables, and the upper level imposes equilibrium conditions that enforce hired and vacant flow balance, relate customer waiting time to vehicle waiting time, and characterize demand as a function of customer waiting time and travel time. \cite{BanDessoukyPangFan2019} further consider multiple ride-hailing platforms, each of which centrally determines the repositioning destinations of its drivers by solving a profit-maximization problem, in contrast to the decentralized driver decision-making assumed in most other work; in their model, solo driving is an outside option for travelers and also contributes to congestion. Finally, \cite{RideSourcingEquilibrium2021} introduce an interzonal matching function that allows vehicles to be matched with customers from neighboring zones, thereby accounting for the contribution of both cruising and deadheading miles to traffic congestion.

\smallskip

\paragraph{Markovian equilibrium for traffic only.}

Markov decision processes, or stochastic dynamic programming, provide an appealing framework when an agent's behavior can naturally be modeled as a sequence of decisions, or when they offer computational advantages in problems with combinatorial solution spaces. An early example of the latter is Dial's LOGIT traffic loading procedure for flow-independent travel times~\citep{Dial71}, in which path-based logit loading over an ``efficient path'' set is decomposed into a sequence of logit loadings at successive nodes over ``reasonable'' links. \cite{Akamatsu96} removes the path-set restriction and extends the approach to the flow-dependent setting, while \cite{MTE_NGEV2022} further incorporate the more behaviorally realistic general extreme value (GEV) model. \cite{MarkovianTA2008} allow for different discrete choice models at different nodes, including mixtures of stochastic and deterministic models. They also derive a dual formulation of the equilibrium as the solution to an unconstrained convex optimization problem in the travel-time space and establish uniqueness of destination-specific link flows.

\smallskip

\paragraph{Markovian equilibrium for ride-hailing only.}

The Markovian approach has also been used to model ride-hailing equilibria without incorporating traffic congestion. \cite{InfiniteHorizonMDPRoutingGame2017} study an infinite-horizon average-reward routing game and shows that it is a potential game. In that model, the reward associated with an action depends on the population mass choosing that action, whereas the state transition probabilities are exogenous, which arguably does not fully capture the endogenous competition among vehicles for passengers. \cite{RIVERTRB2023} study a ride-hailing routing game over discretized time intervals, with the final interval treated as an absorbing state, and incorporates a calibrated matching probability function that governs transitions across intervals based on the population mass in the preceding interval. \cite{SpatialPricingOzanCandoganOR2019} study the joint optimization of passenger prices and driver compensations in a stylized network where travel times are identical between every pair of nodes. For a given price and compensation vector, multiple equilibria with different profit levels may arise. By characterizing equilibria under optimal prices and compensations, the authors reduce the resulting optimization problem with equilibrium constraints to a tractable convex program.

\smallskip

\paragraph{Contribution to the literature.}

We develop to our knowledge the first Markovian traffic equilibrium model for ride-hailing in which both flow-dependent congestion and competition for passengers are determined endogenously in equilibrium. On the modeling side, the sequential decision framework captures the forward-looking behavior of drivers and links the decisions made by the same vehicle in both hired and vacant states. On the computational side, it enables more efficient equilibrium computation for large-scale networks. We establish the existence of equilibrium for general vehicle-flow-based matching probabilities, mass-based travel time functions, and a broad class of discrete choice models by formulating the equilibrium conditions as a fixed-point problem.

\section{Model Formulation}
\label{sec:formulation}
A pool of $M$ independent vehicles, each operated by either a human driver or an autonomous system, circulates over a transportation network $\mathcal{G}=(\mathcal{N},\mathcal{A})$, where $\mathcal{N}$ denotes the set of nodes and $\mathcal{A}$ the set of links, to serve passenger demand. (An extension with endogenous vehicle pool size is presented later in Section~\ref{sec:extension}.) %
Passengers arrive on links with a loss model \citep{ZhangPavone2016,RIVERTRB2023}, that is, a passenger leaves if not matched with an empty vehicle. 
An empty vehicle could receive a destination-revealed order by the end of traversing a link, whose chance is dictated by a matching probability function of the endogenous empty vehicle link flow given the exogenous passenger order arrival rate. An empty vehicle with an order on hand decides to accept or reject at the node before making routing decisions. If rejecting the order or having not received one, it makes a sequence of link choices as an empty vehicle in the hope of meeting a passenger at another location. If accepting the order, it makes sequential link choices as a hired vehicle to carry the passenger to the destination, after which it continues sequential link choices as an empty vehicle. The travel time on a link is a function of the link mass composed of both empty and hired vehicles. 

We are interested in the steady state distribution of hired and empty vehicles over the network, $\{x_a, y^d_a~|~\forall a \in \mathcal{A}, \forall d \in \mathcal{N}\}$, where $x_a$ is the empty vehicle mass (number of vehicles) on link $a$ and $y^d_a$ the hired mass with destination $d$, such that
\begin{equation}
    \sum_a u_a = M,
    \label{eq:totalmass}
\end{equation}
where $u_a$ is the total mass on link $a$, and 
\begin{equation}
    u_a = x_a + \sum_{d \in \mathcal{N}} y_a^d.
    \label{eq:linkmass}
\end{equation} 

The congested travel time on link $a$, $t_a$ is a continuous and strictly increasing function of link mass, that is, 
\begin{equation}
    t_a = t_a(u_a),\quad a \in \mathcal{A}. 
    \label{eq:linktt}
\end{equation}
We assume that the free-flow travel time on every link is strictly positive, i.e., $\hat t \coloneqq \min_{a \in \mathcal{A}} t_a(0) > 0$. Define further $\check t \;\coloneqq\; \max_{a \in \mathcal{A}} t_a(M)$,
so that the congested travel time on any link is bounded within the interval $[\hat t,\check t]$. A mass-based travel-time function is more compatible with standard traffic-flow theories than a flow-based one, because the mapping from flow to travel time is generally not unique~\citep{walters1961theory}. For expository convenience, we assume there is no background traffic; this assumption can be readily relaxed by adding a link-specific constant to $u_a$.

We further write a link flow as the ratio of link mass over travel time and have the expressions for empty link flow $f_a$ and hired link flow $h^d_a$ with destination $d$ as follows:  
\begin{align}
    f_a = \frac{x_a}{t_a}, &\quad \forall a \in \mathcal{A}, \label{eq:linkemptyflow}\\
    h_a^d = \frac{y_a^d}{t_a}, &\quad \forall a \in \mathcal{A},\ \forall d \in \mathcal{N}. \label{eq:linkhiredflow}
\end{align}

In a ride-hailing transportation network, there is another type of congestion---competition among empty vehicles for passengers on each link.
The probability of a randomly chosen empty vehicle on link $a$ receiving an order by the end of the link, $m_a$, 
is a continuous function of empty vehicle flow $f_a$ given the exogenous passenger arrival rate $\lambda_a$, that is,
\begin{equation}
    m_a = m_a(f_a),\ \forall a \in \mathcal{A}.
    \label{eq:matchingfunction}
\end{equation}
The range of the matching probability function $m_a(\cdot)$ is $\{0\}$ when the passenger arrival rate satisfies $\lambda_a = 0$, and $[0,1]$ when $\lambda_a > 0$. 
An example of the matching probability function when $\lambda_a > 0$ is 
\begin{equation}
    m_a(f_a\mid \lambda_a,\gamma) = \left\{
    \begin{array}{ll}    
     \min\left\{\frac{\lambda_a}{f_a}, \left(1 - e^{-\gamma_a\frac{\lambda_a}{f_a}}\right)\right\},   &  \textrm{if} \ f_a > 0 \\
     1, & \textrm{if} \ f_a = 0 \\ %
    \end{array} \right. ,
    \label{eq:matching}
\end{equation}
where $\gamma_a>0$ (friction sensitivity) %
reflects the spatial friction of vehicle-passenger matching and passenger impatience. %
The minimization expression ensures that the matched flow $f_a m_a$ is no greater than the passenger arrival rate $\lambda_a$. %
Note that our theoretical results does \emph{not} rely on the specific form of the function.%

\smallskip

\paragraph{Semi-Markov individual vehicle decision process.} %
Each vehicle makes routing and order-acceptance decisions, given link travel times, node-level matching probabilities, and other parameters such as origin-destination (OD) fares, by solving an infinite-horizon, discounted-return semi-Markov decision process (SMDP). An SMDP is analogous to a Markov decision process (MDP) at transition epochs, but the agent remains in a state for a (possibly random) holding time before the next transition \citep{SemiMarkov1971,MDP_Puterman}.

The state $(s,\epsilon)$ consists of two components. The observable component is
$s \coloneqq (i,o,d)$, where $i\in\mathcal{N}$ denotes the node at which the vehicle is located,
$o\in\{\textrm{hired},\textrm{order-on-hand},\textrm{no-order-on-hand}\}$ indicates the vehicle's service status, and
$d\in\mathcal{N}\cup\{-1\}$ specifies the passenger destination when applicable and equals $-1$ otherwise.
At any node, only the following three combinations of $(o,d)$ can occur:
\begin{itemize}
    \item $o=\textrm{order-on-hand}$ and $d\in\mathcal{N}$: an empty vehicle with an order on hand to destination $d$;
    \item $o=\textrm{no-order-on-hand}$ and $d=-1$: an empty vehicle with no order on hand; and
    \item $o=\textrm{hired}$ and $d\in\mathcal{N}$: a hired vehicle carrying a passenger to destination $d$.
\end{itemize}

The unobserved component $\epsilon(s)\coloneqq\{\epsilon_{\alpha}(s)\mid \alpha\in A(s)\}$, where $A(s)$ denotes the action set at state $s$, represents idiosyncratic action-specific reward shocks at state $s$. These shocks influence the vehicle's decision but are unobserved by the modeler, capturing heterogeneous preferences and random noise in choice behavior. When $\epsilon(s)$ appears alongside $s$, we simplify the notation and write $\epsilon(s)$ as $\epsilon$. We assume that $\epsilon(s)$ admits a joint density on $\mathbb{R}^{|A(s)|}$ that is strictly positive almost everywhere.

The action set $A(s)$ depends on the vehicle's status. If the vehicle has an order on hand with destination $d$, it chooses whether to accept or reject the order, i.e., $A\bigl((i,\textrm{order-on-hand},d)\bigr)=\{\textrm{accept},\textrm{reject}\}$. Otherwise, the vehicle chooses an outgoing link from its current node $i$, so $A(s)=\mathcal{A}_i^{+}$ for $s\in\{(i,\textrm{no-order-on-hand},-1),\, (i,\textrm{hired},d)\}$, where $\mathcal{A}_i^{+}$ denotes the set of outgoing links from node $i$.

A vehicle in state $s=(i,\textrm{order-on-hand},d)$, with an order destined for $d$ on hand, transitions to the next state with probability~$1$. Specifically, upon deciding whether to accept the order, it moves to either $s'=(i,\textrm{hired},d)$, a hired vehicle with destination $d$ ready to choose an outgoing link if it accepts, or $s'=(i,\textrm{empty},-1)$, an empty vehicle if it rejects. Given a link-choice action $a = (i,j)\in \mathcal{A}^+_i$ out of node $i$, there are three types of transitions. %
\begin{itemize}
    \item If $s=(i,\textrm{hired},d)$ and $j\neq d$, then the transition is deterministic to $s'=(j,\textrm{hired},d)$, i.e., the vehicle continues the hired trip toward its destination.

    \item If $s=(i,\textrm{hired},j)$ (equivalently, $d=j$), then the transition is deterministic to $s'=(j,\textrm{empty},-1)$. After dropping off its passenger, %
    it becomes eligible to receive new orders at its subsequent link.%

    \item If $s=(i,\textrm{empty},-1)$, then upon moving to node $j$ the vehicle transitions probabilistically to either $s'=(j,\textrm{order-on-hand},d)$ for some destination $d$, or $s'=(j,\textrm{no-order-on-hand},-1)$, depending on whether an order is received at node $j$.
\end{itemize}
The holding time associated with a link-choice action $a$ is naturally interpreted as the mass-dependent travel time $t_a(u_a)$ from the source node $i$ to the sink node $j$. %

The reward from taking action $\alpha$ in state $(s,\epsilon)$ admits an additively separable decomposition into an observable component and an unobservable shock, namely $r_{\alpha}(s)+\epsilon_{\alpha}(s)$. Both components are realized instantaneously when the action is taken, and no additional reward accrues continuously during the ensuing holding time, as is often assumed in the SMDP literature. Following \cite{Rust1988}, we further impose a conditional independence assumption on state transitions: the next observable state $s'$ depends only on the current observable state $s$ and action $\alpha$, whereas the next-state unobserved shocks $\epsilon'$ depend only on $s'$, implying no serial dependence in unobserved rewards across decision stages. This assumption allows us to work in the reduced observable state space by taking expectations of the Bellman equation defined on the full state space, and it justifies treating the unobserved rewards as part of the state directly, without introducing separate notation for an unobserved state variable and a mapping from that variable to rewards.

A vehicle selects a policy to maximize its long-run expected discounted return, with continuous-time discount rate $\beta>0$. In Appendix~\ref{app:Expected_Bellman_existence_and_contraction}, we establish the existence of a stationary optimal Markov policy in the full state space. Taking expectations on both sides of the associated Bellman equation yields the following expected Bellman equation on the reduced (observable) state space. For notational convenience, we embed the order acceptance/rejection decision directly into the action set, rather than introducing separate decision variables.

We define two link-level action value functions: $z_a$ denotes the (expected) action value of choosing link $a=(i,j)$ when the vehicle is \emph{empty}, and $w_a^d$ denotes the (expected) action value of choosing link $a=(i,j)$ when the vehicle is \emph{hired} with destination $d\in\mathcal{N}$. When an empty vehicle arrives at node $j$, it receives an order with probability $m_a$; conditional on receiving an order at $j$, the destination is $d$ with probability $n_j^d$ (so $\sum_{d\in\mathcal{N}} n_j^d=1$). If the vehicle receives an order to destination $d$, accepting the order yields an immediate deterministic payoff (fare) $\chi_j^d$, whereas rejecting yields no such payoff.

We next define the social surplus operator $G(\cdot)$ \citep{McFadden1981}, which maps a collection of deterministic action values into an expected maximum under additive random utility:
\begin{equation*}
    G\left(\{\eta_\alpha(s)\mid \alpha \in A(s)\}\right)
    \coloneqq
    \int_{\epsilon}
    \left[
        \max_{\alpha \in A(s)}
        \bigl(\eta_\alpha(s)+\epsilon_\alpha(s)\bigr)
    \right]
    \mathfrak{q}(\epsilon\mid s)\,d\epsilon,
\end{equation*}
where $\eta_\alpha(s)$ is the deterministic component of the value of action $\alpha$ at state $s$, and $\mathfrak{q}(\epsilon\mid s)$ is the joint density of the action-specific reward shocks $\epsilon(s)=\{\epsilon_\alpha(s)\mid \alpha\in A(s)\}$. We assume that $G\left(\{\eta_\alpha(s)\mid \alpha \in A(s)\}\right)$ is finite and well defined for every $(\eta_\alpha(s))_{\alpha \in A(s)}\in\mathbb{R}^{|A(s)|}$.

Using $G(\cdot)$, we define the node-level state values. For an empty vehicle at node $j$,
\begin{equation}
    \sigma_j:= G(\{z_{a'} \mid a'\in \mathcal{A}^+_j\}),
    \label{eq:statevalue_empty}
\end{equation}
and for a hired vehicle with destination $d$ at node $j\neq d$,
\begin{equation}
    \tau^d_j:= G(\{w^d_{a'} \mid a' \in \mathcal{A}^+_j\}),\quad j\neq d,
    \label{eq:statevalue_hired}
\end{equation}
with the boundary condition at the destination
\begin{equation}
    \tau^d_d:= \sigma_d.
    \label{eq:statevaluehired_dest}
\end{equation}

We can now write the corresponding link-level action values. For an empty vehicle taking link $a=(i,j)\in\mathcal{A}$ with holding time (travel time) $t_a$,
\begin{equation}
    z_a
    =
    -c^E_a(t_a)
    +
    e^{-\beta t_a}
    \Bigg[
        (1-m_a)\sigma_{j}
        +
        \sum_{d \in \mathcal{N}}
        m_a n^d_{j}\,
        G\bigl(\left\{\chi^d_{j}+\tau^d_{j},\,\sigma_j\right\}\bigr)
    \Bigg],
    \quad \forall a \in \mathcal{A},
    \label{eq:actionvalueempty}
\end{equation}
and for a hired vehicle with destination $d$ taking a link $a=(i,j)$ that does not immediately reach $d$ (i.e., $a\notin\mathcal{A}_d^+$),
\begin{equation}
    w_a^d
    =
    -c^H_a(t_a)
    +
    e^{-\beta t_a}\tau^d_{j},
    \quad \forall a \in \mathcal{A}\setminus \mathcal{A}_d^+.
    \label{eq:actionvaluehired}
\end{equation}

Equation~\eqref{eq:actionvalueempty} expresses the action value $z_a$ for an empty vehicle choosing link $a$ as the sum of an instantaneous cost $-c^E_a(t_a)$ and a discounted continuation value at the downstream node $j$. The continuation value has two components: with probability $(1-m_a)$ the vehicle receives no order and obtains the empty-state value $\sigma_j$; with probability $m_a$ it receives an order, whose destination is $d$ with probability $n_j^d$, and then chooses whether to accept or reject. Accepting yields continuation value $\chi_j^d+\tau_j^d$ (fare plus hired-state value), while rejecting yields $\sigma_j$; applying $G(\cdot)$ to these two alternatives gives the expected value of the acceptance/rejection decision.

Similarly, \eqref{eq:actionvaluehired} gives the action value $w_a^d$ for a hired vehicle with destination $d$ choosing link $a$ as the sum of an instantaneous cost $-c^H_a(t_a)$ and a discounted continuation value $\tau_j^d$ at the downstream node $j$ (for $j\neq d$). At the destination $j=d$, we impose \eqref{eq:statevaluehired_dest}, so the vehicle becomes empty and its continuation value equals $\sigma_d$. Note that the action values $\{w_a^d\mid a\in\mathcal{A}_d^+\}$ are not needed (and hence left undefined), because they do not enter the outgoing flow splits $\{y_a^i\mid a\in\mathcal{A}_i^+\}$, which are set to zero at the destination in \eqref{eq:outgoinghiredflow_self}. We assume that the link cost functions $c^E(\cdot)$ and $c^H(\cdot)$ are bounded and Lipschitz continuous on $[\hat t,\check t]$.

\smallskip

\paragraph{Action choice probabilities.}
If the density of $\epsilon$ has finite first moments for all states, then by the Williams-Daly-Zachary theorem the partial derivative of the social surplus function with respect to $\eta_\alpha(s)$ equals the conditional probability that $\eta_\alpha(s)+\epsilon_\alpha(s)$ is the largest at state $s$. The distributions of $\{\epsilon_\alpha(s)\}_{\alpha\in A(s)}$ can be specified to accommodate a broad class of discrete choice models and may vary across nodes. For example, if $\{\epsilon_\alpha(s)\}_{\alpha\in A(s)}$ are i.i.d.\ Gumbel with zero mean and scale parameter $\vartheta$, then the social surplus function takes the log-sum form 
$$G\big(\{\eta_\alpha(s)\mid \alpha\in A(s)\}\big)=\frac{1}{\vartheta}\ln\!\left(\sum_{\alpha\in A(s)}\exp\big(\vartheta \eta_\alpha(s)\big)\right),$$ 
and the induced choice probabilities follow the multinomial logit model: $$\mathbb{P}(\alpha\mid s)=\frac{\partial G(\{\eta_{\alpha'}(s)\mid \alpha'\in A(s)\})}{\partial \eta_\alpha(s)}=\frac{\exp(\vartheta \eta_\alpha(s))}{\sum_{\alpha'\in A(s)}\exp(\vartheta \eta_{\alpha'}(s))},$$ for all $\alpha\in A(s).$

We specialize the notation to distinguish decisions made by empty and hired vehicles. For an empty vehicle at node $i$, let $p_a$ denote the probability of choosing link $a\in\mathcal{A}^+_i$; for a hired vehicle at node $i$ with destination $d$, let $q_a^d$ denote the probability of choosing link $a\in\mathcal{A}^+_i$. These choice probabilities are given by
\begin{equation}
    p_a \;=\; \frac{\partial G\!\left(\{z_{a'} \mid a'\in\mathcal{A}^+_i\}\right)}{\partial z_a},
    \qquad \forall a \in \mathcal{A}^+_i,
    \label{eq:emptychoice}
\end{equation}
and
\begin{equation}
    q_a^d \;=\; \frac{\partial G\!\left(\{w_{a'}^{\,d} \mid a'\in\mathcal{A}^+_i\}\right)}{\partial w_a^{\,d}},
    \qquad \forall a \in \mathcal{A}^+_i,\ \forall d \in \mathcal{N}\setminus\{i\},
    \label{eq:hiredchoice}
\end{equation}
where $q_a^d$ is not defined for $d=i$ since the vehicle is already at its destination. Finally, let $\xi_i^d$ be the probability that an empty vehicle at node $i$ accepts an on-hand order with destination $d$; then
\begin{equation}
    \xi_i^d \;=\; \frac{\partial G(\chi_i^d+\tau_i^d,\sigma_i)}{\partial(\chi_i^d+\tau_i^d)},
    \qquad \forall i \in \mathcal{N},\ \forall d \in \mathcal{N}\setminus\{i\}.
    \label{eq:orderchoice}
\end{equation}

\smallskip

\paragraph{Flow conservation.}
Given the link matching probabilities $\{m_a\}$, order acceptance probabilities $\{\xi_i^d\}$, and link choice probabilities $\{p_a\}$ and $\{q_a^d\}$, the flow conservation conditions are
\begin{align}
f_a \;=\; p_a \sum_{a' \in \mathcal{A}^-_i} \left[\, f_{a'}(1-m_{a'})+\sum_{d \in \mathcal{N}} f_{a'}m_{a'}n_i^d(1-\xi_i^d) + h_{a'}^{\,i}\right],
&\quad \forall i \in \mathcal{N},\ \forall a \in \mathcal{A}^+_i,
\label{eq:outgoingemptyflow}\\
h_a^{\,d} \;=\; q_a^d \sum_{a' \in \mathcal{A}^-_i} \left[\, f_{a'}m_{a'}n_i^d\xi_i^d + h_{a'}^{\,d}\right],
&\quad \forall i \in \mathcal{N},\ \forall a \in \mathcal{A}^+_i,\ \forall d \in \mathcal{N}\setminus\{i\},
\label{eq:outgoinghiredflow}\\
h_a^{\,i} \;=\; 0,
&\quad \forall i \in \mathcal{N},\ \forall a \in \mathcal{A}^+_i,
\label{eq:outgoinghiredflow_self}
\end{align}
where $\mathcal{A}^-_i$ denotes the set of incoming links to node $i$. Equation~\eqref{eq:outgoingemptyflow} states that the empty flow on an outgoing link $a\in\mathcal{A}^+_i$ equals the total empty flow arriving at node $i$ in an empty state, multiplied by the empty-vehicle choice probability $p_a$. The incoming empty flow consists of three components on each incoming link $a'\in\mathcal{A}^-_i$: vehicles that did not receive an order, $f_{a'}(1-m_{a'})$; vehicles that received an order with destination $d$ but rejected it, $\sum_{d\in\mathcal{N}} f_{a'}m_{a'}n_i^d(1-\xi_i^d)$; and hired vehicles that have reached their destination at $i$, denoted by $h_{a'}^{\,i}$. Similarly, \eqref{eq:outgoinghiredflow} states that the hired flow with destination $d\neq i$ on an outgoing link $a$ equals the total hired flow arriving at node $i$, multiplied by the hired-vehicle choice probability $q_a^d$. This incoming hired flow has two components on each incoming link $a'$: newly hired vehicles whose accepted orders have destination $d$, given by $f_{a'}m_{a'}n_i^d\xi_i^d$, and hired vehicles already en route to $d$, given by $h_{a'}^{\,d}$. Finally, \eqref{eq:outgoinghiredflow_self} enforces that no hired flow departs from its destination.

\subsection{Equilibrium Analysis}

We are now ready to define the equilibrium as a distribution of empty and hired vehicle mass that simultaneously satisfies incentive compatibility and physical feasibility constraints.
\begin{definition}[Markovian Traffic Equilibrium for Ride-Hailing]
A feasible vehicle mass vector $(\mathbf{x},\mathbf{y})\in \mathbb{R}^{|\mathcal{A}|(|\mathcal{N}|+1)}_{\ge 0}$ satisfying the mass balance equation~\eqref{eq:totalmass} is a Markovian Traffic Equilibrium for Ride-Hailing (MTER) if and only if the induced empty and hired link flows $\{f_a\}$ and $\{h_a^d\}$ from \eqref{eq:linkemptyflow}--\eqref{eq:linkhiredflow} satisfy the flow conservation equations~\eqref{eq:outgoingemptyflow}--\eqref{eq:outgoinghiredflow_self} and the expected Bellman equations~\eqref{eq:actionvalueempty}--\eqref{eq:actionvaluehired}, where the choice probabilities are given by~\eqref{eq:emptychoice}--\eqref{eq:orderchoice}, and the link travel times and matching probabilities satisfy the consistency relations $t_a=t_a(u_a)$ and $m_a=m_a(f_a)$ for all $a\in\mathcal{A}$.
\label{def:MTER}
\end{definition}

To analyze MTER, we first establish a proposition showing that the action-value functions $(\mathbf{z},\mathbf{w})$ are continuous implicit functions of $(\mathbf{t},\mathbf{m})$. %

\begin{proposition}[Continuous Implicit Value Function]
For any given travel-time vector $\mathbf{t}$ and matching-probability vector $\mathbf{m}$ the expected Bellman equations~\eqref{eq:actionvalueempty}--\eqref{eq:actionvaluehired} admit a unique solution, denoted by $\mathbf{z}(\mathbf{t},\mathbf{m})$ and $\mathbf{w}(\mathbf{t},\mathbf{m})$. Moreover, the mappings $(\mathbf{t},\mathbf{m})\mapsto \mathbf{z}(\mathbf{t},\mathbf{m})$ and $(\mathbf{t},\mathbf{m})\mapsto \mathbf{w}(\mathbf{t},\mathbf{m})$ are Lipschitz continuous. %
\label{prop:continuousVF}
\end{proposition}

The proof of Proposition~\ref{prop:continuousVF} in Appendix~\ref{app:Expected_Bellman_existence_and_contraction} generalizes results for an MDP from \cite{Rust1988} and \cite{RustTraubHozniakowski2002} to an SMDP. In it we first establish the existence of a stationary optimal Markovian policy in the full state space, and then show that the expected Bellman equation in the observable state space is a contraction. We finally invoke a version of the implicit function theorem for parametric contraction to prove the Lipschitz continuity. This is a stronger result than the continuity needed for the existence proof of Theorem~\ref{thm:existence}, due to the lack of a weaker parametric contraction mapping theorem.  

We are now ready to present the main result on the existence of MTER. 
\begin{theorem}\label{thm:existence}
There exists an MTER $(\mathbf{x}^\ast,\mathbf{y}^\ast)$ that satisfies \Cref{def:MTER}.%
\end{theorem}
\begin{proof}{Proof.}

Define the compact and convex feasible set of link masses as $\mathcal{X}\coloneqq \{(\mathbf{x},\mathbf{y})\in \mathbb{R}_{\ge 0}^{|\mathcal{A}|(|\mathcal{N}|+1)}:\ \sum_{a\in\mathcal{A}}(x_a+\sum_{d\in\mathcal{N}} y_a^d)=M\}$. We formulate the MTER as a fixed-point problem: find $(\mathbf{x},\mathbf{y})\in\mathcal{X}$ such that $(\mathbf{x},\mathbf{y})=F(\mathbf{x},\mathbf{y})$ for a continuous mapping $F:\mathcal{X}\mapsto\mathcal{X}$ defined as the composition $(\mathbf{x},\mathbf{y})\xmapsto[\eqref{eq:matchingfunction}]{\eqref{eq:linktt}}(\mathbf{t},\mathbf{m})\xmapsto[\eqref{eq:actionvaluehired}]{\eqref{eq:actionvalueempty}}(\mathbf{z},\mathbf{w})\xmapsto[\eqref{eq:orderchoice}]{\eqref{eq:emptychoice},\,\eqref{eq:hiredchoice}}(\mathbf{p},\mathbf{q},\boldsymbol{\xi})\xmapsto[\eqref{eq:outgoinghiredflow},\,\eqref{eq:outgoinghiredflow_self}]{\eqref{eq:totalmass},\,\eqref{eq:outgoingemptyflow}}(\mathbf{x},\mathbf{y})$; each constituent mapping is continuous, with the first and third being immediate and the second following from \Cref{prop:continuousVF}.

We establish the continuity of the final mapping.  %
The flow-balance equations \eqref{eq:outgoingemptyflow}, \eqref{eq:outgoinghiredflow}, and \eqref{eq:outgoinghiredflow_self}, together with the total-mass constraint \eqref{eq:totalmass}, form a linear system in masses $(\mathbf{x},\mathbf{y})$, which can be viewed as the full balance condition of an associated continuous-time Markov chain. By Lemma~\ref{lemma:uniqueness} in Appendix~\ref{subsec:ctmc}, this Markov chain has a single recurrent class and the solution is supported on that class. %
Consequently, the balance equation admits a unique solution. For each feasible \((\mathbf{p},\mathbf{q},\boldsymbol{\xi},\mathbf{t},\mathbf{m})\), the flow-balance equations \eqref{eq:outgoingemptyflow}--\eqref{eq:outgoinghiredflow_self}, together with the mass constraint \eqref{eq:totalmass}, form a linear system in \((\mathbf{x},\mathbf{y})\). After deleting one redundant balance equation and appending the mass constraint, we obtain a square system 
\[
\widetilde{\mathbf{A}}(\mathbf{p},\mathbf{q},\boldsymbol{\xi},\mathbf{t},\mathbf{m})
\begin{pmatrix}
\mathbf{x}\\
\mathbf{y}
\end{pmatrix}
=
\mathbf{b}.
\]
By Lemma~\ref{lemma:uniqueness}, this system has a unique solution for every feasible \((\mathbf{p},\mathbf{q},\boldsymbol{\xi},\mathbf{t},\mathbf{m})\), and hence \(\widetilde{\mathbf{A}}(\mathbf{p},\mathbf{q},\boldsymbol{\xi},\mathbf{t},\mathbf{m})\) is nonsingular throughout the feasible domain. Moreover, each entry of \(\widetilde{\mathbf{A}}(\mathbf{p},\mathbf{q},\boldsymbol{\xi},\mathbf{t},\mathbf{m})\) depends continuously on \((\mathbf{p},\mathbf{q},\boldsymbol{\xi},\mathbf{t},\mathbf{m})\). Therefore, 
\[
\begin{pmatrix}
\mathbf{x}\\
\mathbf{y}
\end{pmatrix}
=
\widetilde{\mathbf{A}}(\mathbf{p},\mathbf{q},\boldsymbol{\xi},\mathbf{t},\mathbf{m})^{-1}\mathbf{b}
\]
is continuous in \((\mathbf{p},\mathbf{q},\boldsymbol{\xi},\mathbf{t},\mathbf{m})\), because matrix inversion is continuous on the set of nonsingular matrices. This proves the continuity of the last mapping and the existence of a fixed point follows from Brouwer's fixed-point theorem. $\hfill\square$
\end{proof} 

\smallskip

\paragraph{Multiple equilibria.}
There may exist multiple equilibria, potentially even infinitely many. To see this, consider a two-node network with \emph{two} parallel, identical links from node $1$ to node $2$, and a single return link from node $2$ to node $1$. There is positive passenger demand at node $1$, with all passengers traveling to node $2$, and zero demand at node $2$. This can be interpreted as a shuttle system in which vehicles pick up passengers at node $1$ and deliver them to node $2$. By symmetry, one natural class of candidate equilibria equalizes the link mass, travel time, and flow on the two outgoing links from node $1$ to node $2$. Under this symmetry, the flow conservation condition becomes $\frac{u_{2\rightarrow 1}}{t_{2\rightarrow 1}} = 2\,\frac{u_{1\rightarrow 2}}{t_{1\rightarrow 2}}$, and the total-mass balance is $u_{2\rightarrow 1} + 2u_{1\rightarrow 2} = 1$. Suppose the travel time functions are given by $t_{1\rightarrow 2} = \frac{1}{1-2u_{1\rightarrow 2}}$ and $t_{2\rightarrow 1} = \frac{1}{1-u_{2\rightarrow 1}}$. Then there is a continuum of equilibria parameterized by $u_{2\rightarrow 1}\in(0,1)$, with $u_{1\rightarrow 2} = \frac{1-u_{2\rightarrow 1}}{2}$. In \Cref{sec:special}, we establish equilibrium uniqueness for a special class of networks.

\subsection{Equilibrium Computation}
\label{sec:solution}
A natural choice is the fixed-point iteration, since our equilibrium formulation reduces to a fixed-point problem. Algorithm~\ref{alg:fixed-point} presents a relaxed version. Specifically Step~5 presents a general update rule with three specializations to the fixed-point iteration, method of successive average (MSA) and momentum respectively. 

\begin{algorithm}[htbp]
\footnotesize
\caption{Relaxed Fixed-Point Iteration for Computing the MTER}
\label{alg:fixed-point}	
\textbf{Input}: The travel time function $t_a(\cdot)$; the driver pool size $M$; the passenger arrival rate on each link $\lambda_a$; the probability of demand from node $i$ to node $d$, $n_i^d$; the matching function on each link $m_a(\cdot)$; the discount rate $\beta$; the fare from node $i$ to node $d$, $\chi_i^d$; the social surplus function $G(\cdot)$; step-size rule $\psi^{(k)}$; update direction $\mathbb{\omega}^{(k)}$; and the convergence criterion $\epsilon$.

\textbf{Output}: The equilibrium masses of vacant and hired vehicles, $\mathbf{x}^*$ and $\mathbf{y}^{*}$.

    Set the iteration counter $k=1$.\\
    Initialize the mass distribution $x_a^{(k)}$ and $y_a^{(k)}$ on each link.\\
    \SetKwRepeat{DO}{Do}{while}
 	\DO{$\Big\|
    \begin{bmatrix}
    \tilde{\mathbf{x}}^{(k)} & \tilde{\mathbf{y}}^{(k)}
    \end{bmatrix}^{\!\top}
    -
    \begin{bmatrix}
    \mathbf{x}^{(k)} & \mathbf{y}^{(k)}
    \end{bmatrix}^{\!\top}
    \Big\| > \epsilon$}
        {
        \textbf{Step 1:} Update travel time $\mathbf{t}^{(k)}$ and matching probability $\mathbf{m}^{(k)}$ on each link
        $$t_a^{(k)} = t_a\left(x_a^{(k)} + \sum_d y_a^{d,(k)}\right),\quad m_a^{(k)} = m_a\left(x_a^{(k)}/t_a^{(k)}\right);$$\\
        \textbf{Step 2:} Update the value functions $(\mathbf{z}^{(k)}, \mathbf{w}^{(k)}, \boldsymbol{\sigma}^{(k)}, \boldsymbol{\tau}^{(k)})$ by solving \eqref{eq:statevalue_empty}--\eqref{eq:actionvaluehired} via value iteration;\\
        \textbf{Step 3:} Compute link choice probabilities $\mathbf{p}^{(k)}$ and $\mathbf{q}^{(k)}$ and order acceptance probability $\boldsymbol{\xi}^{(k)}$
        \begin{align*}
        p_a^{(k)} &= \frac{\partial G\!\left(\{z_a^{(k)} \mid a \in \mathcal{A}_i^+\}\right)}{\partial z_a^{(k)}},
        \qquad \forall a \in \mathcal{A}_i^+, \\
        q_a^{d,(k)} &= \frac{\partial G\!\left(\{w_a^{d,(k)} \mid a \in \mathcal{A}_i^+\}\right)}{\partial w_a^{d,(k)}},
        \qquad \forall a \in \mathcal{A}_i^+, \; \forall i \neq d, \\
        \xi_i^{d,(k)} &= \frac{\partial G\!\left(\chi_i^d + \tau_i^{d,(k)}, \sigma_i^{(k)}\right)}{\partial \!\left(\chi_i^d + \tau_i^{d,(k)}\right)},
        \qquad \forall i \neq d.
        \end{align*}
        
        \textbf{Step 4:} %
        Solve the flow balance equations \eqref{eq:outgoingemptyflow} to \eqref{eq:outgoinghiredflow_self} to obtain $\tilde{\mathbf{x}}^{(k)}$ and $\tilde{\mathbf{y}}^{(k)}$; \\
        \textbf{Step 5:} Update mass distributions 
        $$\mathbf{x}^{(k+1)} \leftarrow \mathbf{x}^{\,(k)} + \psi^{(k)} \mathbf{\omega_x}^{(k)}, \mathbf{y}^{(k+1)} \leftarrow \mathbf{y}^{(k)} + \mathbf{\psi}^{(k)} \mathbf{\omega_y}^{(k)},$$ where $\mathbf{\omega_x}^{(k)}$ and $\mathbf{\omega_x}^{(k)}$ are the update directions and $\psi^{(k)}$ step size, with the following specializations:\\
        \quad \ Fixed-point iteration: $\mathbf{\omega_x}^{(k)} = \tilde{\mathbf{x}}^{(k)} - \mathbf{x}^{(k)}$, $\mathbf{\omega_y}^{(k)} = \tilde{\mathbf{y}}^{(k)} - \mathbf{y}^{(k)}$, and $\psi^{(k)} = 1$; \\
        \quad \ MSA: $\mathbf{\omega_x}^{(k)} = \tilde{\mathbf{x}}^{(k)} - \mathbf{x}^{(k)}$, $\mathbf{\omega_y}^{(k)} = \tilde{\mathbf{y}}^{(k)} - \mathbf{y}^{(k)}$, and $\psi^{(k)} = 1/(k+1)$; \\
        \quad \ Momentum: $\mathbf{\omega_x}^{(k)} = b  \mathbf{\omega_x}^{(k-1)} + (1-b)\left(\tilde{\mathbf{x}}^{(k)} - \mathbf{x}^{(k)}\right)$,  
$\mathbf{\omega_y}^{(k)} = b \mathbf{\omega_y}^{(k-1)} + (1-b)\left(\tilde{\mathbf{y}}^{(k)} - \mathbf{y}^{(k)}\right)$, \\ \quad \quad where $b \in [0,1)$ is the momentum parameter and $\omega^{(0)} = 0$, and $\psi^{(k)}$ is a constant independent of $k$; \\
        Set the iteration counter $k \leftarrow k + 1$;}
        $\mathbf{x}^*\leftarrow \mathbf{x}^{(k+1)}$, $\mathbf{y}^{*}\leftarrow \mathbf{y}^{(k+1)}.$
\end{algorithm}
Note that the value iteration in Step 2 of the algorithm is provably convergent, since the associated Bellman operator is a contraction mapping (see Appendix \ref{app:Expected_Bellman_existence_and_contraction}). 
We do not provide a general convergence guarantee for Algorithm~\ref{alg:fixed-point}, although our computational experiments demonstrate strong convergence behavior. Moreover, in Section~\ref{sec:special}, we establish convergence results for a special class of networks.

\section{Equilibrium Analysis on a Directed Cycle Network}
\label{sec:special}

We consider a special class of transportation networks that consists of a single directed cycle, for which equilibrium uniqueness and convergence of  \Cref{alg:fixed-point} can be established formally. %
Consider a network  $\mathcal{N}=\{1,\dots,|\mathcal{N}|\}$, arranged on a single directed cycle with no subcycles. The network flow satisfies the following properties. 
Each node $i$ has exactly one outgoing link to node $i+1$, where node $|\mathcal{N}|+1 \equiv 1$. 
For this directed cycle network, we have the following result on equilibrium uniqueness and convergence.

\begin{proposition}[Uniqueness and Convergence of Equilibrium Link Masses]
\label{prop:single_cycle_uniqueness_convergence}
Consider a directed cycle network. Then the equilibrium vector of total link masses $(u_a)_{a \in \mathcal{A}}$ is unique, where $u_a$ denotes the sum of the hired and vacant masses on link $a$. Moreover, the sequence of total link masses generated by the fixed-point iteration in Algorithm~\ref{alg:fixed-point} converges to this unique equilibrium, provided that, for every $a \in \mathcal{A}$, the map $u_a \mapsto u_a/t_a(u_a)$ is strictly increasing and $t_a(\cdot)$ is continuously differentiable.
\end{proposition}

Classical traffic flow theory suggests that the relationship between flow, \(u_a/t_a(u_a)\), and mass need not be monotone because of the hypercongestion phenomenon: when the vehicle density becomes sufficiently high, travel speeds may deteriorate so sharply that throughput decreases (see, e.g., \citealt{walters1961theory}). Accordingly, the monotonicity assumption on \(u_a / t_a(u_a)\) may be interpreted as requiring that congestion remain within a moderate regime and never reach a saturated, hypercongested state.

By Definition~\ref{def:MTER}, the equilibrium mass distribution is characterized by flow conservation across links together with link choice and order acceptance probabilities induced by the value functions. In the directed-cycle setting, however, each node has exactly one outgoing link, so link choices are deterministic with probability one, which substantially simplifies the equilibrium analysis. %

The proof of \Cref{prop:single_cycle_uniqueness_convergence} shows that the fixed-point mapping is a contraction with respect to a spread metric on the domain of feasible total link masses. Therefore, according to the Banach fixed-point theorem, there exists a unique fixed point characterizing the equilibrium, and the iterates generated by \Cref{alg:fixed-point} converge globally to it. %

\Cref{prop:single_cycle_uniqueness_convergence} establishes uniqueness and convergence only for the equilibrium total link mass aggregated over hired and vacant vehicles, while leaving the disaggregated hired and vacant masses uncharacterized. These disaggregated masses are difficult to analyze even in a directed cycle network. Although the topology is simple, passengers may still travel between any origin-destination pair on the cycle, so the mass on each link may include flows associated with multiple origin-destination pairs. Moreover, even though link choices are deterministic, trip acceptance probabilities remain endogenous through the equilibrium value-to-go functions.

\section{Computational Experiments}
\label{sec:computational}
We conduct computational tests for two main purposes. First we show the convergence of \Cref{alg:fixed-point} 
in medium- and large-sized networks, attesting to its practicality despite the lack of convergence proof in a general network.  Next we conduct two ablation studies to understand the modeling bias introduced by restricting drivers' forward-looking capabilities and ignoring traffic congestion respectively. Furthermore, Appendix~\ref{sec:braess} presents an example of the Braess's paradox adapted to our ride-hailing traffic equilibrium setting, and Appendix~\ref{subsec:sensitivity} provides  sensitivity analyses with respect to some key parameters to understand the model's behaviors.

The tests are conducted on some toy networks whose settings will be described in the relevant subsections, as well as two networks commonly used for traffic equilibrium modeling: the Sioux Falls and  Chicago sketch network. The Sioux Falls network comprises 24 nodes, 76 links, 528 effective origin-destination (OD) pairs with positive demand, and a total demand of 360,600 per hour. The Chicago sketch network comprises 933 nodes, 2,950 links, 142,890 effective OD pairs, and a total demand of 1,260,907 per hour. We set the default discount rate $\beta=0.1$ per hour and a uniform matching probability friction parameter $\gamma=0.8$ \citep{JiangGao2026}.  We use a linear travel time function $t_a(u_a) = t_a(0)\left(1+ u_a/\bar{c}_a\right)$, where $\bar{c}_a$ is the link jam mass (capacity) calculated by using an average vehicle length of 6 meters and a two-lane configuration. There is a lack of research on a mass-based travel time function suitable for traffic equilibrium analysis, and we use a simple functional form at this initial stage. %
Both hired and vacant vehicles have the same operating costs structure, $c_a^E(t_a) = c_a^H (t_a) = 6 t_a$, where 6 (\$/hour) is the unit fuel cost calculated based on an average fuel price of \$3/gallon, fuel efficiency of 20 mpg and speed of 40 mph. Other costs such as vehicle ownership and service platform costs are treated as independent of vehicle routing and order acceptance decisions and thus excluded from a driver's utility function.  The fare $\chi_j^d$ from the origin $j$ to destination $d$ follows the New York City taxi fare rule with a \$3 base fare plus \$0.70 per 1/5 miles, where the OD distance is calculated by multiplying the fastest path travel time between $j$ and $d$ with an assumed free flow speed of 40 mph. A logit model is used for link choice and order acceptance with a scale parameter of 10 (close to deterministic choices). 
To mitigate the impact of potentially multiple equilibria, for each parameter setting we compute the MTER under three different randomly chosen initial mass distributions and select the equilibrium that yields the highest profit.

\subsection{Algorithm convergence}

\begin{figure}[!h]
    \begin{subfigure}[b]{0.49\textwidth}
        \centering
        \includegraphics[width=\linewidth]{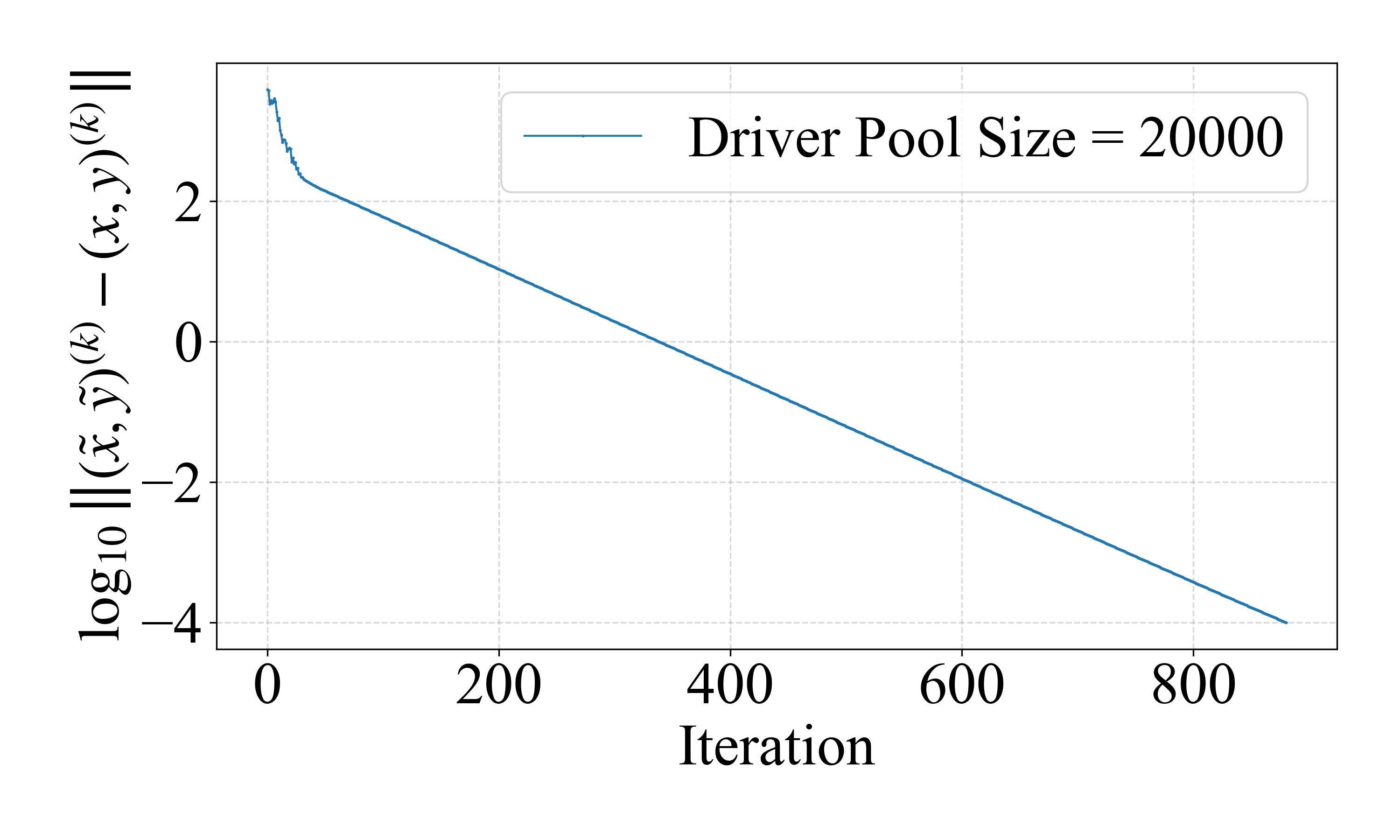}
        \caption{Sioux Falls (MSA with a step size floor)}
        \label{fig:siouxfalls_msa}
    \end{subfigure}
    \hfill
    \begin{subfigure}[b]{0.49\textwidth}
        \centering
        \includegraphics[width=\linewidth]{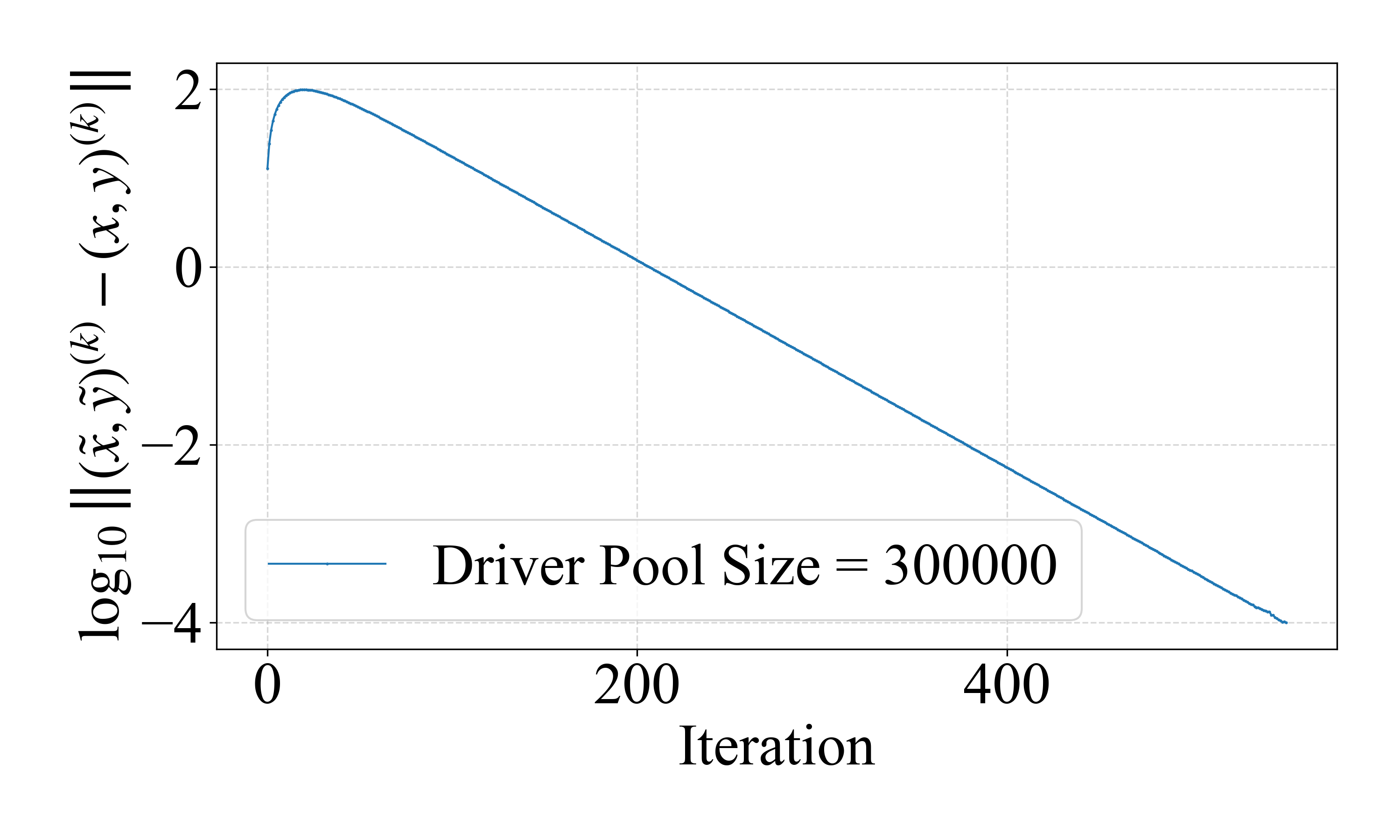}
        \caption{Chicago Sketch (Momentum)}
        \label{fig:chicago_msa}
    \end{subfigure}
    \caption{Convergence of Algorithm~\ref{alg:fixed-point}.}
    \label{fig:msa_convergence}
\end{figure}

Convergence is measured by the $\ell_2$ norm between the current solution and its fixed point mapping, $\|(\tilde x^{(k)}, \tilde y^{(k)}) - (x^{(k)}, y^{(k)})\|_2$.  On the Sioux Falls network with a pool of 20{,}000 vehicles, we adopt MSA with a step size floor, $a^{(k)} = \max\{\frac{1}{k+1},\, 0.02\}$. The gap drops below $10^{-4}$ after 881 iterations with a computational time of 2,116 seconds on a Mac laptop with an Apple M4 chip. On the Chicago sketch network with a pool of 300{,}000 vehicles, we adopt the momentum acceleration with the momentum parameter $b=0.9$ and fixed step size $\psi^{(k)} = 0.02$ for all $k$. The gap drops below $10^{-4}$ after 932 iterations in about 
22 hours on a workstation with an AMD Ryzen Threadripper PRO processor (16-core, 4.00--4.50\,GHz). %

\subsection{Ablation Studies}
\paragraph{Myopic vs. forward-looking drivers.}

There is empirical evidence that drivers make forward-looking repositioning decisions beyond their next passenger~\citep{JiangGao2026}, that is, not only they consider the profit from the next passenger but also the prospect after dropping them off. The MTER framework allows for capturing such behavior as well as studying the impact of ignoring it. A myopic driver makes decisions accounting for the return only up to the drop-off of the next passenger, and is modeled by setting the value of a hired vehicle with destination $d$ to 0 once reaching its destination, that is, equation \eqref{eq:statevaluehired_dest} becomes $\tau^d_d = 0$. 

We construct a simple network in Figure~\ref{fig:myopic_nonmyopic} to compare results of myopic vs. forward-looking drivers. %
Node 1 in the middle is a major decision point for a driver seeking their next passenger, who could either go north towards the airport (node 5) with a large demand that mostly goes to the suburban areas (nodes 6 and 7) with a low demand, or south towards downtown areas (nodes 2, 3 and 4) with a equally large demand as the airport that however mostly stay downtown with a larger demand than in the suburb areas. Detailed network parameters are included in Appendix~\ref{app:myopic_network}. 

Figure~\ref{fig:myopic_nonmyopic} shows the distribution of equilibrium vehicle mass across these regions and particularly part (c) the difference, with red indicating a larger number of forward-looking drivers. Forward-looking drivers prefer downtown more than myopic drivers do---the subsequent profitability after the immediate next passenger is higher downtown than in the suburb but myopic drivers do not account for it. 

In practice, drivers may be heterogeneous in both their ability and willingness to act strategically. MTER can be readily extended to provide a unified framework that captures both forward-looking and myopic behavior within the same driver population by introducing two driver types: one governed by forward-looking expected Bellman equations and the other by myopic decision rules. This naturally gives rise to two sets of expected Bellman equations, two corresponding sets of choice probabilities, and a single joint loading process.

\begin{figure}
    \begin{subfigure}{0.28\textwidth}
        \centering
        \includegraphics[width=\linewidth]{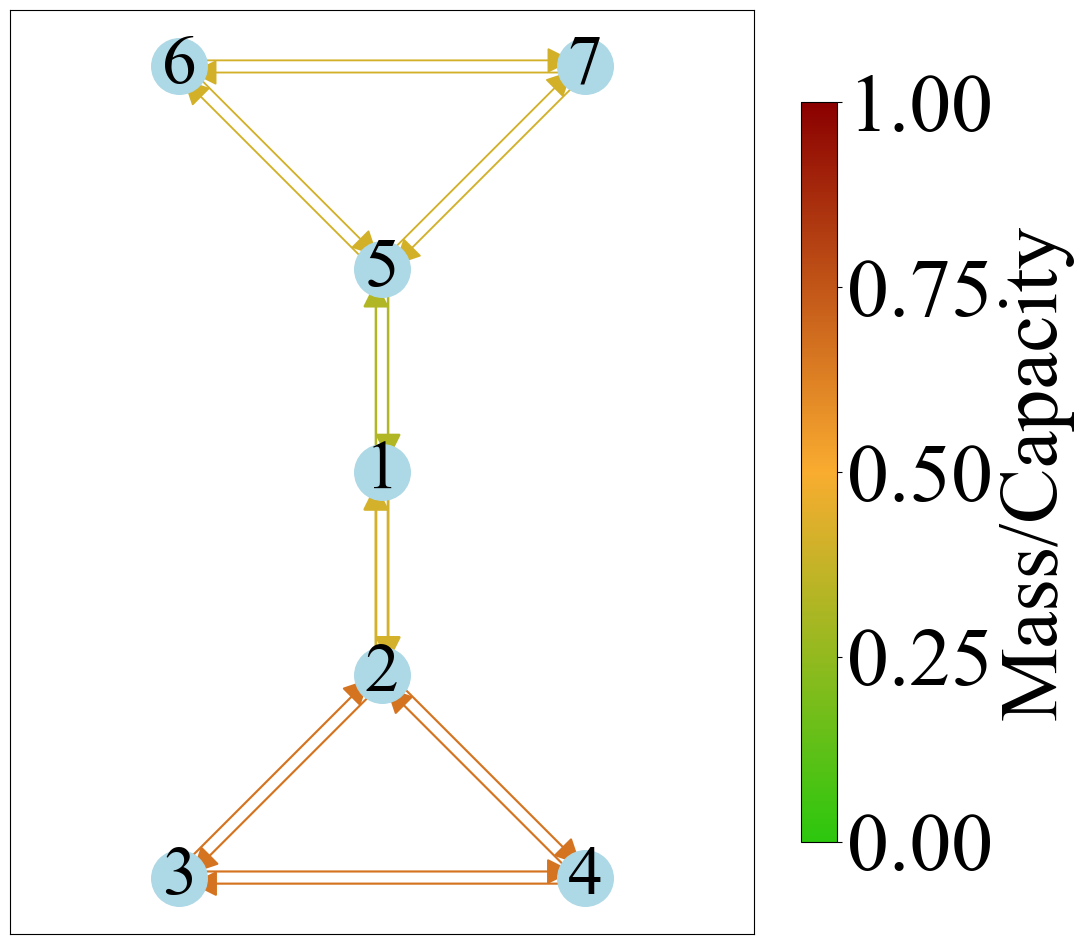}
        \caption{Forward-looking}
    \end{subfigure}
    \hfill
    \begin{subfigure}{0.28\textwidth}
        \centering
        \includegraphics[width=\linewidth]{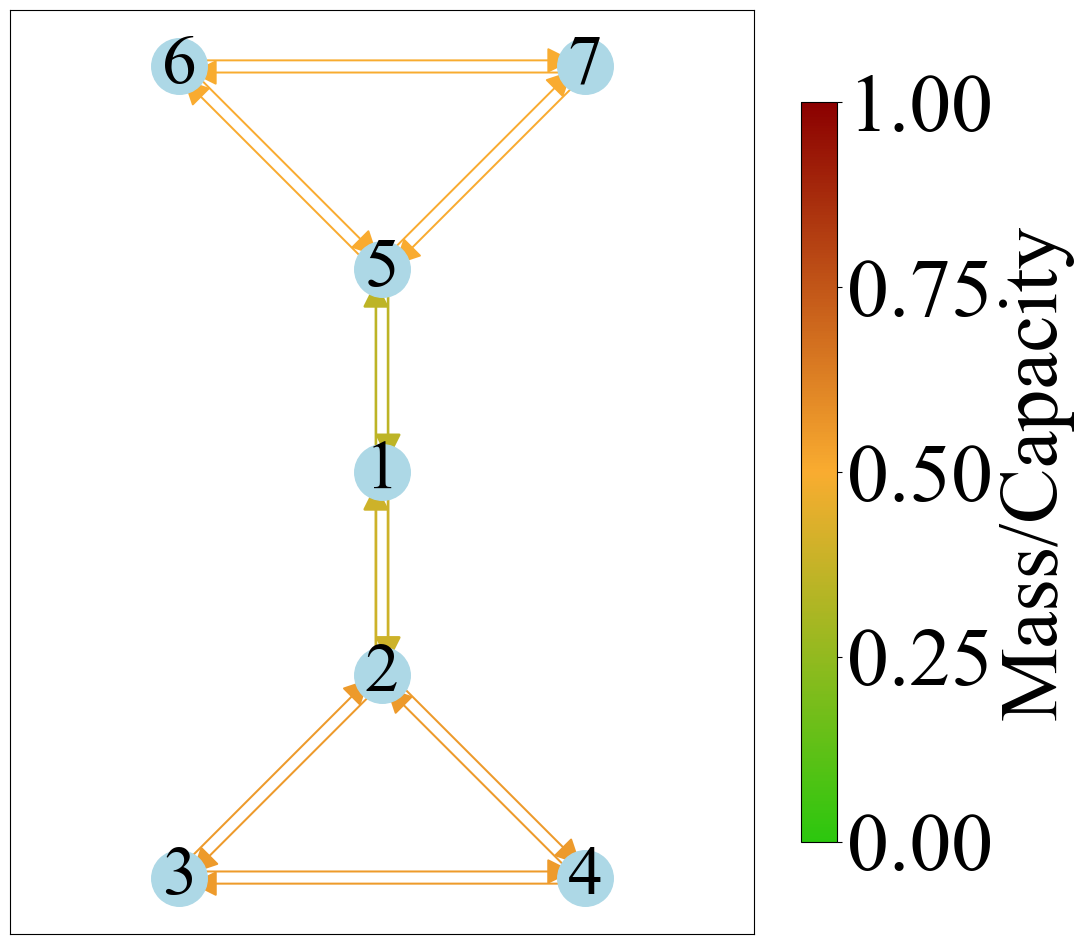}
        \caption{Myopic}
    \end{subfigure}
    \hfill
        \begin{subfigure}{0.28\textwidth}
        \centering
        \includegraphics[width=\linewidth]{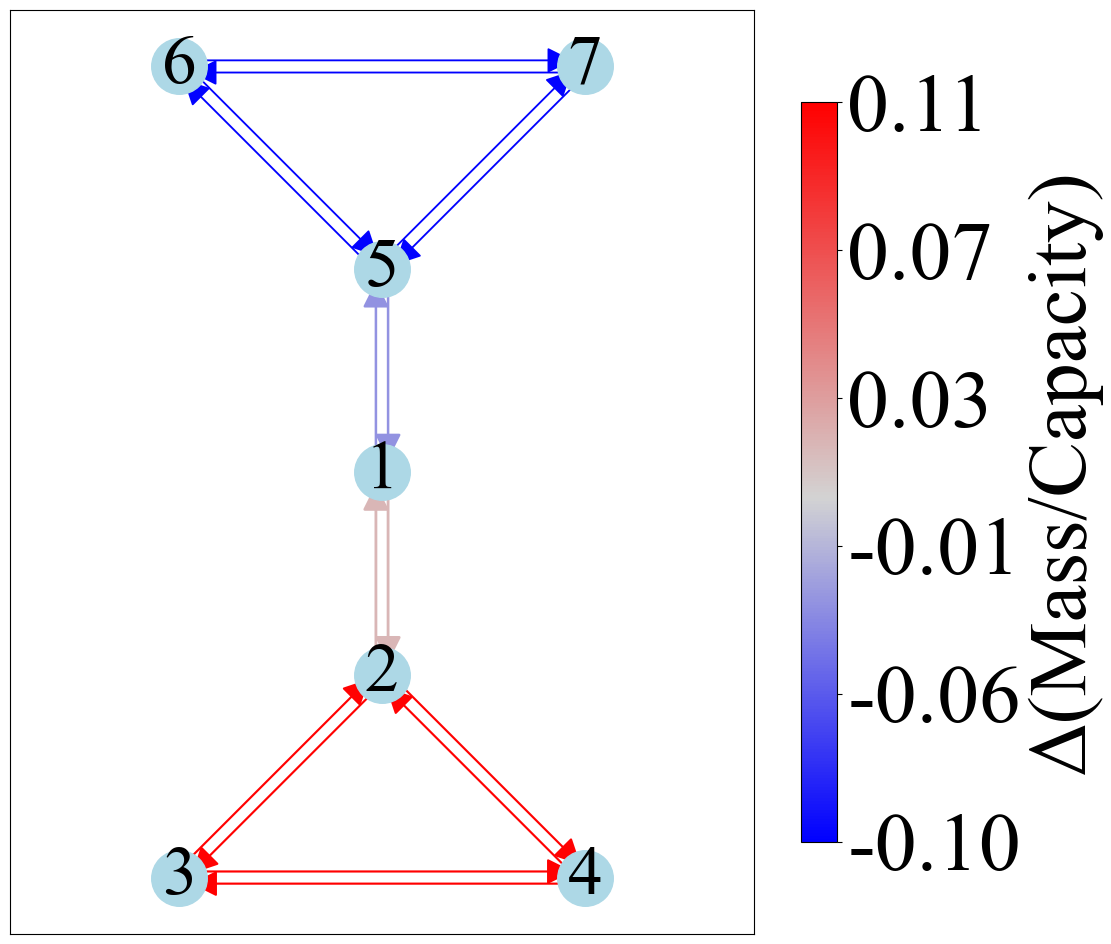}
        \caption{(a) minus (b)}
        \label{fig:myopic_delta}
    \end{subfigure}
    \caption{Equilibrium traffic on a stylized network (myopic vs. forward-looking drivers).}
    \label{fig:myopic_nonmyopic}
\end{figure}

\smallskip

\paragraph{Toll revenue projection with or without traffic congestion awareness.}

Congestion pricing has been recognized as an effective strategy for managing traffic congestion and generating revenue for the more sustainable transit mode, with New York City in January 2025 joining the moderate list of major cities implementing such a strategy including London, Stockholm and Singapore.  It is politically controversial, and for its success accurate traffic prediction and the resulting revenue projection could play an important role.%

We set up a cordon charge scheme in the Sioux Falls network and examine how ignoring traffic congestion in a ride-hailing equilibrium model affects the revenue prediction. In the congestion-aware case, a regular MTER is obtained and relevant metrics are calculated. In the congestion-unaware case, an MTER is calculated with link travel times fixed to obtain drivers' link choice and acceptance probabilities in a non-congested network. These choice probabilities are then held constant in a network loading process that calculates the stationary distribution of the vehicle mass with traffic congestion so that the two modeling results are compared on the same basis. Note that an iterative process is needed to do the loading as there is still  flow-dependent competition for passengers. 
The congestion charge zone is the shaded area defined by nodes 10, 11, 14 and 15. Drivers are charged a fee of \$2 if entering the zone using any of the links leading to one of the four nodes%
. The driver pool size is set to 20{,}000%
. The resulting mass distributions and their differences are presented in Figure~\ref{fig:congestion_noncongestion}.

\begin{figure}
    \begin{subfigure}{0.325\textwidth}
        \centering
        \includegraphics[width=\linewidth]{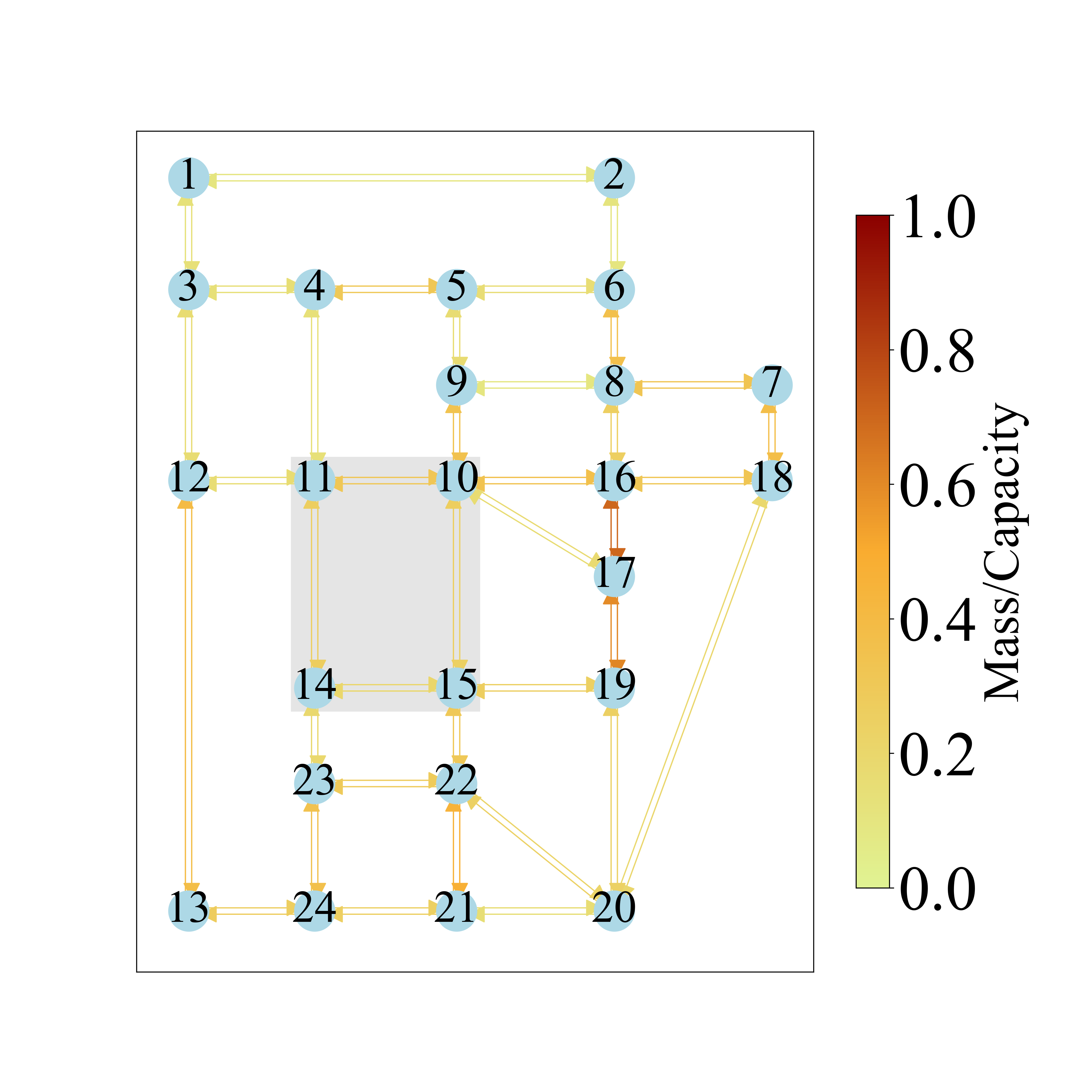}
        \caption{Congestion-aware}
    \end{subfigure}
    \begin{subfigure}{0.325\textwidth}
        \centering
        \includegraphics[width=\linewidth]{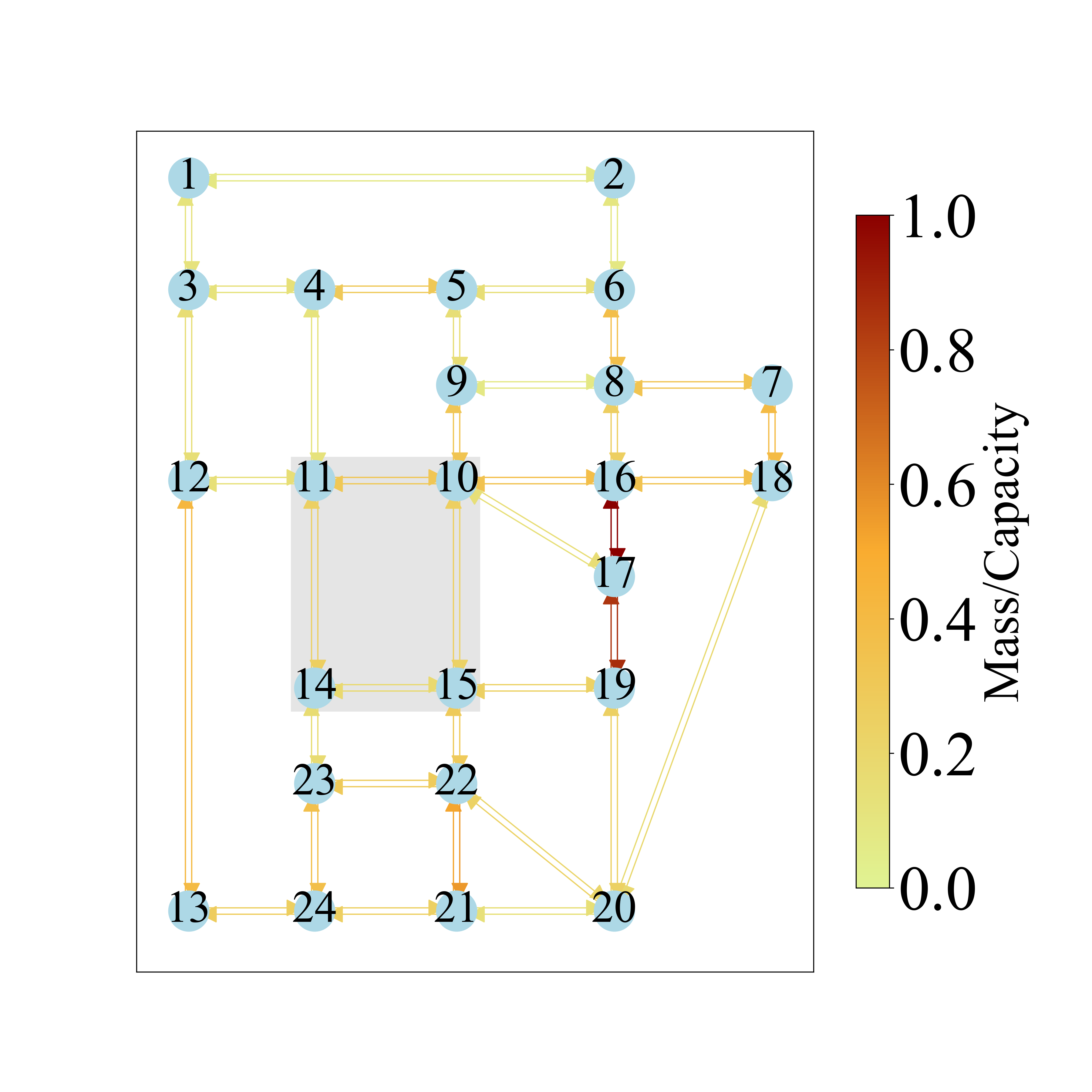}
        \caption{Congestion-unaware}
    \end{subfigure}
    \begin{subfigure}{0.325\textwidth}
        \centering
        \includegraphics[width=\linewidth]{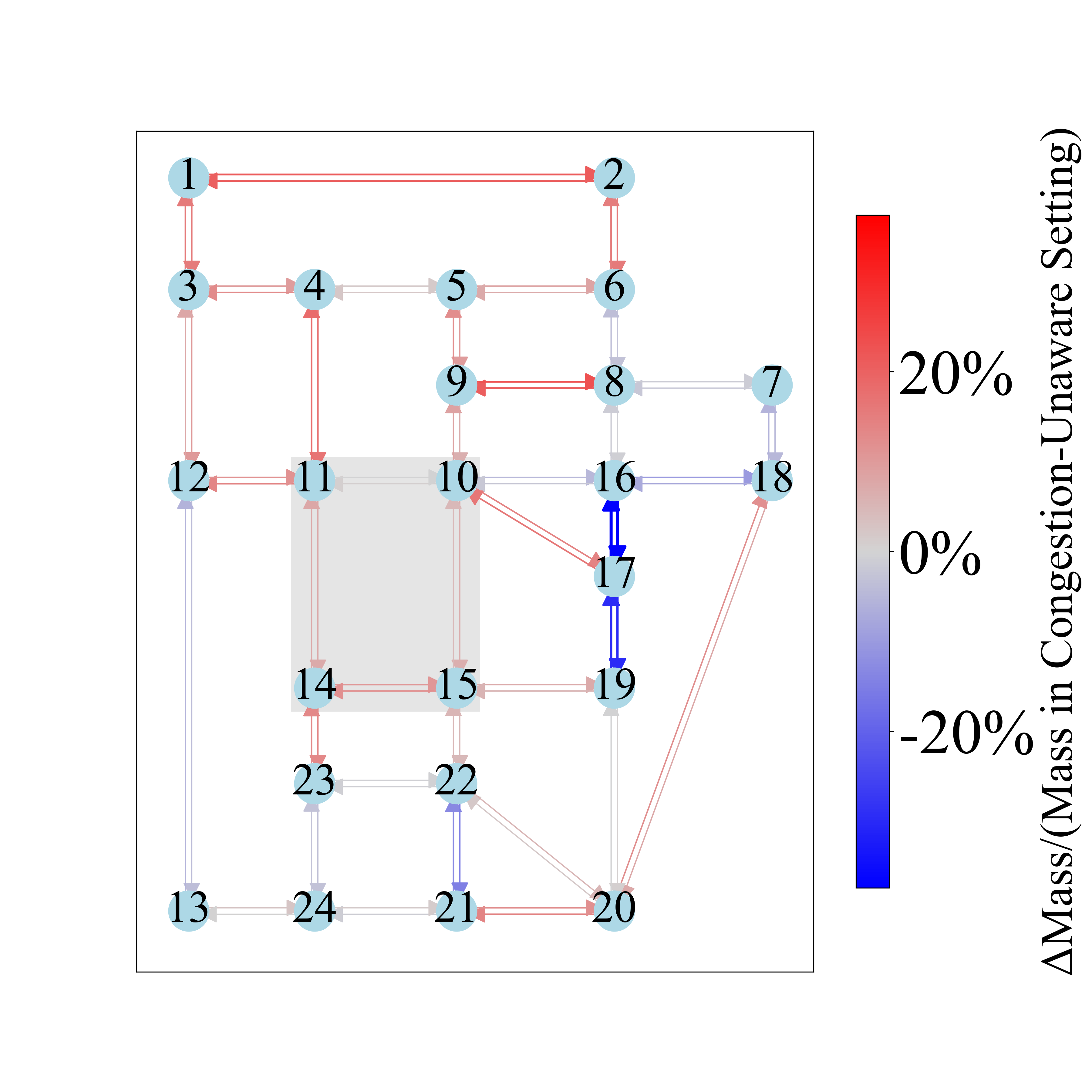}
        \caption{(a) minus (b)}
        \label{fig:congestion_delta}
    \end{subfigure}
    \caption{Equilibrium traffic on the Sioux Falls network (congestion-aware vs. congestion-unaware).}
    \label{fig:congestion_noncongestion}
\end{figure}

Congestion awareness significantly alters the mass distribution in some parts of the network as seen in Figure~\ref{fig:congestion_delta}:  the most congested links $(16,17)$ and $(17,19)$ experience reduced congestion (dark blue), and  %
links entering the cordoned area, such as $(4,11)$ and $(17,10)$, become more congested (dark red).
The congestion-aware setting, which better reflects real-world conditions, leads to a predicted toll revenue of \$50.5K, a $6.09\%$ increase compared to \$47.6K  obtained from the congestion-unaware case. %
This difference arises because congestion-unaware drivers may simply divert to avoid tolled links, whereas congestion-aware drivers anticipate the congestion caused by such diversion and may instead choose to pay the toll to avoid surrounding congestion.

\section{Extension: Endogenous Driver Pool Size}
\label{sec:extension}

We extend MTER by allowing drivers to decide whether to participate in the ride-hailing service. 
At each node $i \in \mathcal{N}$, there are $M_i$ potential drivers. %
Each driver at node $i$ chooses to join the platform with probability $\mathbb{P}_i(\sigma_i)$, which is a continuous function of the value function of an empty vehicle at node $i$.  For example, a binary logit model with an outside option value of $0$ gives 
\begin{equation}
    \mathbb{P}_i = \frac{1}{1 + \exp(-\zeta\sigma_i)},
    \label{eq:participation_choice}
\end{equation}
where $\zeta$ is a dispersion parameter reflecting the level of noise in the participation choice.  A higher value of $\zeta$ indicates a lower level of noise and thus more deterministic choice.  The expected number of participating drivers at node $i$ is
\begin{equation}
    n_i = M_i \mathbb{P}_i .
    \label{eq:n}
\end{equation}
The total mass of vehicles in the network satisfies
\begin{equation}
    \sum_{a\in\mathcal{A}}\Bigl(x_a + \sum_{d\in\mathcal{N}} y_a^d \Bigr)
    = \sum_{i\in\mathcal{N}} n_i.
    \label{eq:extended_sum_M}
\end{equation}

We now define the equilibrium for the extended model.

\begin{definition}[MTER with Driver Participation]
A feasible vehicle mass vector $(\mathbf{x},\mathbf{y})\in \mathbb{R}^{|\mathcal{A}|(|\mathcal{N}|+1)}_{\ge 0}$ satisfying the aggregate participation constraint $\sum_{a\in\mathcal{A}} \left(x_a+\sum_{d\in\mathcal{N}} y_a^d\right)\le \sum_{i\in\mathcal{N}} M_i$ is a Markovian Traffic Equilibrium for Ride-Hailing with Driver Participation if and only if the vehicle masses satisfy the mass balance equations~\eqref{eq:n} and~\eqref{eq:extended_sum_M}, the induced empty and hired link flows satisfy the flow conservation equations~\eqref{eq:outgoingemptyflow}--\eqref{eq:outgoinghiredflow_self}, the expected Bellman equations~\eqref{eq:actionvalueempty}--\eqref{eq:actionvaluehired} hold, the link choice probabilities are given by~\eqref{eq:emptychoice}--\eqref{eq:hiredchoice}, the order acceptance probabilities are given by~\eqref{eq:orderchoice}, the market participation choices are given by~\eqref{eq:participation_choice}, and the link travel times and matching probabilities satisfy the consistency relations $t_a=t_a(u_a)$ and $m_a=m_a(f_a)$ for all $a\in\mathcal{A}$.
\label{def:extended_MTER}
\end{definition}

For the extended model, we retain the existence of an equilibrium.
\begin{theorem}
There exists an extended MTER solution $(\mathbf{x}^\ast, \mathbf{y}^\ast)$
that satisfies Definition~\ref{def:extended_MTER}.
\end{theorem}

\begin{proof}{Proof.}
The proof follows a similar argument as that for Theorem~\ref{thm:existence}. The compact and convex set of feasible link mass is revised to be $\mathcal{X}:=\{(\mathbf{x},\mathbf{y})\in\mathbb{R}_{\ge0}^{|\mathcal{A}|(|\mathcal{N}|+1)}:  \sum_{a\in\mathcal{A}} (x_a + \sum_{d \in \mathcal{N}} y^d_a) \leq \sum_{i\in\mathcal{N}} M_i\}$. The fixed-point mapping $F$ is modified to account for endogenous participation,  specifically, 
\[
(\mathbf{x},\mathbf{y})
\overset{\eqref{eq:linktt}, \eqref{eq:matchingfunction}}{\longmapsto} (t,m)
\overset{\eqref{eq:actionvalueempty}, \eqref{eq:actionvaluehired}}{\longmapsto} (z,w)
\overset{\eqref{eq:emptychoice}, \eqref{eq:hiredchoice}, \eqref{eq:orderchoice}, \eqref{eq:n}}{\longmapsto} (p,q,\xi,n)
\overset{\eqref{eq:outgoingemptyflow}, \eqref{eq:outgoinghiredflow}, \eqref{eq:outgoinghiredflow_self},\eqref{eq:extended_sum_M}}{\longmapsto} (\mathbf{x},\mathbf{y}).
\]
Each component of the mapping is continuous. The existence of a fixed point follows from the Brouwer's fixed-point theorem.\hfill$\square$
\end{proof}

Numerical experiments including convergence and sensitivity analyses for this extended model can be found in Appendix~\ref{subsec:end_driver}.

\section{Conclusions and Future Directions}
\label{sec:conclusion}
We study the traffic equilibrium problem for a ride-hailing system where drivers make sequential order acceptance and link-level routing decisions with endogenous competition and congestion to maximize discounted total return subject to random perturbations over an infinite horizon in a semi-Markov fashion. The equilibrium, MTER, is cast as a fixed-point problem of the vehicle mass distribution over hired and vacant states on all links, and an extension is provided to allow for an endogenous vehicle pool size. The existence of MTER is proved by showing the involved mappings as continuous functions.  Uniqueness of aggregate link masses is proved in a special network, although multiple equilibria could exist in a general network.  A momentum-accelerated fixed point iteration algorithm converges in about 22-hour's computational time on the Chicago sketch network with a realistic driver pool size of 300,000.  Ablation studies show a bias towards more short-term gains from myopic drivers vs. their forward-looking counterparts, and a lower toll revenue prediction when congestion is ignored in a ride-hailing cordon charge scheme. 

The research can take a few different directions in modeling and theoretical analysis. On the modeling side, the current customer-driver matching is local on a link, and the extension to include surrounding links such as in \cite{RideSourcingEquilibrium2021} would make it better reflecting how ride-hailing matching is carried out. Furthermore, drivers/vehicles might not act completely independently in a future ride-hailing system, e.g., robotaxi firms like Waymo could centrally dictate its vehicle's behaviors. A framework that can accommodate a range of driver organizational structures would be desirable. On the theoretical side, existing Markovian equilibrium models with flow or mass-dependent reward/cost~\citep{MarkovianTA2008,InfiniteHorizonMDPRoutingGame2017} can be reformulated as potential games, while the flow-dependent transition matrix in MTER imposes a major challenge in the endeavor. A promising direction is to explore the boundary of a Markovian equilibrium model with both endogenous rewards and transition probabilities and identify the conditions under which a potential function or equivalent mathematical programming formulation exists.

\ACKNOWLEDGMENT{Gao and Yan gratefully acknowledge the Schloss Dagstuhl Seminar 24281, \emph{Dynamic Traffic Models in Transportation Science}, for inspiring ideas on dynamic congestion games and for fostering this collaboration.}

\bibliographystyle{informs2014} %
\bibliography{gao-refs} %

\medskip

\begin{APPENDICES}
\section{Additional Results and Proofs in \Cref{sec:formulation}}

\subsection{Single Agent SMDP in the Full and Reduced State Space}
\label{app:Expected_Bellman_existence_and_contraction}

In this section we present technical details in analyzing an SMDP with potentially unbounded state elements and rewards.  The results are general and not specific to the MTER model. Certain parts are repeated from the main text for a self-contained exposition. The development follows largely those in \cite{Rust1988} for an MDP, but is more general with action-dependent, deterministic holding times. 

A state $(s, \epsilon)$ has two parts. The first part $s$ belongs to the observable state space $\mathcal{S}$ with a dimension $d$, a compact subset of  $\mathbb{R}^d$.  The unobserved $\epsilon(s) \equiv \{\epsilon_\alpha(s)~|~\alpha \in A(s)\}$ are rewards of taking available actions at state $s$ that affect the agent's decision but unknown to the modeler. The unobservables $\epsilon$ can be unbounded, e.g., as jointly normally distributed random variables. An agent's action set $A(s)$ at state $s$ is discrete, and its cardinality is uniformly bounded across all states.

Following \cite{Rust1988}, we make an assumption of conditional independence in state transition, such that the transition probability density $\pi$ can be factored into two parts, that is, 
\begin{equation*}
    \pi(s',\epsilon'|s, \epsilon, \alpha) = \mathfrak{p} (s'|s,\alpha)\mathfrak{q}(\epsilon'|s').
\end{equation*}
where $\mathfrak{p}$ and $\mathfrak{q}$ are probability density functions. 
This is to say that 
the transition to the next observable state $s'$ is determined only by the current observable state $s$ and action $\alpha$; the next unobservable rewards $\epsilon'$ are determined only by the next observable state $s'$.
There are no serial dependencies between unobservable rewards across decision stages. This assumption allows us to work in the reduced observable state space by taking expectation of the Bellman equation in the full state space.  

The holding time for an action $\alpha$ is deterministic and denoted $t_\alpha$. The reward of taking action $\alpha$ at state $(s, \epsilon)$ has two additively separable parts based on the observable and unobservable states respectively, that is, $r_\alpha(s) + \epsilon_\alpha(s)$. An action's rewards, both observable and unobservable, occur instantaneously after an action is taken, and there is no other reward that accumulates at a constant rate during the holding time, as is usually assumed for a general SMDP.

A Markovian policy at stage $l$ maps the current state, without regard to the history that led to that state, to an action. Formally, it is a mapping
$\mu_l:(s,\epsilon)\mapsto \alpha$. 
A stationary policy $\mu$ is a Markovian policy that does not depend on the decision stage, that is,
$\mu_l=\mu,~\forall l$. An agent chooses a policy to maximize the long-run expected discounted return, where the discount rate per unit time is $\beta>0$. The long-run expected discounted return from an initial state $(s_0,\epsilon_0)$ is referred to as the \emph{state value function} and is defined as
\begin{equation*}
V(s_0,\epsilon_0)
\coloneqq
\sup_{\Pi}\;
\mathbb{E}\left[
\sum_{l=0}^{\infty}
e^{-\beta \sum_{e=0}^{l-1} t_{\mu_e}}
\bigl(r_{\mu_l}(s_l)+\epsilon_{\mu_l}(s_l)\bigr)
\,\middle|\, s_0,\epsilon_0
\right],
\end{equation*}
where $\Pi=\{\mu_0,\mu_1,\ldots\}$ denotes the policy sequence, with $\mu_l(s_l,\epsilon_l)\in A(s_l)$ for each stage $l$. We adopt the convention that $\sum_{e=0}^{-1} t_{\mu_e}=0$. The expectation is taken with respect to the stochastic evolution of the system under policy $\Pi$, so that the state sequence $\{(s_1,\epsilon_1),(s_2,\epsilon_2),\ldots\}$ and the induced action sequence $\{\mu_1(s_1,\epsilon_1),\mu_2(s_2,\epsilon_2),\ldots\}$ are random.

We need Assumptions~\ref{a:continuous}, \ref{a:boundedreturn} and \ref{a:upperbound} below for the existence of a stationary optimal policy in the full state space. They ensure that all discount factors are strictly smaller than 1, the regularity of probability density functions and the finiteness of the state value function. These are fairly mild assumptions and can be safely regarded as reasonable in the application context of this paper.

\begin{assumption}[Continuous density functions]
For every \(s \in \mathcal{S}\), the density function \(\mathfrak{q}(\epsilon \mid s)\) is continuous in \(\epsilon\). For every \(s \in \mathcal{S}\) and every \(\alpha \in A(s)\), the density function \(\mathfrak{p}(s' \mid s,\alpha)\) is continuous in \(s'\).
\label{a:continuous}
\end{assumption}

\begin{assumption}[Finite discounted total return]
For all \(s\) and \(\epsilon\),
\begin{equation*}
    \sum_{l=0}^{\infty} \delta^k R_l(s,\epsilon) < \infty,
\end{equation*}
where
\begin{equation*}
    R_0(s,\epsilon) = \max_{\alpha} \bigl|r_{\alpha}(s) + \epsilon_{\alpha}\bigr|,
\end{equation*}
and
\begin{equation*}
    R_{l+1}(s,\epsilon)
    = \max_{\alpha} \int_{s'} \int_{\epsilon'} R_l(s',\epsilon') \, \mathfrak{q}(d\epsilon' \mid s') \, \mathfrak{p}(ds' \mid s,\alpha),
    \qquad l = 0,1,\dots
\end{equation*}
\label{a:boundedreturn}
\end{assumption}
Note that $R_l(s,\epsilon)$ represents an upper bound on the expected immediate reward at stage $l$ with an initial state $(s, \epsilon)$.

\begin{assumption}[Discount factors strictly less than 1]
$\delta \coloneqq \sup_{\alpha} e^{-\beta t_{\alpha}} < 1.$
\label{a:upperbound}
\end{assumption}

\begin{theorem}\label{thm:optimalpolicy}
There exists a stationary optimal policy \(\mu\). Moreover, the optimal decision rule is deterministic and Markovian, and is characterized by the Bellman equation
\begin{equation}
    V(s,\epsilon)
    =
    \max_{\alpha \in A(s)}
    \left[
        r_{\alpha}(s) + \epsilon_{\alpha}(s) + e^{-\beta t_{\alpha}} E_{\alpha}V(s)
    \right],
    \label{eq:originalbellman}
\end{equation}
with the optimal policy given by
\begin{equation*}
    \mu(s,\epsilon)
    =
    \arg\max_{\alpha \in A(s)}
    \left[
        r_{\alpha}(s) + \epsilon_{\alpha}(s) + e^{-\beta t_{\alpha}} E_{\alpha}V(s)
    \right].
\end{equation*}
Here, \(E_{\alpha}V(s)\) denotes the conditional expectation
\begin{equation*}
    E_{\alpha}V(s)
    \coloneqq
    \int_{s'} \int_{\epsilon'}
    V(s',\epsilon')\,\mathfrak{q}(d\epsilon' \mid s')\,\mathfrak{p}(ds' \mid s,\alpha).
\end{equation*}
\end{theorem}

Theorem~\ref{thm:optimalpolicy} can be viewed as a specialization of Theorem~3.3 in \cite{BhattacharyaMajumdar1989} in two respects: (1) it considers an instantaneous reward \(r_\alpha(s)+\epsilon_\alpha\), rather than a reward rate that accumulates over time, and (2) it assumes a deterministic holding time \(t_\alpha\), rather than a probabilistic holding time with a general distribution. These specializations modify the factors multiplying the reward term and the expected continuation value on the right-hand side of the Bellman equation \eqref{eq:originalbellman}. Theorem~3.1 of \cite{Rust1988} further specializes Theorem~3.3 of \cite{BhattacharyaMajumdar1989} to the discrete-time MDP setting, in which the holding time under every state-action pair is exactly one unit. Therefore, our result is more general than that of \cite{Rust1988}.

We then establish a contraction mapping on the observable state space $\mathcal{S}$ instead of the full state space.  The difficulty of the latter comes from the potentially unbounded $\epsilon$ \citep{Rust1988}. 

\begin{definition}
The state value function on the observable state space is defined by
\begin{equation*}
    v(s) \coloneqq \int_{\epsilon} V(s,\epsilon)\,\mathfrak{q}(d\epsilon \mid s),
    \qquad \forall s \in \mathcal{S}.
\end{equation*}
\end{definition}
\begin{definition}
Let \(\mathcal{C}\) denote the set of all continuous functions on \(\mathcal{S}\). Define \(\mathcal{B} \subset \mathcal{C}\) to be the Banach space of functions that are bounded under the supremum norm, where for any function \(v : \mathcal{S} \to \mathbb{R}\), \(\|v\|_{\infty} \coloneqq \sup_{s \in \mathcal{S}} |v(s)|\).
\end{definition}

For convenience of exposition, we assume that the action set has the same cardinality at every state, that is, \(|A(s)| = D\) for all \(s \in \mathcal{S}\). Any smaller action set can be padded with dummy actions to satisfy this assumption.

We next introduce an additional assumption.

\begin{assumption}[Bounded observable reward]
We assume that for each \(\alpha\in\{1,2,\dots,D\}\), the observable reward function satisfies \(r_{\alpha}(\cdot) \in \mathcal{B}\). %
\label{a:rinB}
\end{assumption}

\begin{definition}
Define the mapping \(\Lambda : \mathcal{B} \to \mathcal{C}\) by
\begin{equation}
    \Lambda(v)(s)
    =
    G\bigl(
    r_1(s) + e^{-\beta t_1} E_1 v(s), \dots, r_D(s) + e^{-\beta t_D} E_D v(s)
    \bigr),
    \qquad \forall s \in \mathcal{S},
    \label{eq:contraction_mapping}
\end{equation}
where \(E_{\alpha}v(s)\) is the shorthand notation
\begin{equation}
    E_{\alpha}v(s)
    :=
    \int_{s'} v(s') \,\mathfrak{p}(ds' \mid s,\alpha).
    \label{eq:shorthandE_av}
\end{equation}
\end{definition}

Note that \(E_{\alpha}v \in \mathcal{B}\) for any \(v \in \mathcal{B}\), by the continuity of the transition density \(\mathfrak{p}(ds' \mid s,\alpha)\). The definition does not claim that \(\Lambda(v) \in \mathcal{B}\); it only defines the mapping \(\Lambda\) on the domain \(\mathcal{B}\).

Our main result is stated in the following theorem.

\begin{theorem}
The value function in Theorem~\ref{thm:optimalpolicy} admits the representation
\begin{equation}
    V(s,\epsilon)
    =
    \max_{\alpha \in A(s)}
    \left[
        r_{\alpha}(s) + e^{-\beta t_{\alpha}} E_{\alpha}v(s) + \epsilon_{\alpha}(s)
    \right],
    \label{eq:rewriteV}
\end{equation}
where \(v : \mathcal{S} \to \mathbb{R}\) is the unique fixed point of the contraction mapping \(\Lambda : \mathcal{B} \to \mathcal{B}\) defined in \eqref{eq:contraction_mapping}.
\label{thm:contraction_statevaluefunction}
\end{theorem}

\begin{proof}{Proof.}
The proof follows closely the argument for Theorem~1 in \cite{RustTraubHozniakowski2002}. Equation~\eqref{eq:rewriteV} is obtained directly from \eqref{eq:originalbellman} by substituting the definition of \(E_{\alpha}v(s)\) from \eqref{eq:shorthandE_av}.

The social surplus function \(G : \mathbb{R}^D \to \mathbb{R}\) is quasi-linear \citep{RustTraubHozniakowski2002} and satisfies the following properties:
\begin{enumerate}
    \item \(G\) is nondecreasing in each argument; and
    \item for every \(a \in \mathbb{R}\), \(G(\eta + a\mathbf{e}) = G(\eta) + a\), where \(\mathbf{e} = [1,\dots,1] \in \mathbb{R}^D\).
\end{enumerate}

Note that these properties imply that \(G\) is Lipschitz continuous with respect to the \(\ell_\infty\) norm, and therefore continuous. We first show that \(\Lambda(v) \in \mathcal{B}\) for every \(v \in \mathcal{B}\). %
For any \(v \in \mathcal{B}\) and any \(\alpha \in \{1,\dots,D\}\), $\|E_\alpha v\|_\infty \le \|v\|_\infty$ due to the triangle inequality,
and, by the definition of \(\delta\), 
$e^{-\beta t_\alpha} \le \delta < 1$.
Hence, for every \(s \in \mathcal{S}\),
\[
-\delta \|v\|_\infty
\;\le\;
e^{-\beta t_\alpha} E_\alpha v(s)
\;\le\;
\delta \|v\|_\infty,
\qquad \forall \alpha \in \{1,\dots,D\}.
\]
Using the monotonicity and quasi-linearity of \(G\), we obtain
\begin{equation*}
    G\bigl(r_1(s),\dots,r_D(s)\bigr) - \delta \|v\|_\infty
    \;\le\;
    \Lambda(v)(s)
    \;\le\;
    G\bigl(r_1(s),\dots,r_D(s)\bigr) + \delta \|v\|_\infty,
    \qquad \forall s \in \mathcal{S}.
\end{equation*}
Therefore,
\begin{equation*}
    \|\Lambda(v) - G\bigl(r_1(s),\dots,r_D(s)\bigr)\|_\infty \le \delta \|v\|_\infty.
\end{equation*}
By the triangle inequality,
\begin{equation*}
    \|\Lambda(v)\|_\infty
    \le
    \|G\bigl(r_1(s),\dots,r_D(s)\bigr)\|_\infty + \delta \|v\|_\infty
    <
    +\infty,
\end{equation*}
where finiteness follows from Assumption~\ref{a:rinB}. Thus, \(\Lambda\) maps \(\mathcal{B}\) into itself.

Next we show that \(\Lambda\) is a contraction. For any \(v,w \in \mathcal{B}\), any \(s \in \mathcal{S}\), and any \(\alpha \in \{1,\dots,D\}\),
\[
e^{-\beta t_\alpha} \bigl|E_\alpha(v-w)(s)\bigr|
\le
\delta \|E_\alpha(v-w)\|_\infty
\le
\delta \|v-w\|_\infty.
\]
Hence, using again the monotonicity and quasi-linearity of \(G\),
\begin{align*}
    \Lambda(v)(s)
    &=
    G\Bigl(
        r_1(s) + e^{-\beta t_1}E_1 w(s) + e^{-\beta t_1}E_1(v-w)(s),
        \dots,
        r_D(s) + e^{-\beta t_D}E_D w(s) + e^{-\beta t_D}E_D(v-w)(s)
    \Bigr)
    \nonumber\\
    &\le
    G\Bigl(
        r_1(s) + e^{-\beta t_1}E_1 w(s) + \delta\|v-w\|_\infty,
        \dots,
        r_D(s) + e^{-\beta t_D}E_D w(s) + \delta\|v-w\|_\infty
    \Bigr)
    \nonumber\\
    &=
    \Lambda(w)(s) + \delta\|v-w\|_\infty.
\end{align*}
Swapping the roles of \(v\) and \(w\) yields
\begin{equation*}
    \Lambda(w)(s) \le \Lambda(v)(s) + \delta\|v-w\|_\infty.
\end{equation*}
Therefore,
\begin{equation*}
    |\Lambda(v)(s) - \Lambda(w)(s)|
    \le
    \delta\|v-w\|_\infty,
    \qquad \forall s \in \mathcal{S},
\end{equation*}
and taking the supremum over \(s\) gives
\begin{equation*}
    \|\Lambda(v)-\Lambda(w)\|_\infty \le \delta \|v-w\|_\infty.
\end{equation*}
Since \(\delta < 1\), \(\Lambda\) is a contraction on \((\mathcal{B},\|\cdot\|_\infty)\). Because \((\mathcal{B},\|\cdot\|_\infty)\) is a Banach space, Banach's fixed-point theorem implies that the fixed-point equation
$v = \Lambda(v)$ has a unique solution in \(\mathcal{B}\).\hfill$\square$
\end{proof}

We further define the state-action value function \(\eta\) in the observable state space as
\begin{equation}
    \eta_\alpha(s)
    =
    r_\alpha(s) + e^{-\beta t_\alpha} E_\alpha v(s),
    \qquad
    \forall s \in \mathcal{S},\ \forall \alpha \in A(s),
    \notag%
\end{equation}
which allows the state value function in the observable space to be written more compactly as
\begin{equation}
    v(s) = G\bigl(\eta_1(s), \dots, \eta_D(s)\bigr),
    \qquad \forall s \in \mathcal{S}.
    \notag%
\end{equation}

\begin{proof}{Proof of \Cref{prop:continuousVF}.}

Note that the quantity \(\eta_{\alpha}(s)\) above generalizes the action-value terms \(\mathbf{z}\) and \(\mathbf{w}\) used in the main body of the paper.

\emph{Existence and uniqueness.}
Theorem~\ref{thm:contraction_statevaluefunction} establishes that the mapping \(\Lambda:\mathcal{B}\to\mathcal{B}\) admits a unique fixed point, where \(\mathcal{B}\) denotes the space of continuous functions on the observable state space that are bounded under the infinity norm. Since the observable state space is countable and discrete, every function defined on it is continuous. Hence, the state value function belongs to \(\mathcal{B}\), and a unique bounded solution exists.

\emph{Continuity.}
Continuity follows from Theorem~1A.4, the parametric contraction mapping principle, in \cite{DontchevRockafellar2014}. The first condition of Theorem~1A.4, namely that the mapping is a contraction, has already been established above. We therefore verify the second condition, namely Lipschitz continuity of the mapping with respect to \(m_a \in [0,1)\) and \(t_a \in [\hat t,\check t]\), as defined in \eqref{eq:actionvalueempty} and \eqref{eq:actionvaluehired}.

The mapping is linear in \(m_a\), and its coefficient is a linear combination of bounded discounted future state values, which is itself bounded. It follows immediately that the mapping is Lipschitz continuous in \(m_a\). For \(t_a\), the first component of the mapping is given by the link cost function \(c^E(t_a)\) or \(c^H(t_a)\), both of which are assumed Lipschitz continuous in \(t_a\). The second component depends on \(t_a\) through the discount factor \(e^{-\beta t_a}\), which multiplies the expected future state value. Its derivative with respect to \(t_a\) is bounded because the discounted state value is bounded. Therefore, the mapping is also Lipschitz continuous in \(t_a\).\hfill$\square$

\end{proof}

\subsection{On the Continuity of the Mapping $(\mathbf{p},\mathbf{q},\mathbf{\xi}, \mathbf{t}, \mathbf{m})\mapsto (\mathbf{x},\mathbf{y})$}
\label{subsec:ctmc}

We provide the lemma for proving the continuity of the mapping $(\mathbf{p},\mathbf{q},\boldsymbol{\xi},\mathbf{t}, \mathbf{m})\mapsto (\mathbf{x},\mathbf{y})$ in the main text. %

\begin{lemma}
\label{lemma:uniqueness}
The mass distribution $(\mathbf{x},\mathbf{y})$ defined by equations~\eqref{eq:totalmass}, \eqref{eq:outgoingemptyflow}, \eqref{eq:outgoinghiredflow}, and \eqref{eq:outgoinghiredflow_self} is unique given $(\mathbf{p},\mathbf{q},\boldsymbol{\xi},\mathbf{t}, \mathbf{m})$. %
\end{lemma}
\begin{proof}{Proof.}
We construct a continuous-time Markov chain (CTMC) whose full balance conditions are given by
\eqref{eq:outgoingemptyflow}, \eqref{eq:outgoinghiredflow}, and \eqref{eq:outgoinghiredflow_self}.
For each state $a \in \mathcal{A}$, let $x_a$ denote its steady-state probability. For each state $(a,i)$ with $i \in \mathcal{N}$ and $a \notin \mathcal{A}_i^{+}$, let $y_a^{i}$ denote its steady-state probability. We define the CTMC state space as
$\bar{\mathcal{S}} \;\coloneqq\; \mathcal{A}\ \cup\ \{(a,i)\,:\, a \in \mathcal{A}\setminus \mathcal{A}_i^{+},\ i \in \mathcal{N}\}$.

For any node $i \in \mathcal{N}$, the transition rate from $a' \in \mathcal{A}_i^{-}$ to $a \in \mathcal{A}_i^{+}$ equals $\bigl(1 - m_{a'} + \sum_{d\in\mathcal{N}} m_{a'} n_i^d(1-\xi_i^d)\bigr)\, p_a / t_{a'}$, and the transition rate from $a' \in \mathcal{A}_i^{-}$ to $(a,d)$ with $a \in \mathcal{A}_i^{+}$ and $d \neq i$ equals $m_{a'} n_i^d \xi_i^d\, q_a^d / t_{a'}$; moreover, the transition rate from $(a',i)$ with $a' \in \mathcal{A}_i^{-}$ to $a \in \mathcal{A}_i^{+}$ is $1$, and the transition rate from $(a',d)$ with $a' \in \mathcal{A}_i^{-}$ and $d \neq i$ to $(a,d)$ with $a \in \mathcal{A}_i^{+}$ and $d \neq i$ equals $q_a^d / t_{a'}$. Collecting these rates yields a generator matrix $Q$, with diagonal entries defined by $Q_{ss} = -\sum_{s' \in \bar{\mathcal{S}},\, s' \neq s} Q_{ss'}$ for all $s \in \bar{\mathcal{S}}$. Note that the full state space $\bar{\mathcal{S}}$ may include states that are transient for certain matching-probability vectors $m=(m_a)$; we therefore define the recurrent class by isolating transient states from $\bar{\mathcal{S}}$ as follows:

\begin{itemize}
    \item If \( m_a = 0 \) for all \( a\in \mathcal{A} \), all states in $\{(a,i)\,:\, a \in \mathcal{A}\setminus \mathcal{A}_i^{+},\ i \in \mathcal{N}\}$ are transient as no vehicles can be hired.
    \item If \( \sum_{a\in\mathcal{A}_i^-}m_a > 0 \) for only one \( i\in\mathcal{N} \), and \( \sum_{a\in\mathcal{A}_j^-} m_a = 0 \) for all \( j \ne i \), all hired states of the form \( (a, i) \) are transient, as cars cannot be hired with destination \( i \).
    \item If \( \sum_{i \in \mathcal{N}} \sum_{a \in \mathcal{A}_i^-} m_a n_i^d = 0 \) for some \( d \in \mathcal{N} \), all hired states \( (a, d) \) are transient, since no cars can be hired with destination \( d \).
\end{itemize}

Any state in the remaining set can reach any other state in that set in finitely many transitions. 
We now show how such transitions arise:
\begin{enumerate}
    \item For any node $i\in\mathcal{N}$, the transition rate from an empty state $a'\in\mathcal{A}_i^{-}$ to an empty state $a\in\mathcal{A}_i^{+}$ is
    $(1-m_{a'}+\sum_{d\in\mathcal{N}} m_{a'}\,n_i^d(1-\xi_i^d))p_a/t_{a'}$,
    which is strictly positive due to $\xi_i^d \in (0,1)$ and $p_a \in (0,1)$. Hence, all empty states communicate through positive-rate transitions and thus form a recurrent class.

    \item Suppose $m_a>0$ for some $a\in\mathcal{A}$. Then: (i) transitions among empty states remain feasible and occur at strictly positive rates; (ii) for any empty state $a'\in\mathcal{A}_i^{-}$, if $n_i^d>0$ for some $d\neq i$, there is a positive-rate transition to a hired state $(a,d)$ with $a\in\mathcal{A}_i^{+}$; and (iii) from any hired state $(a',d)$ with $a'\in\mathcal{A}_i^{-}$, if $d\neq i$ then the transition to $(a,d)$ (with $a\in\mathcal{A}_i^{+}$) occurs at rate $q_a^d/t_{a'}>0$, whereas if $d=i$ the vehicle can return to an empty state $a\in\mathcal{A}_i^{+}$ with positive rate.
\end{enumerate}
It follows that the set consisting of all empty states together with all hired states reachable from them is closed and irreducible, and hence forms a single recurrent class of the CTMC for the given $m$. Moreover, this class is unique: from any hired state $(a,i)$ there is a positive-rate path back to some empty state $a'\in\mathcal{A}_i^{+}$. Since the state space is finite and the CTMC has exactly one recurrent class (i.e., it is a unichain), the steady-state distribution satisfying the balance equations is unique; together with the mass constraint \eqref{eq:totalmass}, this also uniquely pins down the steady-state mass distribution.\hfill$\square$
\end{proof}

\section{Proofs in \Cref{sec:special}}

\begin{proof}{Proof of \Cref{prop:single_cycle_uniqueness_convergence}.}
In a directed cycle network, each node has exactly one outgoing link, so the link choice is deterministic. The equilibrium condition in Definition~\ref{def:MTER} reduces to 
$
\frac{u_a}{t_a(u_a)} = \rho,~\forall a \in \mathcal{A},
$
for some unknown constant flow $\rho$, together with the total mass constraint $\sum_a u_a = M$. It follows that at equilibrium $u_a = \rho t_a(u_a) \ge \frac{M}{\sum_a t_a(M)} t_a(0) \ge \frac{M\hat t}{|\mathcal{N}|\check t}, \forall a$. Let $\tilde{\epsilon}:=\frac{M\hat t}{|\mathcal{N}|\check t}$. The fixed-point problem characterizing the equilibrium reduces to 
$$
u_a = \Phi_a(\mathbf{u}) := \frac{M\, t_a\!\left(u_a\right)}{\sum_{a'\in\mathcal{A}} t_{a'}\!\left(u_{a'}\right)}, \qquad \forall a \in \mathcal{A},
$$
on the feasible domain $\Delta_M = \{u_a \ge \tilde\epsilon, \sum_a u_a = M\}$.  It can be seen that for any $\mathbf{u} \in \Delta_M$, $\tilde\epsilon \leq \Phi_a(\mathbf{u}) \leq M$ and $\sum_a \Phi_a(\mathbf{u})=M$, and thus the mapping $\Phi$ is from $\Delta_M$ onto itself. 

Since \(t_a(u_a)\) is continuously differentiable, \(u_a\) ranges over a bounded set, and \(u_a/t_a(u_a)\) is strictly increasing, we have 
$$\frac{d}{d u_a}\left(\frac{u_a}{t_a(u_a)}\right)=\frac{t_a(u_a)-u_a t_a'(u_a)}{t_a^2(u_a)}>0,$$ 
which implies \(0 \le \frac{u_a t_a'(u_a)}{t_a(u_a)}<1\) for all \(a\) and feasible \(u_a\). Moreover, by the compactness of the feasible domain and $\frac{u_a t_a'(u_a)}{t_a(u_a)}$ being continuous, there exists a constant \(\kappa<1\) such that \(\frac{u_a t_a'(u_a)}{t_a(u_a)}\le \kappa<1\) for all \(a\) and feasible \(u_a\), where \(\kappa:=\sup_{a,\,u_a}\frac{u_a t_a'(u_a)}{t_a(u_a)}\).

For $\varsigma\in[\log \tilde\epsilon, \log M]$, define $g_a(\varsigma) \coloneqq \log t_a(e^{\varsigma})$ (so $e^{\varsigma}$ is a  mass argument on $\Delta_M$). Then  
$$g_a'(\varsigma) = \frac{e^{\varsigma} t_a'(e^{\varsigma})}{t_a(e^{\varsigma})} \in [0, \kappa].$$

Fix $a\in\mathcal{A}$ and let $u_a,\tilde{u}_a>0$ denote two arbitrary positive values of aggregated link mass on link $a$ (so that the logarithms below are well defined). Then  
$$|\log t_a(u_a) - \log t_a(\tilde{u}_a)| \le \kappa |\log u_a - \log \tilde{u}_a|,$$
which implies  
$$-\left|\log \left(\frac{u_a}{\tilde{u}_a}\right)^{\kappa}\right| \le \log \frac{t_a(u_a)}{t_a(\tilde{u}_a)} \le \left|\log \left(\frac{u_a}{\tilde{u}_a}\right)^{\kappa}\right| \implies \min \left\{ \left( \frac{u_a}{\tilde{u}_a} \right)^{\kappa}, \left( \frac{u_a}{\tilde{u}_a} \right)^{-\kappa} \right\} \le \frac{t_a(u_a)}{t_a(\tilde{u}_a)} \le \max \left\{ \left( \frac{u_a}{\tilde{u}_a} \right)^{\kappa}, \left( \frac{u_a}{\tilde{u}_a} \right)^{-\kappa} \right\}.$$
Since $t_a(\cdot)$ is strictly increasing, these inequalities can be simplified to  
$$\min \left\{ \left( \frac{u_a}{\tilde{u}_a} \right)^{\kappa}, 1 \right\} \le \frac{t_a(u_a)}{t_a(\tilde{u}_a)} \le \max \left\{ \left( \frac{u_a}{\tilde{u}_a} \right)^{\kappa}, 1 \right\}.$$
For any two vectors $\mathbf{u}, \tilde{\mathbf{u}} \in \Delta_M$, we then compute  
$$\frac{\Phi_a(\mathbf{u})/\Phi_a(\tilde{\mathbf{u}})}{\Phi_{a'}(\mathbf{u})/\Phi_{a'}(\tilde{\mathbf{u}})} = \frac{t_a(u_a)/t_a(\tilde{u}_a)}{t_{a'}(u_{a'})/t_{a'}(\tilde{u}_{a'})} \quad \implies \quad \frac{\max_a \Phi_a(\mathbf{u})/\Phi_a(\tilde{\mathbf{u}})}{\min_a \Phi_a(\mathbf{u})/\Phi_a(\tilde{\mathbf{u}})} \le \frac{\max_a \left\{ \left( \frac{u_a}{\tilde{u}_a} \right)^{\kappa}, 1 \right\}}{\min_a \left\{ \left( \frac{u_a}{\tilde{u}_a} \right)^{\kappa}, 1 \right\}} = \frac{\max_a \left( \frac{u_a}{\tilde{u}_a} \right)^{\kappa}}{\min_a \left( \frac{u_a}{\tilde{u}_a} \right)^{\kappa}},$$
because for any $\mathbf{u}$ and $\tilde{\mathbf{u}}$, there is always at least one $a$ such that $u_a \ge \tilde{u}_a$ and at least one $a$ such that $u_a \le \tilde{u}_a$.

Define the spread  
$$\mathcal{D}(\mathbf{u}, \tilde{\mathbf{u}}) := \log \left( \frac{\max_a u_a / \tilde{u}_a}{\min_a u_a / \tilde{u}_a} \right).$$
Then we have  
$$\mathcal{D}(\Phi(\mathbf{u}), \Phi(\tilde{\mathbf{u}})) \le \kappa \, \mathcal{D}(\mathbf{u}, \tilde{\mathbf{u}}).$$
Hence, each application of $\Phi$ shrinks the spread by at least a factor of $\kappa < 1$. 

It is straightforward to verify that the spread is a metric on the complete space $\Delta_M$. We thus establish that $\Phi$ is a contraction from the non-empty, complete metric space $\Delta_M$ onto itself. According to the Banach fixed-point theorem, there exists a unique fixed point $\mathbf{u}^*=\Phi(\mathbf{u}^*)$, which corresponds to the unique equilibrium total link mass vector, and the fixed-point iteration from any initial $\mathbf{u}_0 \in \Delta_M$ converges to $\mathbf{u}^*$. This completes the proof. \hfill$\square$
\end{proof}

\section{Additional Results in \Cref{sec:computational}}

\subsection{A Braess's Paradox for Ride-Hailing}
\label{sec:braess}

The network setup is illustrated in Figure~\ref{fig:braess_network}. Initially, the network consists of four nodes and five links, with node 0 serving as the sole origin (the sole incoming link (3, 0) has a positive passenger arrival rate), and node 3 the only destination. To examine the paradox, we add a new link $1\to2$ with effectively infinite capacity and zero travel time. Figure~\ref{fig:braess_results} shows the performance metrics before (solid line) and after (dashed line) adding the link  as functions of the exogenous driver pool size. Adding the link results in decreased average speeds shown in Figure~\ref{fig:braess_speed} and as a result the operating costs increase and profits drop shown in Figure~\ref{fig:braess_profit}. %
The rationale is similar to that in the original paradox: the new and fast link entices drivers to take a third route which results in additional congestion on the more congestable links $(0,1)$ and $(2,3)$. %
\begin{figure}[htbp]
  \centering
  \begin{tikzpicture}[
    >=Latex,
    node distance=26mm and 42mm,
    flowNode/.style={circle, draw=blue!70!black, fill=blue!70, text=white,
                     minimum size=10mm, inner sep=0pt},
    endNode/.style={rectangle, draw=blue!70!black, fill=blue!70, text=white,
                    minimum width=20mm, minimum height=9mm, inner sep=2pt},
    edge/.style={-Latex, line width=1.1pt, draw=black},
    bigedge/.style={-Latex, line width=1.6pt, draw=black},
    midedge/.style={-Latex, line width=1.0pt, draw=gray!60, dashed},
    edgelabel/.style={font=\small, align=center, fill=white, inner sep=1.5pt},
    fatEArrow/.style={
        single arrow, draw=blue!65!black, fill=blue!30, 
        minimum height=8mm, minimum width=8mm,
        single arrow head extend=0.5mm, 
        inner sep=1pt, anchor=east
    },
    fatWArrow/.style={
        single arrow, draw=blue!65!black, fill=blue!30, 
        minimum height=8mm, minimum width=8mm,
        single arrow head extend=0.5mm, 
        inner sep=1pt, anchor=west, rotate=180
    }
  ]

  \node[endNode] (s) {$0(\mathrm{START})$};
  \node[endNode, right=50mm of s] (t) {$3(\mathrm{END})$};
  \coordinate (mid) at ($(s)!0.5!(t)$);
  \node[flowNode, above=18mm of mid] (n1) {1};
  \node[flowNode, below=18mm of mid] (n2) {2};

  \draw[edge] (s) -- (n1)
    node[edgelabel, pos=0.50, xshift=2mm, yshift=1mm, above left]
      {$\bar{c}_a$=150 \\ $t_a(0)=0.009$};

  \draw[edge] (n1) -- (t)
    node[edgelabel, pos=0.50, xshift=-2mm, yshift=1mm, above right]
      {$\bar{c}_a=200$ \\ $t_a(0)=0.012$};

  \draw[edge] (s) -- (n2)
    node[edgelabel, pos=0.50, xshift=2mm, yshift=-1mm, below left]
      {$\bar{c}_a=200$ \\ $t_a(0)=0.012$};

  \draw[edge] (n2) -- (t)
    node[edgelabel, pos=0.50, xshift=-2mm, yshift=-1mm, below right]
      {$\bar{c}_a=200$ \\ $t_a(0)=0.009$};

  \draw[midedge] (n1) -- (n2)
    node[edgelabel, pos=0.52, right]
      {$\bar{c}_a=+\infty$\\ $t_a(0) \approx 0$};

  \coordinate (sb) at ($(s)+(0,-30mm)$);
  \coordinate (tb) at ($(t)+(0,-30mm)$);
  \draw[bigedge] (t) -- (tb) -- (sb) -- (s);
  \draw (tb) -- (sb)
    node[midway, below, yshift=-2mm, edgelabel]
      {$\bar{c}_a=200$\\ $t_a(0)=0.012$};

  \node[fatEArrow, xshift=8mm, yshift=-18mm] at (s.west) (arrArrow) {};
  \node[left=1mm of arrArrow, font=\small, align=center] {Passenger \\ Origin};

  \node[fatWArrow, xshift=-10mm] at (t.east) (destArrow) {};
  \node[right=10mm of destArrow, font=\small, align=center] {Passenger \\ Destination};

  \end{tikzpicture}
  \caption{A Braess's Paradox network.}
  \label{fig:braess_network}
\end{figure}
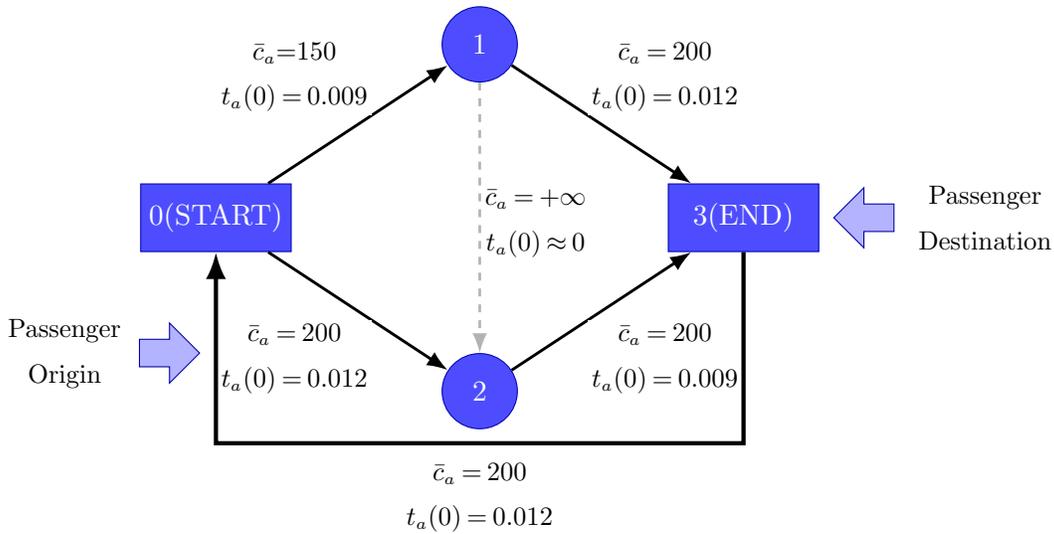

\begin{figure}[!h]
    \centering
    \begin{subfigure}{0.45\textwidth}
        \centering
        \includegraphics[height=5cm,keepaspectratio]{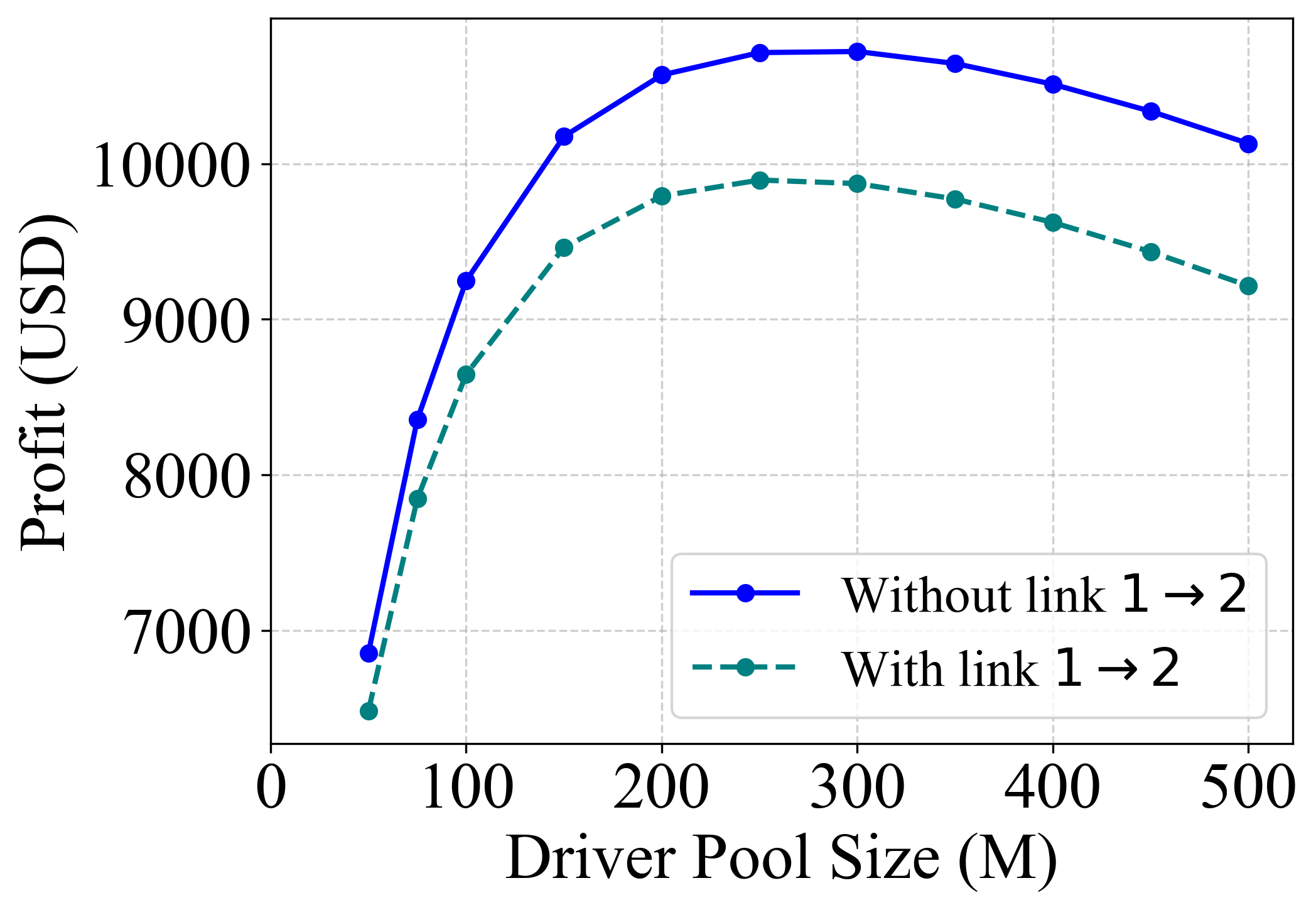}
        \caption{Profit}
        \label{fig:braess_profit}
    \end{subfigure}
    \hfill
    \begin{subfigure}{0.45\textwidth}
        \centering
        \includegraphics[height=5cm,keepaspectratio]{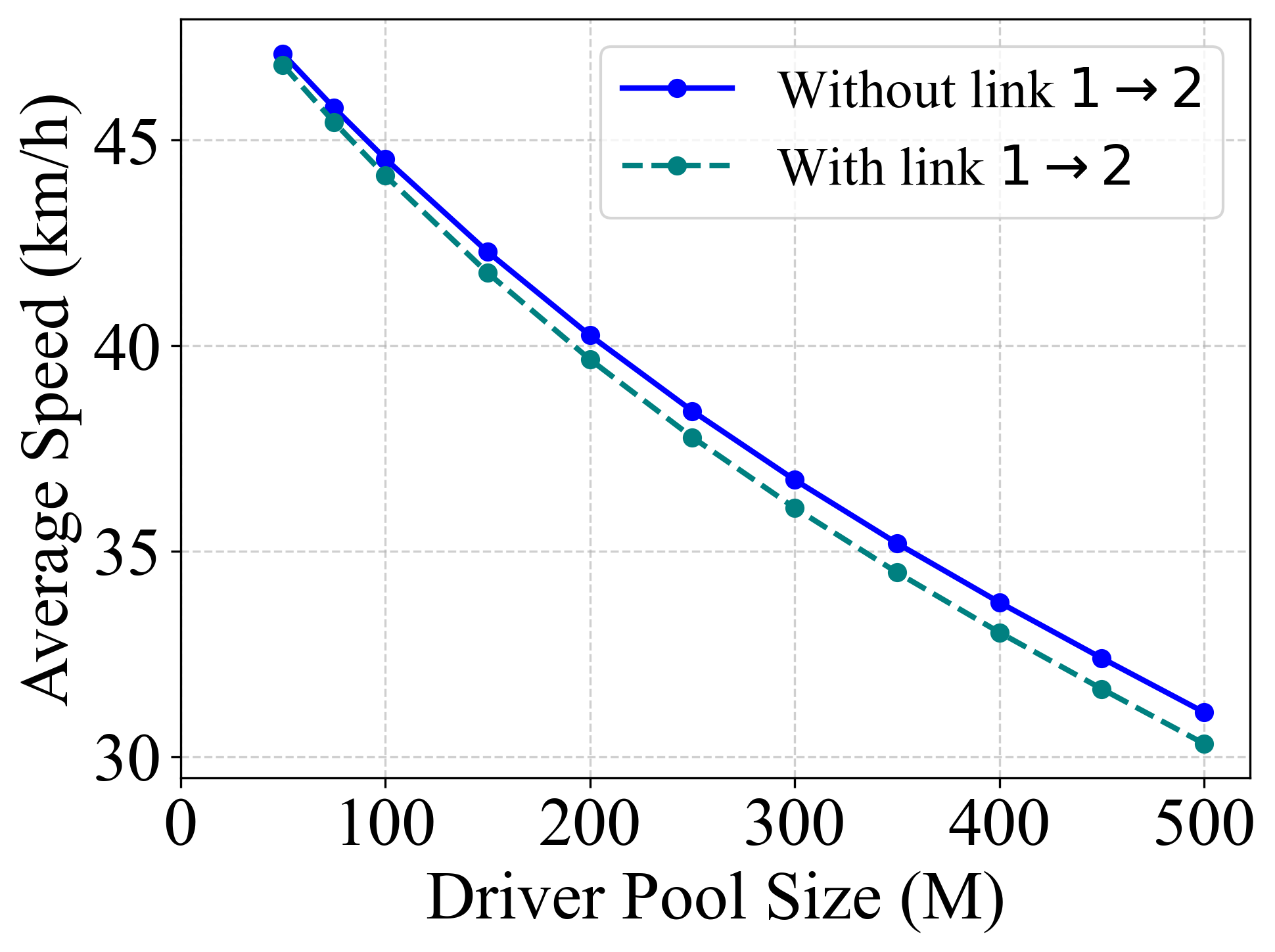}
        \caption{Average Speed}
        \label{fig:braess_speed}
    \end{subfigure}
    \caption{Comparison of performance metrics with and without the additional link.}
    \label{fig:braess_results}
\end{figure}

\subsection{Network Parameters for the Myopic Driver Ablation Study}
\label{app:myopic_network}

The distances and free-flow travel times between nodes are summarized in Table~\ref{tab:myopic_distance}. For example, the links between nodes 1 and 2, denoted as (1,2) and (2,1), have a length of 15~km and a free-flow time of 0.3 hour. The arrival rates on links connecting nodes 1, 5, 6, and 7 are 1,000 per hour, whereas the demand on all other links is 5,000 per hour. The destination probabilities by origin are detailed in Table~\ref{tab:myopic_odprob}. For instance, for passengers originating from node 1 (arriving via links (5,1) and (2,1)), 50\%  choose node 2 as their destination, while the remaining 50\%  head toward node 5. The driver pool size is 18,000.

\begin{table}[h]
    \small
    \centering
    \caption{The distance and free-flow travel time between nodes.}
    \label{tab:myopic_distance}
    \begin{tabular}{lcccccccc}
        \toprule
        Node Pair & (1,2) & (1,5) & (2,3) & (2,4) & (3,4) & (5,6) & (5,7) & (6,7) \\
        \midrule
        Distance (km) & 15 & 15 & 5 & 5 & 5 & 5 & 5 & 5 \\
        Free-flow time (hour) & 0.3 & 0.3 & 0.1 & 0.1 & 0.1 & 0.1 & 0.1 & 0.1 \\
        \bottomrule
    \end{tabular}
\end{table}

\begin{table}[h]
    \small
    \centering
    \caption{Destination probabilities by origin.}
    \label{tab:myopic_odprob}
    \begin{tabular}{lcccccccc}
        \toprule
        O-D & 1 & 2 & 3 & 4 & 5 & 6 & 7 \\
        \midrule
        1 & 0 & 0.5 & 0 & 0 & 0.5 & 0 & 0 \\
        2 & 0.1 & 0 & 0.9 & 0 & 0 & 0 & 0 \\
        3 & 0 & 0 & 0 & 1.0 & 0 & 0 & 0 \\
        4 & 0 & 1.0 & 0 & 0 & 0 & 0 & 0 \\
        5 & 0.1 & 0 & 0 & 0 & 0 & 0.9 & 0 \\
        6 & 0 & 0 & 0 & 0 & 0 & 0 & 1.0 \\
        7 & 0 & 0 & 0 & 0 & 1.0 & 0 & 0 \\
        \bottomrule
    \end{tabular}
\end{table}

\subsection{Sensitivity Analyses with Exogenous Driver Pool Sizes}
\label{subsec:sensitivity}
We systematically vary the value of one of the three key parameters, the driver pool size (default value 20{,}000), matching friction parameter $\gamma$ (default value 0.8), and discount rate $\beta$ (default value 0.1 per hour), and examine how system performance changes accordingly. Four system-wide metrics are calculated at the equilibrium  distribution of vehicle mass: the total profit (USD),  average speed (km per hour),  demand fulfillment proportion (\%) and the ratio of the number of vacant to hired vehicles. Note that the profit is up to a constant as we only subtract from revenue the operating cost. 

\begin{figure}[!h]
    \begin{subfigure}[b]{0.45\textwidth}
        \centering
        \includegraphics[width=\linewidth]{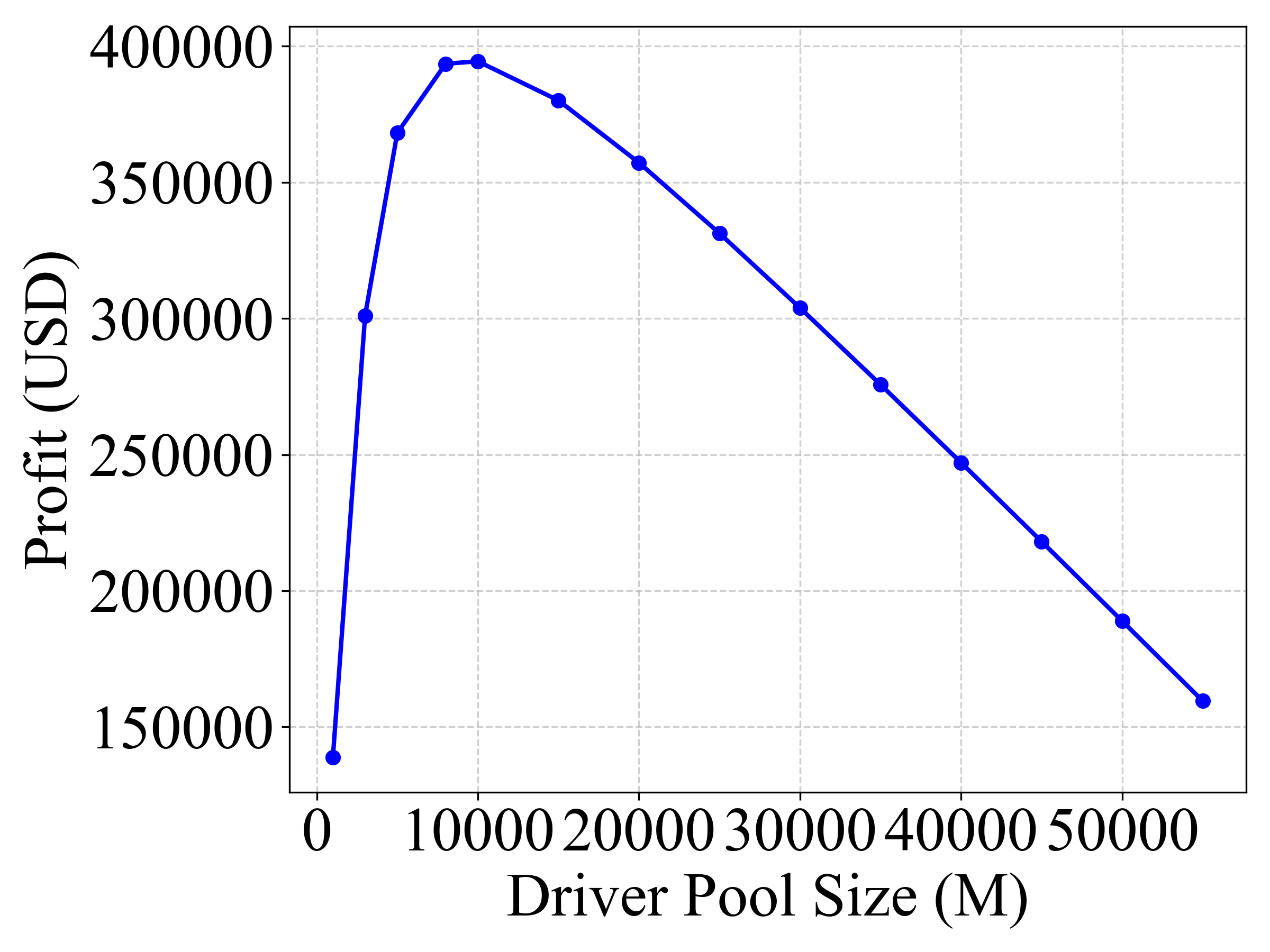}
        \caption{Profit}
        \label{fig:sen_fleetsize_profit}
    \end{subfigure}
    \begin{subfigure}[b]{0.45\textwidth}
        \centering
        \includegraphics[width=\linewidth]{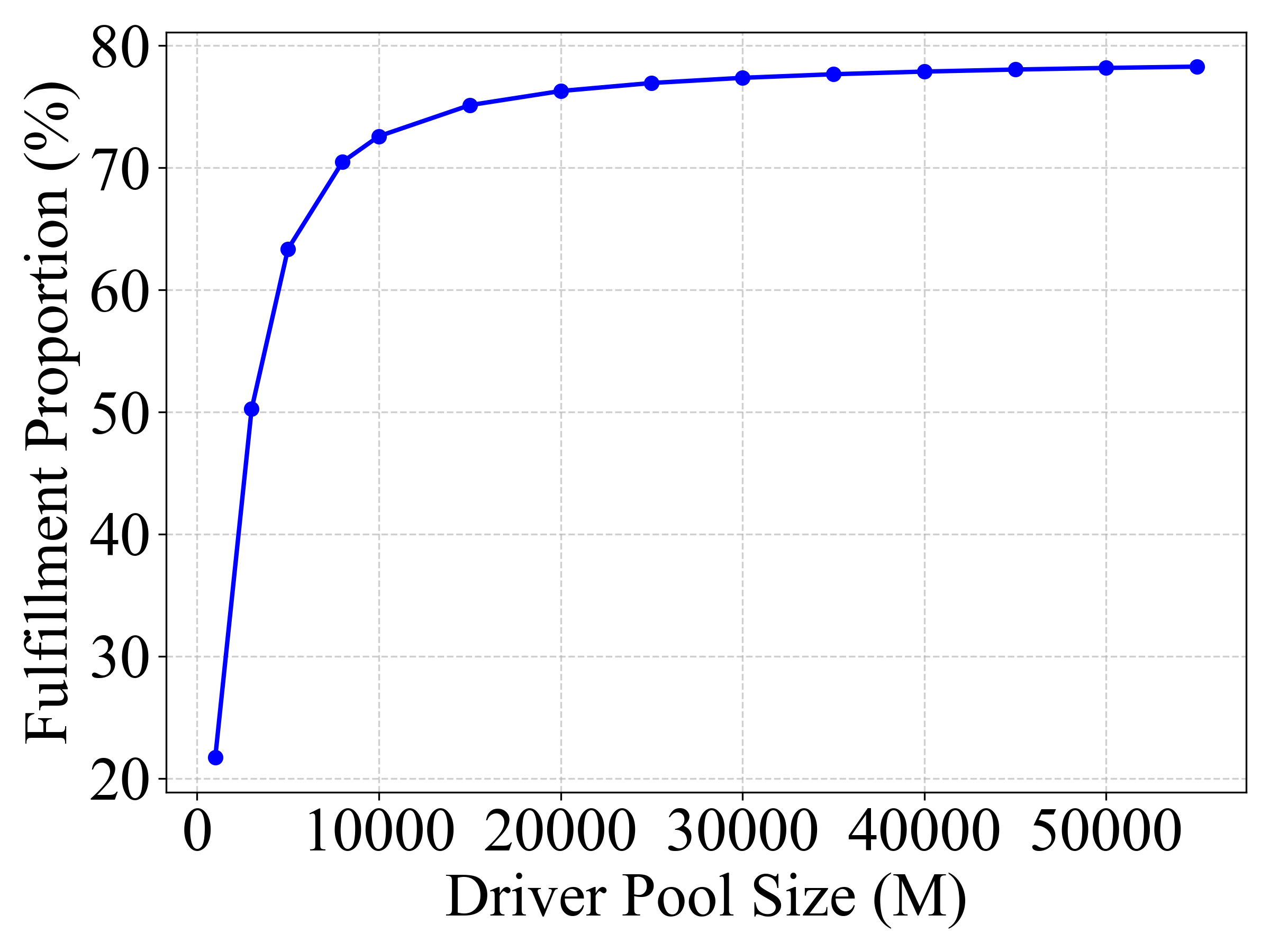}
        \caption{Fulfillment Proportion}
        \label{fig:sen_fleetsize_fulfillment}
    \end{subfigure}
    \hfill
    \begin{subfigure}[b]{0.45\textwidth}
        \centering
        \includegraphics[width=\linewidth]{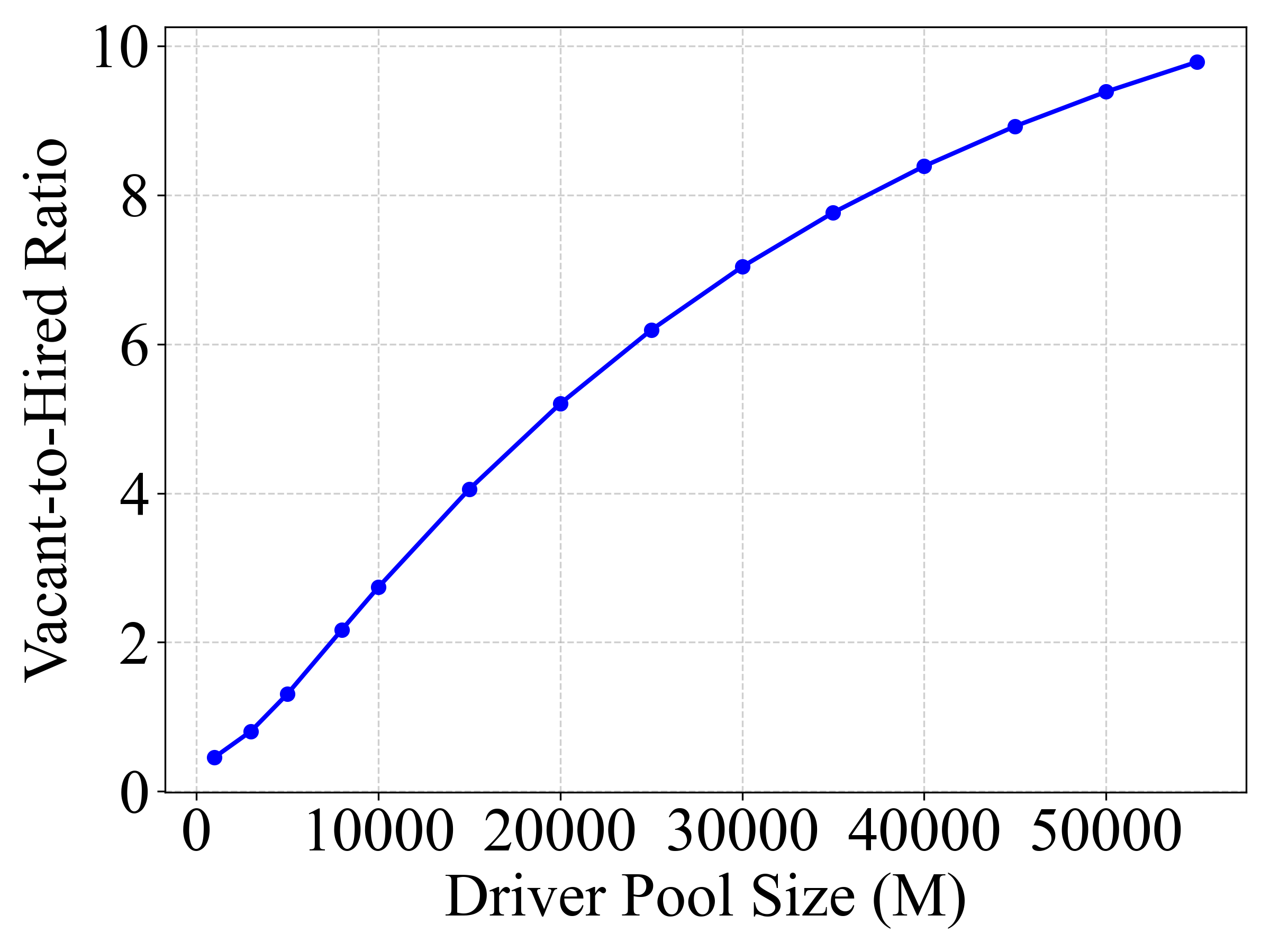}
        \caption{Vacant-to-hired Ratio}
        \label{fig:sen_fleetsize_vhratio}
    \end{subfigure}
    \hfill
    \begin{subfigure}[b]{0.45\textwidth}
        \centering
        \includegraphics[width=\linewidth]{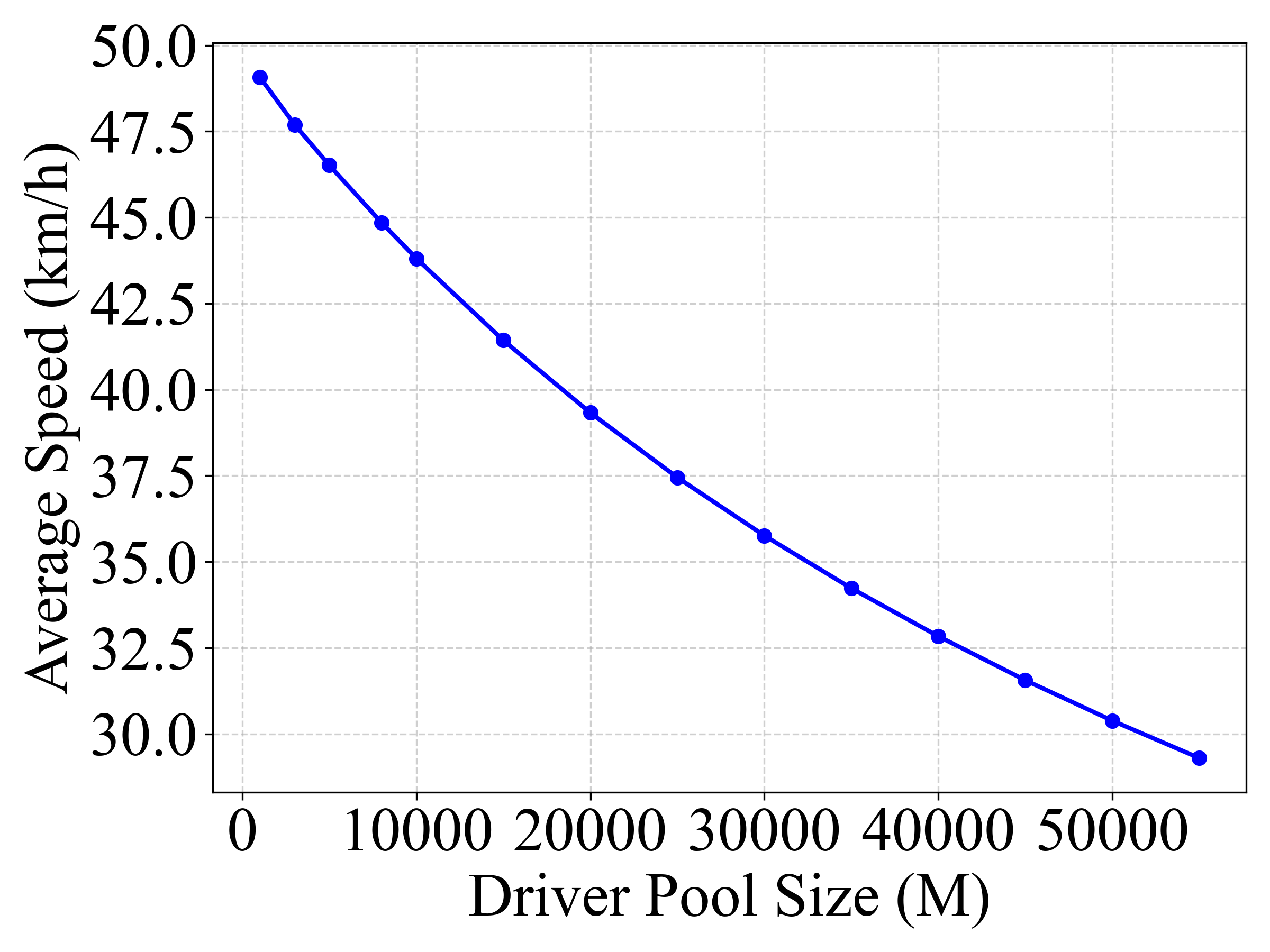}
        \caption{Average Speed}
        \label{fig:sen_fleetsize_speed}
    \end{subfigure}
    \caption{Sensitivity analyses w.r.t. the driver pool size on the Sioux Falls network.}
    \label{fig:sensitivity_analysis_siouxfalls_fleetsize}
\end{figure}

Figure~\ref{fig:sen_fleetsize_profit} shows that as the vehicle pool size increases, the total profit first rises and then declines. Initially, a larger number of vehicles improves demand fulfillment, leading to higher revenue. However, once the pool becomes sufficiently large, the fulfillment proportion reaches its upper bound, shown in Figure~\ref{fig:sen_fleetsize_fulfillment}. Specifically, as the vacant flow $f_a$ increases, the demand fulfillment proportion 
\begin{align*}\frac{m_a f_a}{\lambda_a}=\min\left(1,\frac{f_a  \left(1 - e^{-\frac{\gamma_a\lambda_a}{f_a}}\right)}{\lambda_a}\right)\end{align*}
approaches $\min(1,\gamma_a)$%
. The growth of revenue thus flattens. 
The operational cost however continues to grow with the driver pool size at about the same rate. Consequently, profit starts to decline at some point, around 10,000 in this specific set-up. Furthermore, as the driver pool size increases, more competition leads to a higher ratio of vacant to hired vehicles shown in   %
Figure~\ref{fig:sen_fleetsize_vhratio}, although the number of matched vehicles increases, that is, more demand is fulfilled. Meanwhile,  the network has more vehicles circulating leading to a decline of average speed shown in Figure~\ref{fig:sen_fleetsize_speed}.

Figure~\ref{fig:sen_discountrate_profit} shows that as the discount rate increases, i.e., drivers become more short-sighted, the total profit decreases. This is because a positive reward occurs only when a match happens, some time into the future, and the more drivers ignore the future, the less profitable their routing choices become. Similarly, the demand fulfillment proportion decreases as shown in Figure~\ref{fig:sen_discountrate_fulfillment}. It is seemingly counterintuitive that the vacant-to-hired ratio decreases as shown in Figure~\ref{fig:sen_discountrate_vhratio}, which could be explained by the asymmetric effects myopia has on vacant and hired vehicles. A more myopic vehicle's routing is more like random walk. For a vacant vehicle, a random walk still has some good chance of leading to a match; for a hired vehicle, a random walk could make taking the passenger to the destination quite difficult. With more myopia, there will be more hired vehicles stuck on inefficient routes to their passengers' destinations. Finally, Figure~\ref{fig:sen_discountrate_speed} shows 
that the average speed remains almost unchanged, as the total number of vehicles circulating in the network remains constant.

\begin{figure}[!h]
    \begin{subfigure}[b]{0.45\textwidth}
        \centering
        \includegraphics[width=\linewidth]{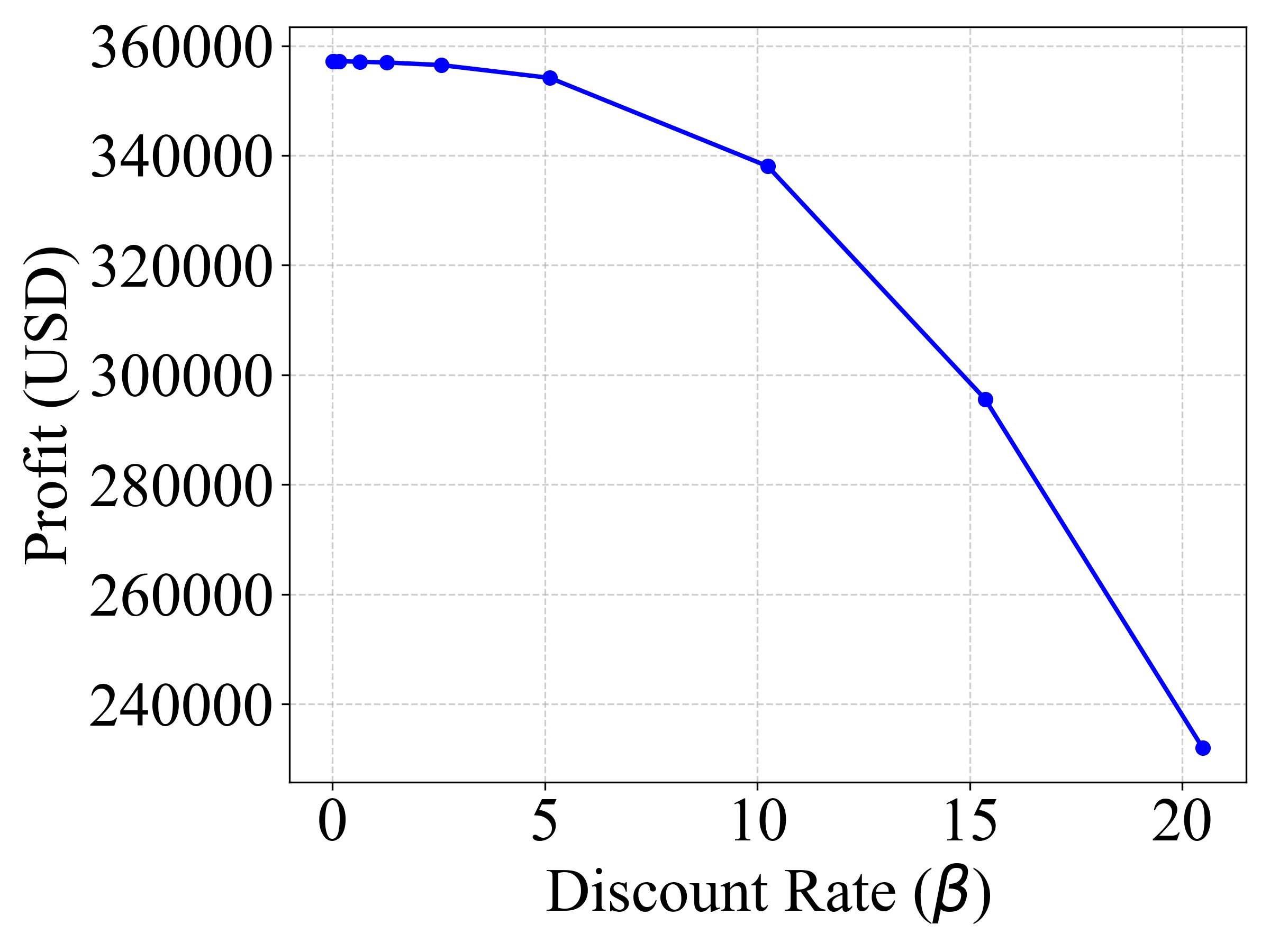}
        \caption{Profit}
        \label{fig:sen_discountrate_profit}
    \end{subfigure}
    \begin{subfigure}[b]{0.45\textwidth}
        \centering
        \includegraphics[width=\linewidth]{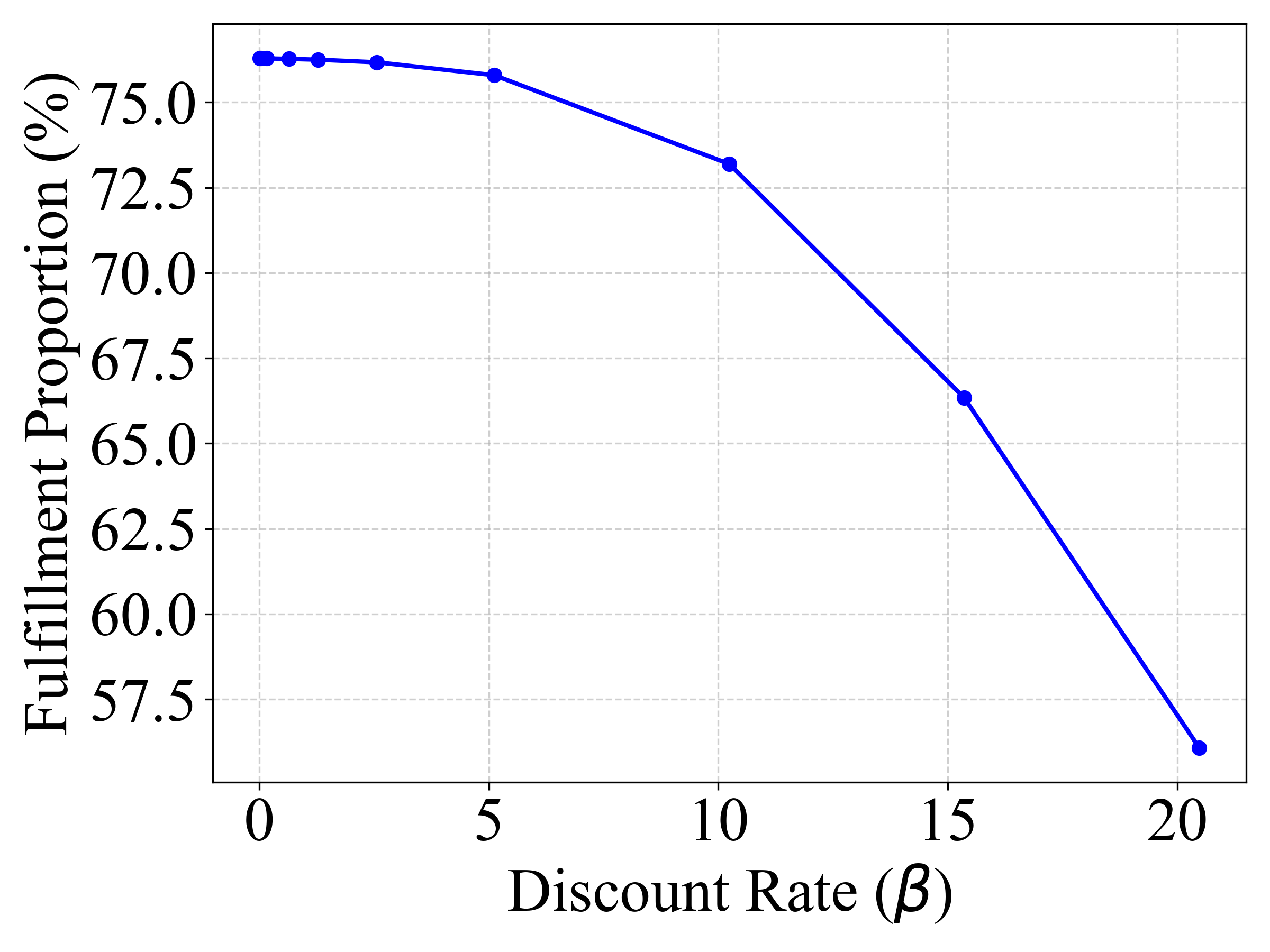}
        \caption{Fulfillment Proportion}
        \label{fig:sen_discountrate_fulfillment}
    \end{subfigure}
    \hfill
    \begin{subfigure}[b]{0.45\textwidth}
        \centering
        \includegraphics[width=\linewidth]{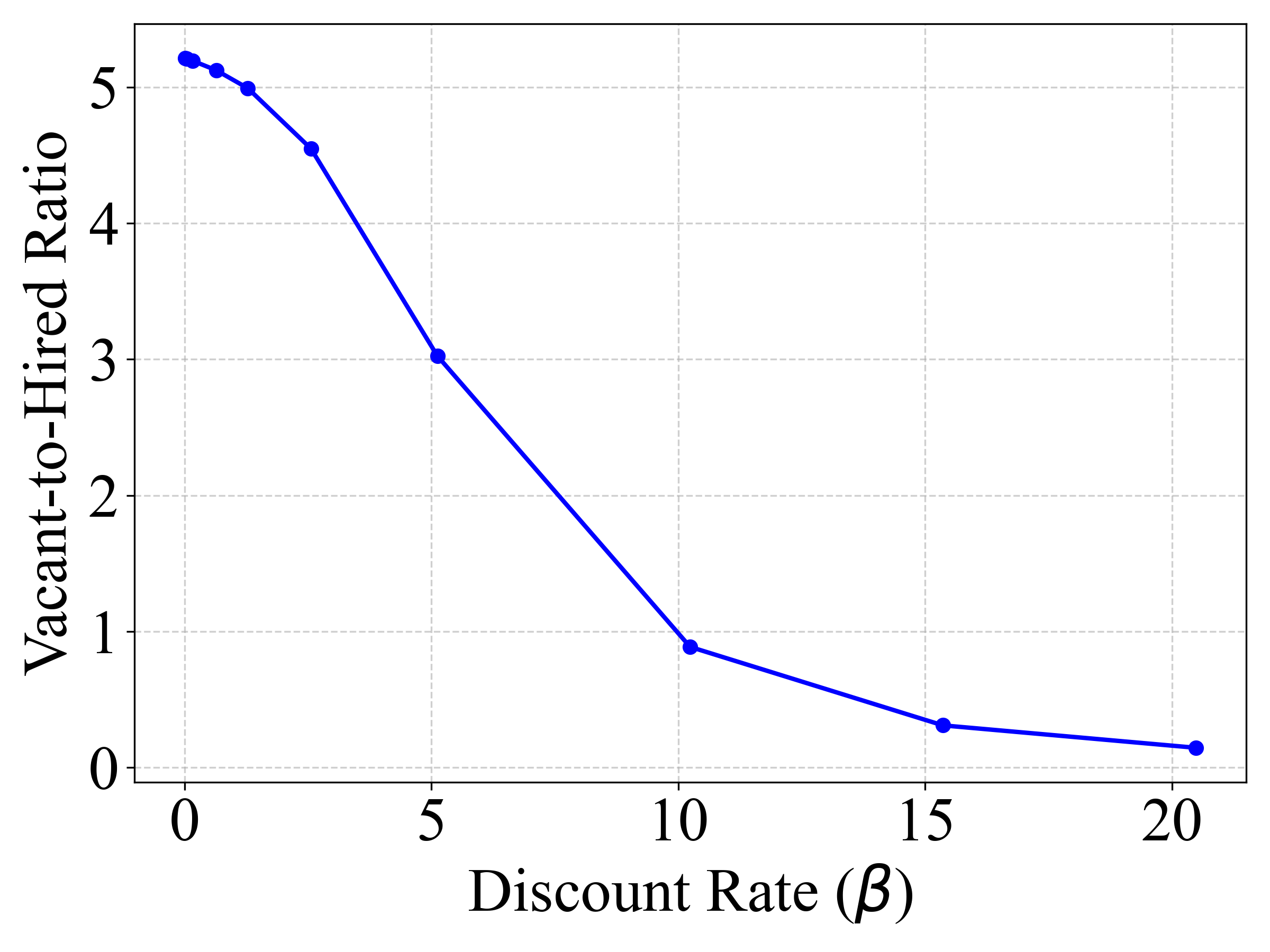}
        \caption{Vacant-to-hired Ratio}
        \label{fig:sen_discountrate_vhratio}
    \end{subfigure}
    \hfill
    \begin{subfigure}[b]{0.45\textwidth}
        \centering
        \includegraphics[width=\linewidth]{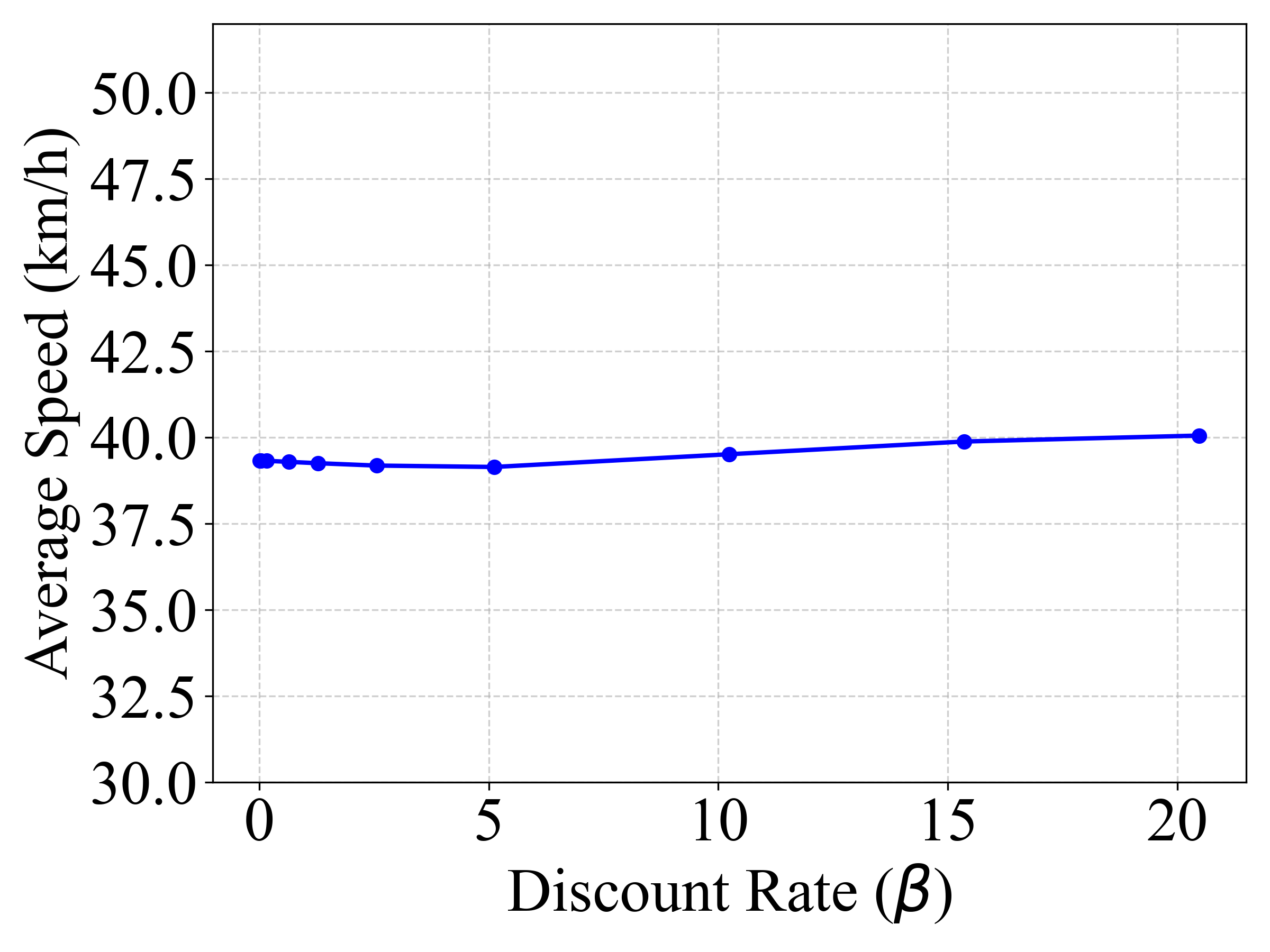}
        \caption{Average Speed}
        \label{fig:sen_discountrate_speed}
    \end{subfigure}
    \caption{Sensitivity analyses w.r.t. the discount rate on the Sioux Falls network.}
    \label{fig:sensitivity_analysis_siouxfalls_discountrate}
\end{figure}

Figure~\ref{fig:sensitivity_analysis_siouxfalls_matching} shows that as the matching friction parameter $\gamma$ increases, the profit and fulfillment ratio all rise and the vacant-to-hired ratio declines. This is intuitive as a larger $\gamma$ means less matching friction and thus a higher matching probability.  These metrics eventually converge to steady-state values, since once $\gamma$ becomes sufficiently large the number of matches is limited by the demand. 
Finally the average speed does not change much as the total number of vehicles remain constant. 

\begin{figure}[!h]
    \begin{subfigure}[b]{0.45\textwidth}
        \centering
        \includegraphics[width=\linewidth]{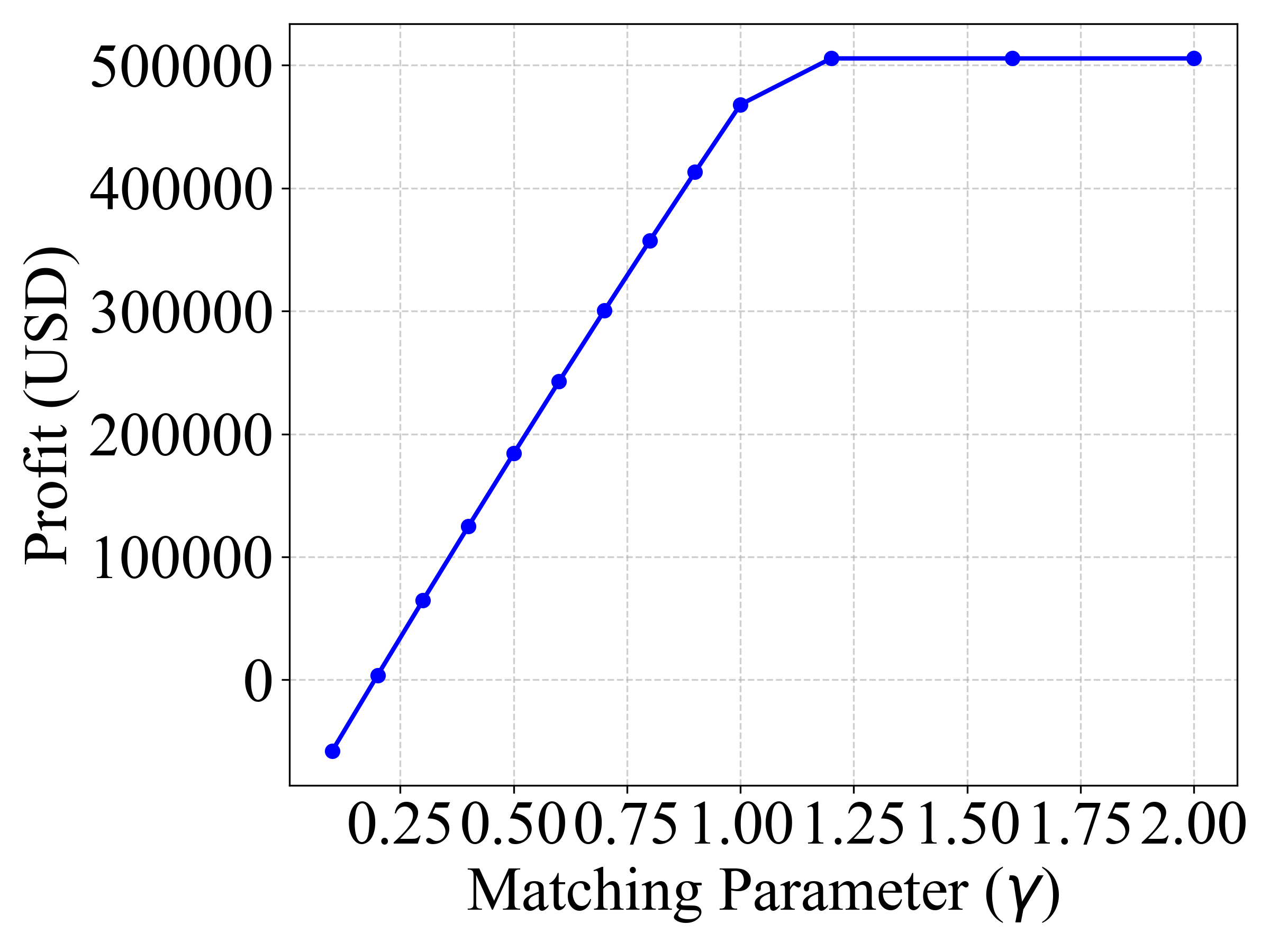}
        \caption{Profit}
        \label{fig:sen_matching_profit}
    \end{subfigure}
    \begin{subfigure}[b]{0.45\textwidth}
        \centering
        \includegraphics[width=\linewidth]{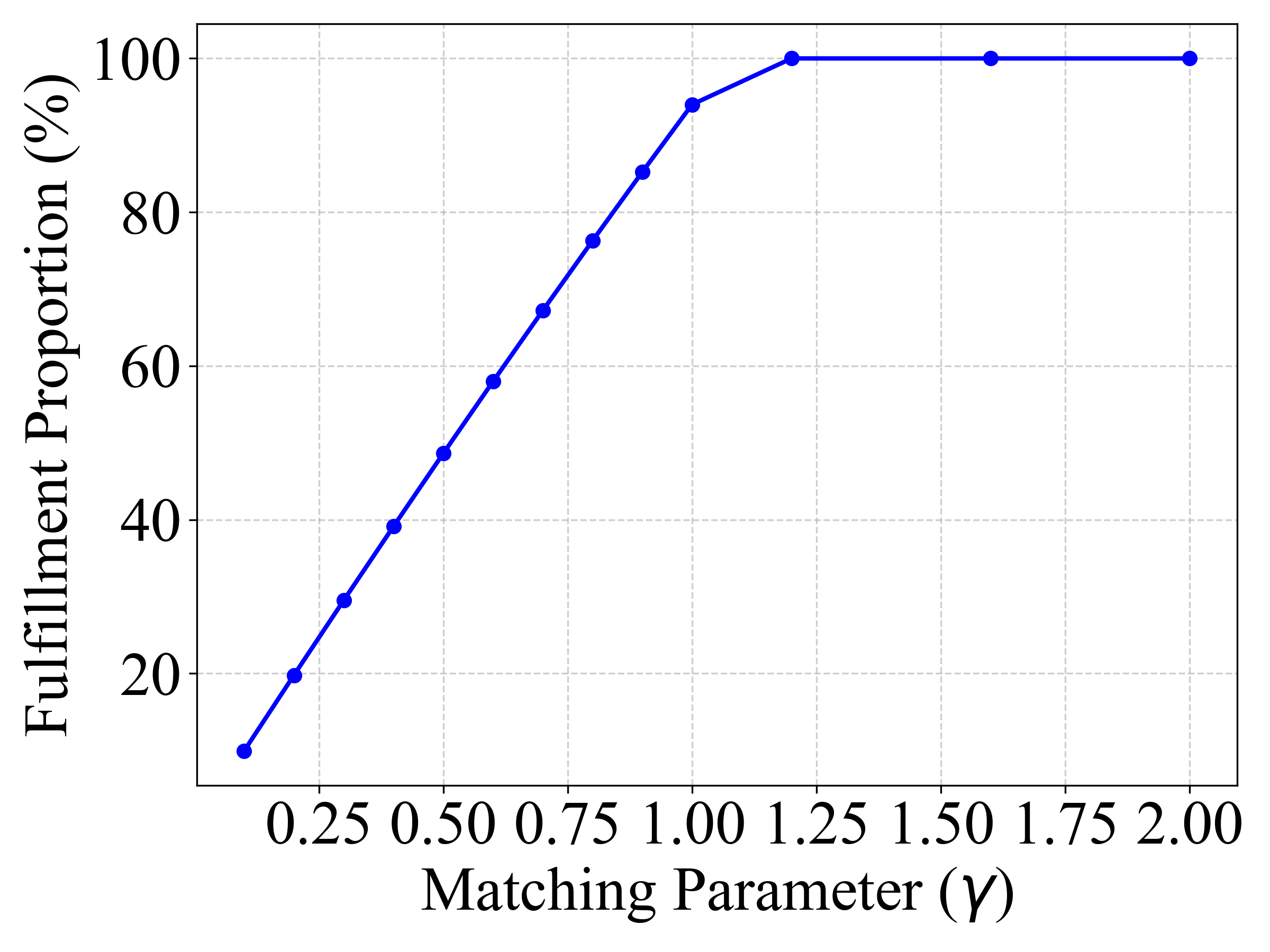}
        \caption{Fulfillment Proportion}
        \label{fig:sen_matching_fulfillment}
    \end{subfigure}
    \hfill
    \begin{subfigure}[b]{0.45\textwidth}
        \centering
        \includegraphics[width=\linewidth]{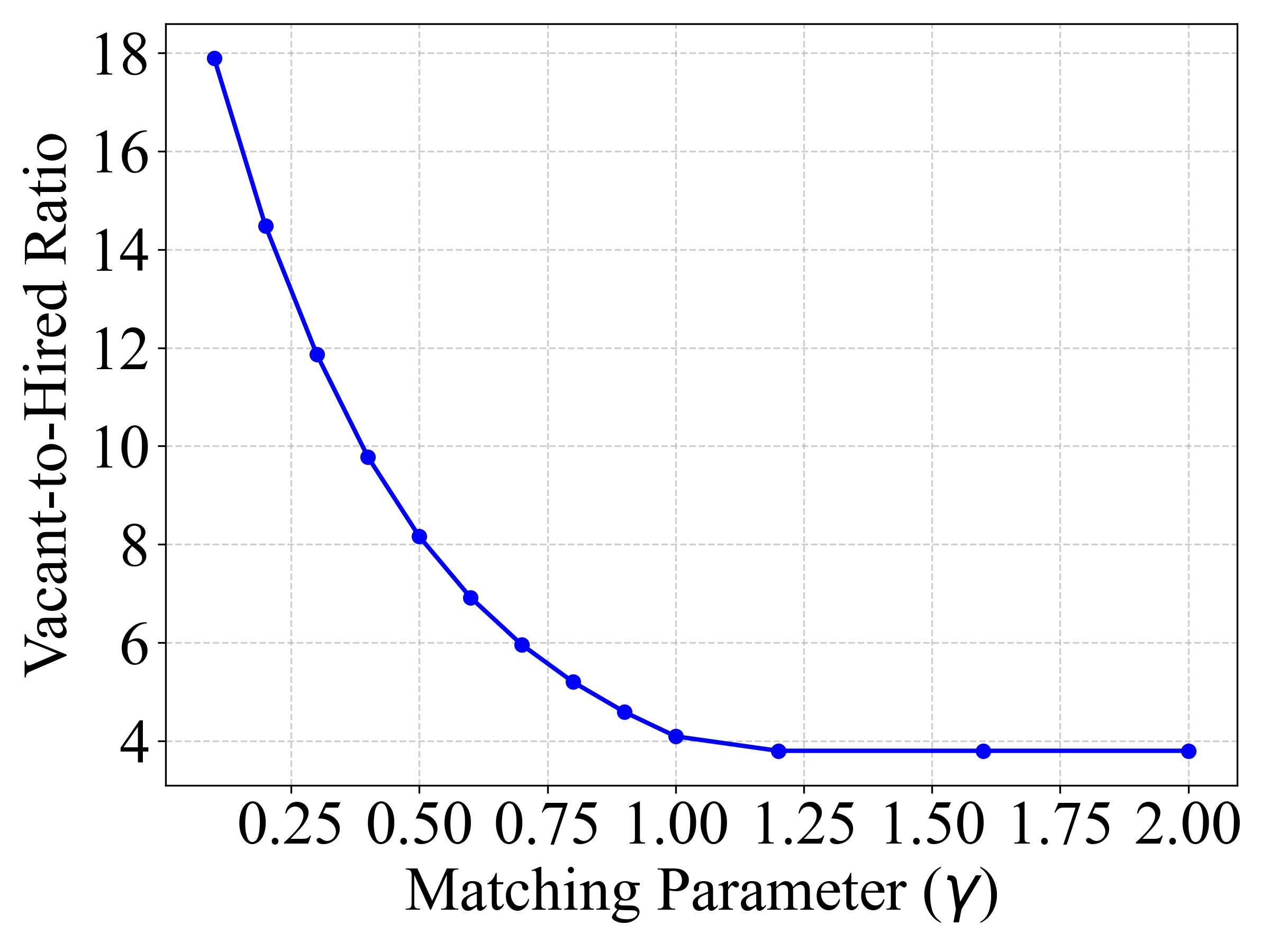}
        \caption{Vacant-to-hired Ratio}
        \label{fig:sen_matching_vhratio}
    \end{subfigure}
    \hfill
    \begin{subfigure}[b]{0.45\textwidth}
        \centering
        \includegraphics[width=\linewidth]{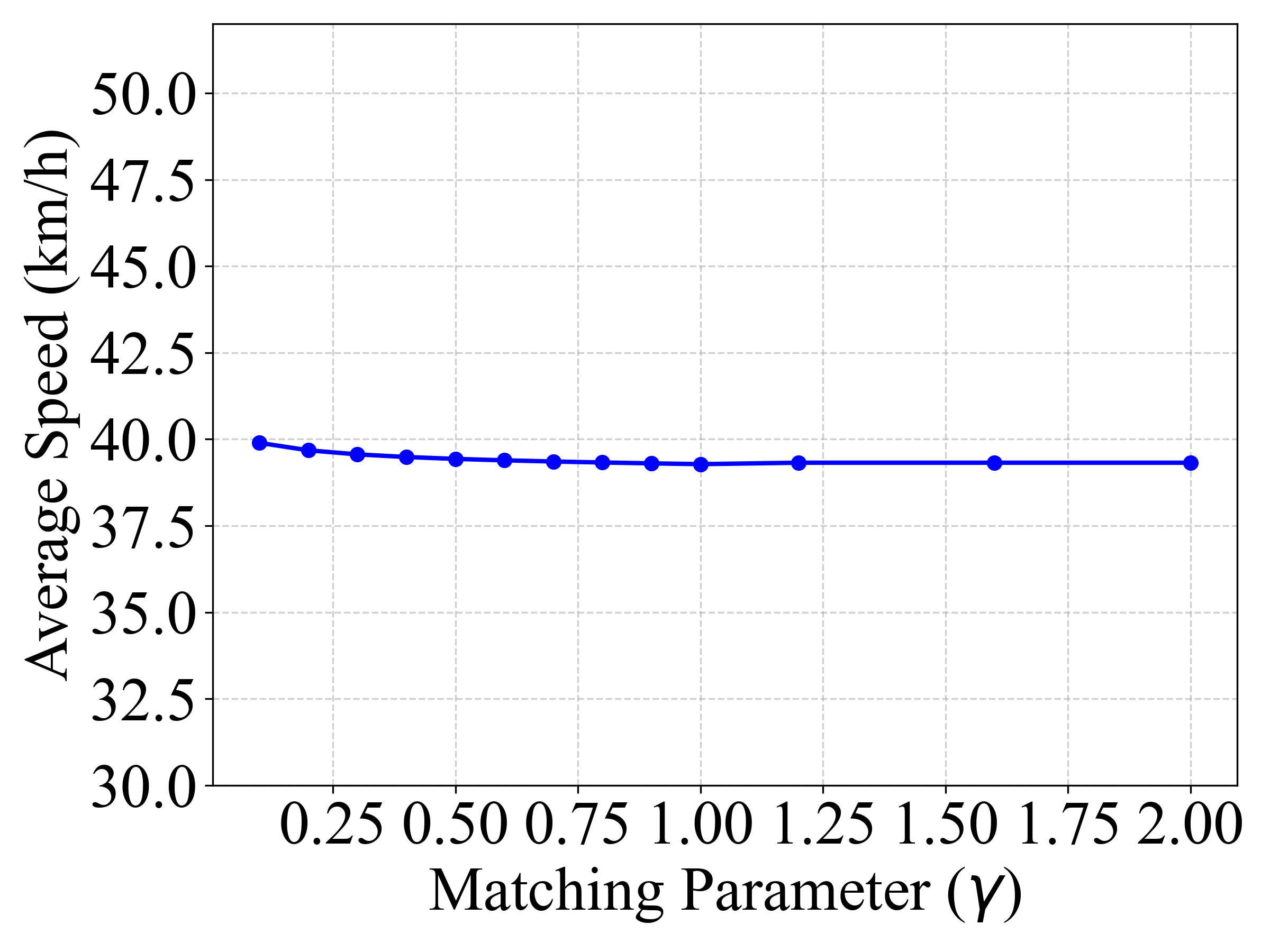}
        \caption{Average Speed}
        \label{fig:sen_matching_speed}
    \end{subfigure}
    \caption{Sensitivity analyses w.r.t. the matching parameter on the Sioux Falls network.}
    \label{fig:sensitivity_analysis_siouxfalls_matching}
\end{figure}

\subsection{Endogenous Driver Participation}
\label{subsec:end_driver}

We conduct numerical experiments of the extended MTER with endogenous driver participation on the Sioux Falls network to examine convergence behavior and how the participation rate varies with system parameters. The participation choice dispersion parameter $\zeta$ is set to 0.01. We assume that the potential driver pool is uniformly distributed across nodes. 

\paragraph{Convergence.}

On the Sioux Falls network with a pool size of 20{,}000 drivers, the MSA with a step size floor of 0.02 achieves a gap of below $10^{-4}$ after 889 iterations, with a total computational time of 1,074 seconds on a Mac laptop equipped with an Apple M4 chip. Figure~\ref{fig:extended_MSA_convergence} plots the gap across iterations and demonstrates a convergence pattern consistent with a linear rate of convergence.

\begin{figure}
    \centering
    \includegraphics[width=0.6\linewidth]{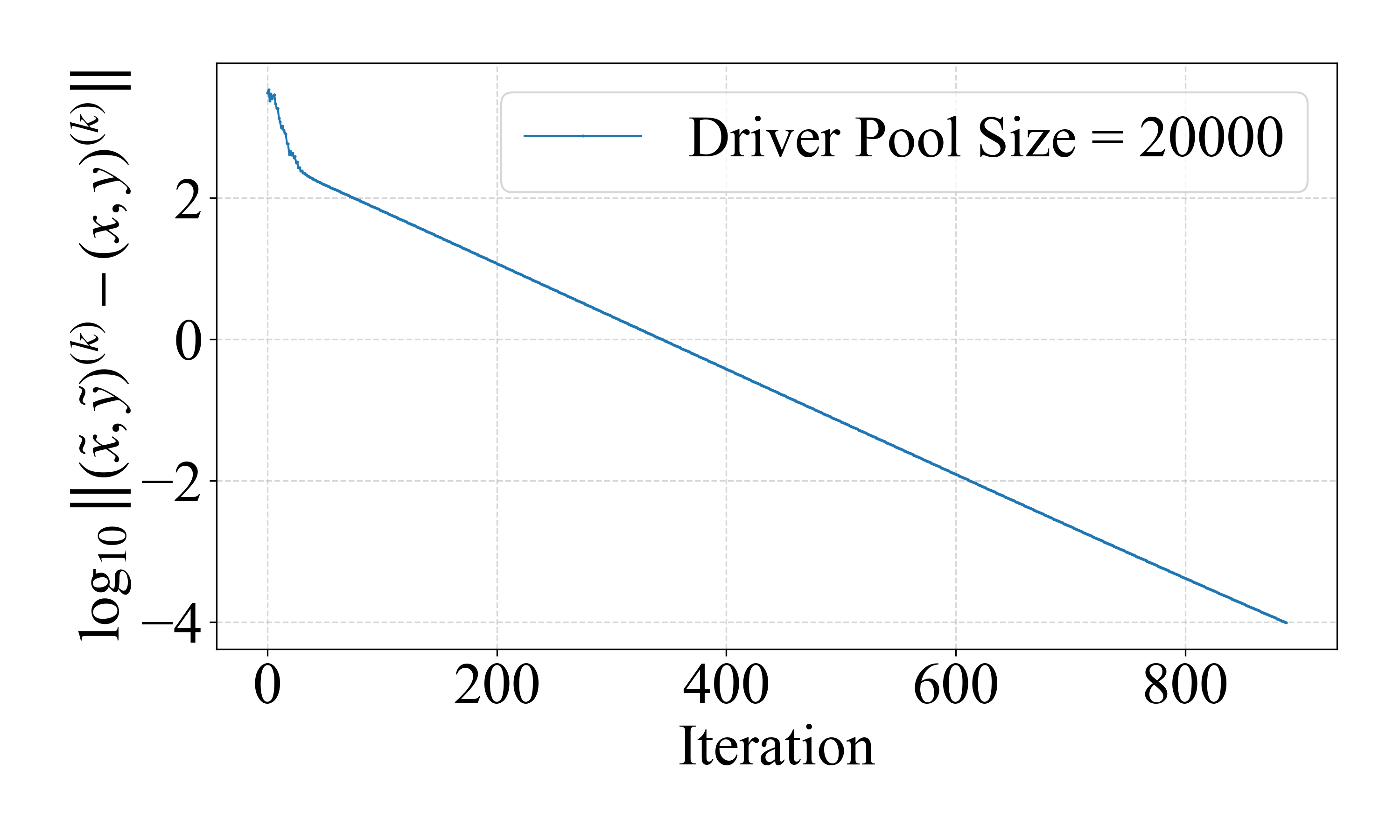}
    \caption{MSA convergence with endogenous driver participation on the Sioux Falls network.}
    \label{fig:extended_MSA_convergence}
\end{figure}

\smallskip

\paragraph{Vehicle participation rate.}

Driver participation rate is calculated as  $\frac{\sum_{i\in\mathcal{N}} M_i \mathbb P_i}{M}$, where $M = \sum_i M_i$ is the total potential driver pool size. Figure~\ref{fig:extension_utilization_fleet} shows that the participation rate decreases as the potential pool size $M$ increases. This trend occurs because the profit available to each vehicle diminishes as the total number of vehicles grows, as demonstrated in Figure~\ref{fig:extension_avgprofit_fleet}. Note that due to the probabilistic nature of the participation choice, some drivers still decide not to participate even when the participation profit is positive. 

In Figure~\ref{fig:extension_utilization_discount}, we observe that a higher discount rate $\beta$ leads to a lower participation rate for a given potential pool size. This effect is likely due to drivers perceiving an ever smaller gain from participation as the cost is generally more immediate than the reward. It is interesting that the real average profit increases at first in Figure~\ref{fig:extension_avgprofit_discount} highlighting the discrepancy between perception and reality: higher discounting of the future discourages participation leading to less competition, which actually makes participants better off. This effect however is short-lived, as myopic behaviors eventually make decisions so sub-optimal that the average profit suffers. 

In Figure~\ref{fig:extension_utilization_matching}, increasing the matching parameter $\gamma$ results in higher participation rates for a given potential driver pool size. This is because vehicles have a higher probability of being matched with passengers, which increases the expected value from participating in the service, as shown in Figure~\ref{fig:extension_avgprofit_matching}.

\begin{figure}[htbp]
    \centering
    \begin{subfigure}{0.45\textwidth}
        \centering
        \includegraphics[width=\linewidth]{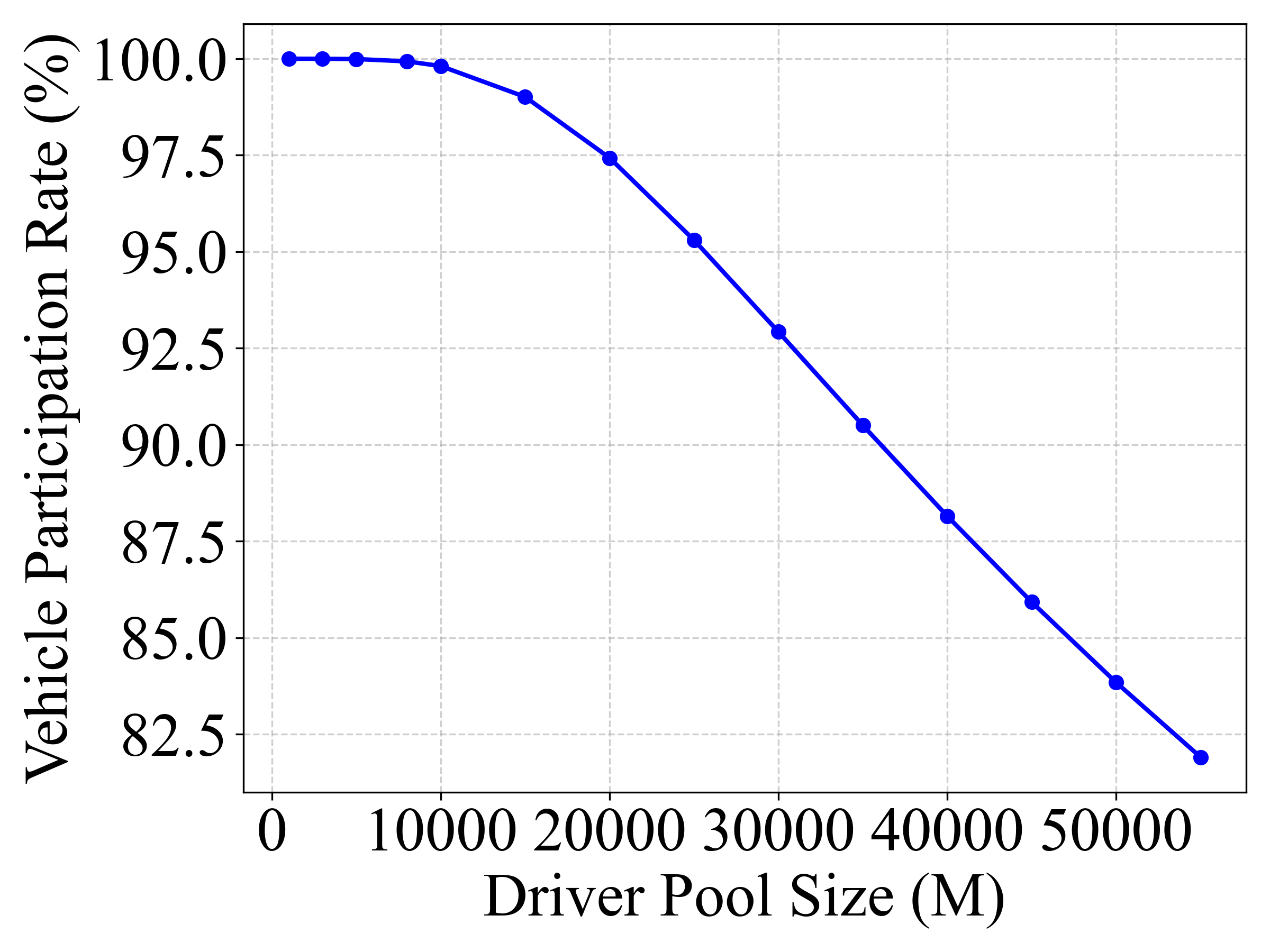}
        \caption{\footnotesize Driver participation vs. potential pool size $M$}
        \label{fig:extension_utilization_fleet}
    \end{subfigure}
    \hfill
    \begin{subfigure}{0.45\textwidth}
        \centering
        \includegraphics[width=\linewidth]{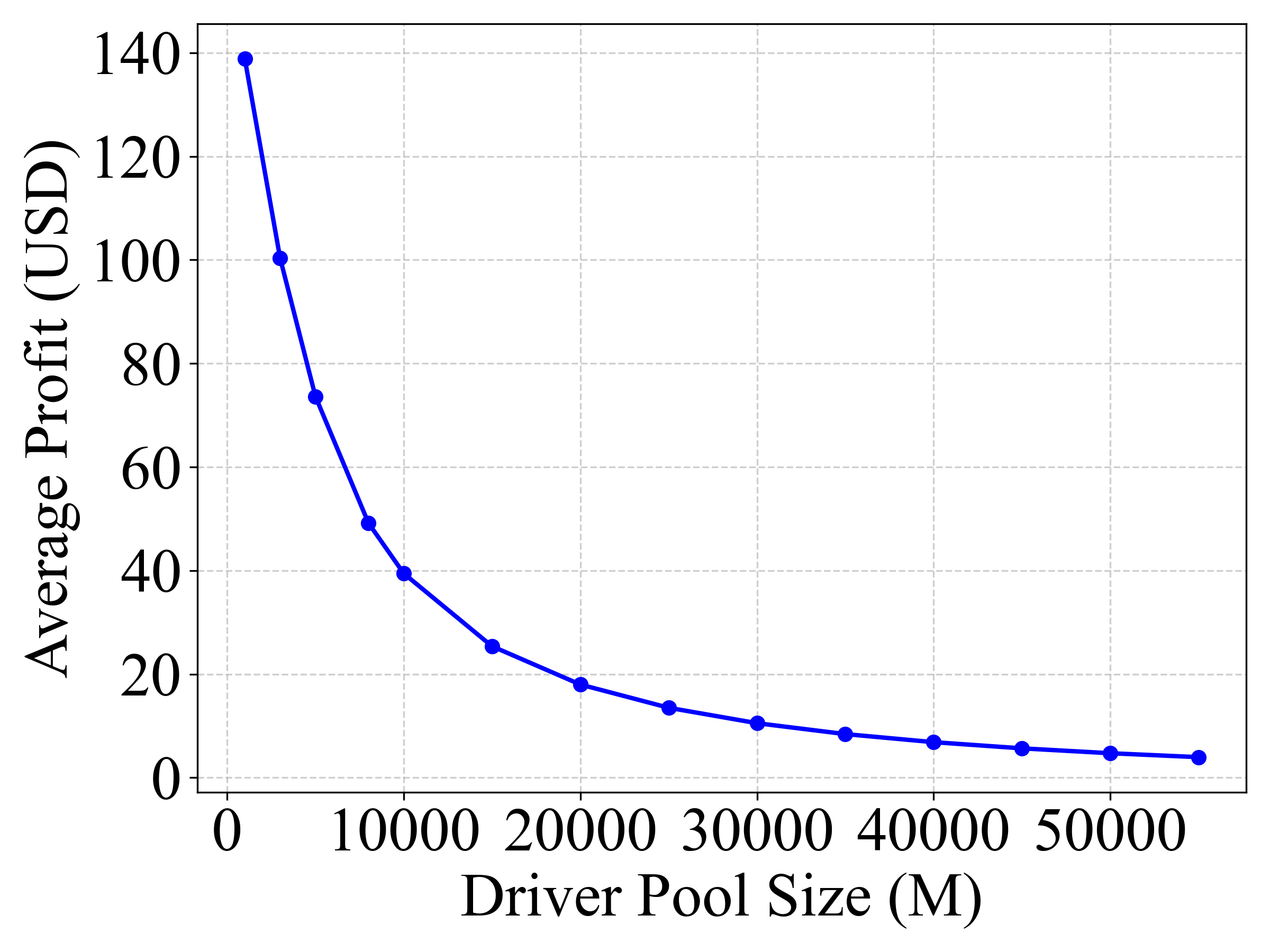}
        \caption{\footnotesize Average profit vs. potential pool size $M$}
        \label{fig:extension_avgprofit_fleet}
    \end{subfigure}

    \begin{subfigure}{0.45\textwidth}
        \centering
        \includegraphics[width=\linewidth]{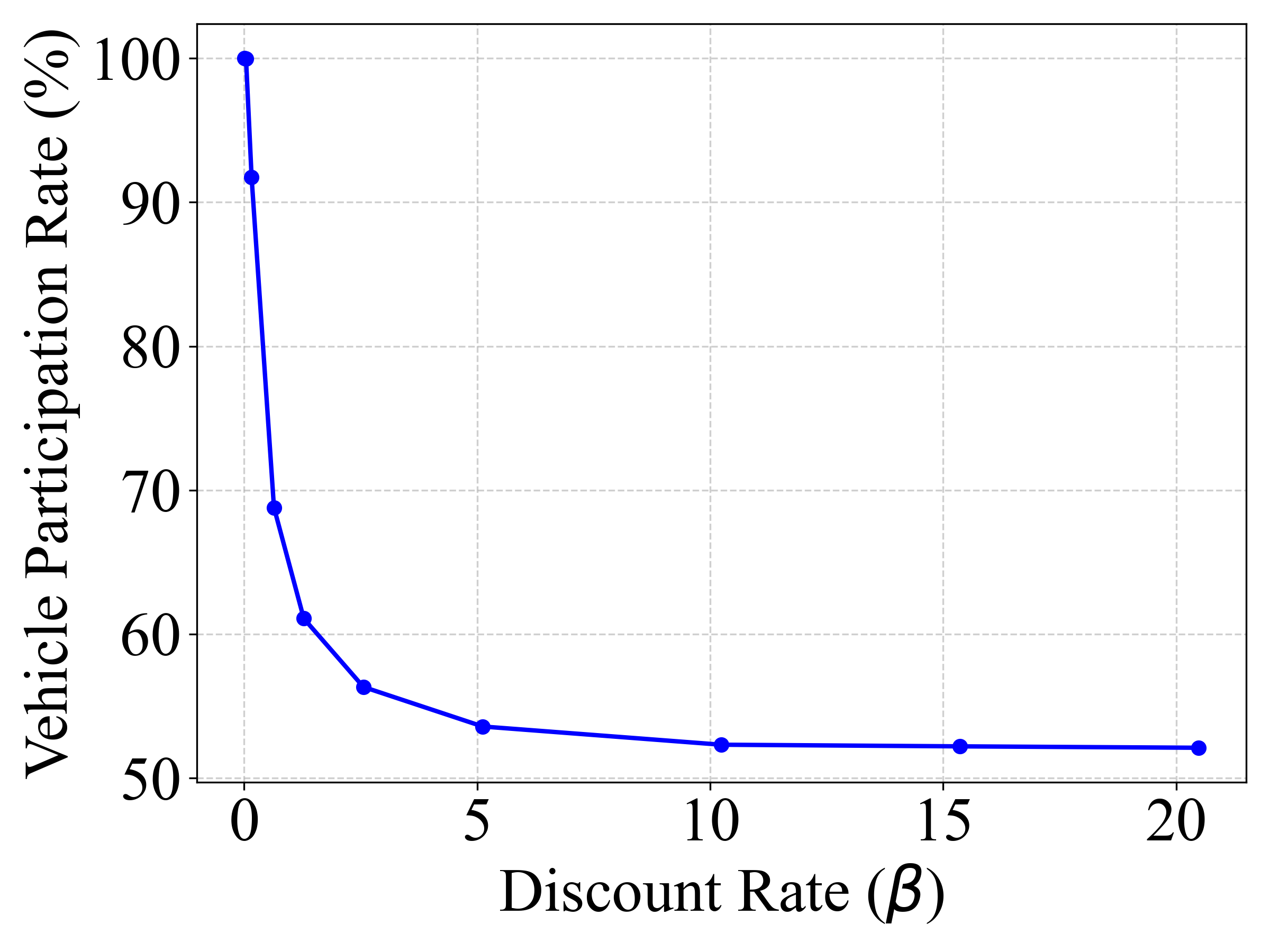}
        \caption{\footnotesize Driver participation vs. discount rate $\beta$}
        \label{fig:extension_utilization_discount}
    \end{subfigure}
    \hfill
    \begin{subfigure}{0.45\textwidth}
        \centering
        \includegraphics[width=\linewidth]{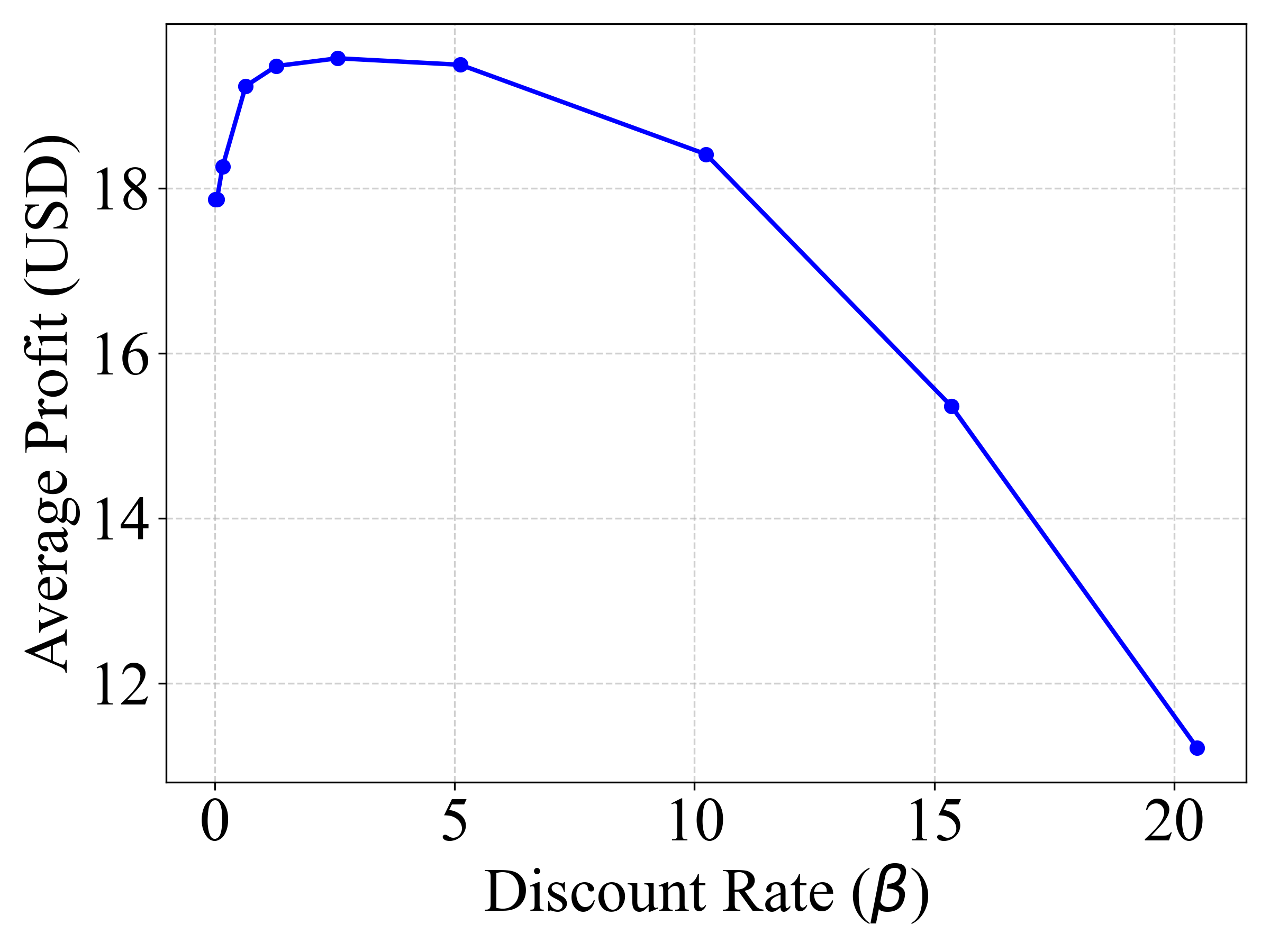}
        \caption{\footnotesize Average profit vs. discount rate $\beta$}
        \label{fig:extension_avgprofit_discount}
    \end{subfigure}

    \begin{subfigure}{0.45\textwidth}
        \centering
        \includegraphics[width=\linewidth]{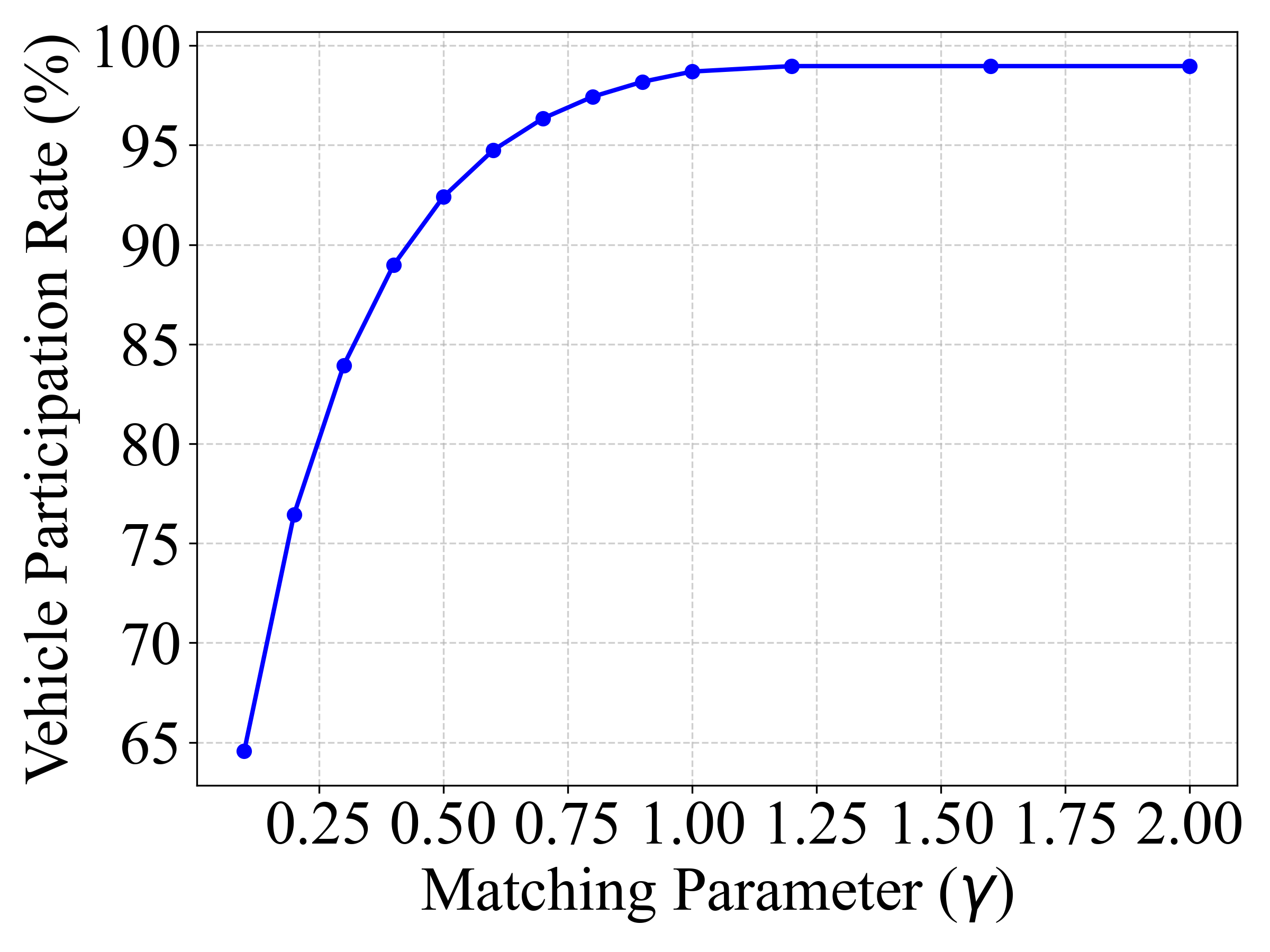}
        \caption{\footnotesize Driver participation vs. matching parameter $\gamma$}
        \label{fig:extension_utilization_matching}
    \end{subfigure}
    \hfill
    \begin{subfigure}{0.45\textwidth}
        \centering
        \includegraphics[width=\linewidth]{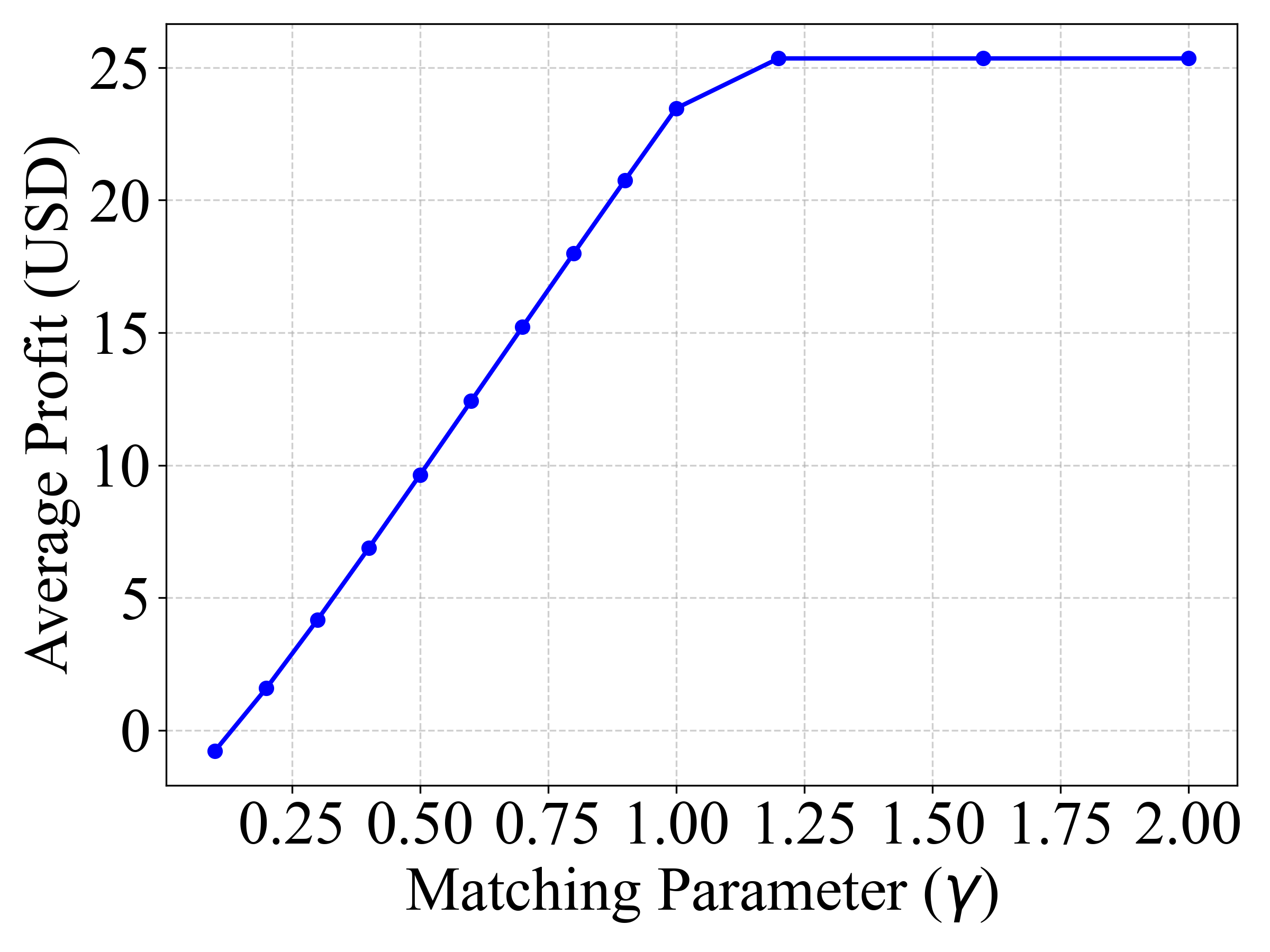}
        \caption{\footnotesize Average profit vs. matching parameter $\gamma$}
        \label{fig:extension_avgprofit_matching}
    \end{subfigure}

    \caption{The driver participation rate and average profit under varying potential driver pool sizes, discount rates and matching parameters.}
    \label{fig:combined_sensitivity_analysis}
\end{figure}

\end{APPENDICES}

\end{document}